\documentclass[english,11pt]{article}
\usepackage[T1]{fontenc}
\usepackage[latin9]{inputenc}
\usepackage{amsmath}
\usepackage{hyperref}
\usepackage{amsthm}
\usepackage{amssymb}
\usepackage{ifthen}

\usepackage{thmtools} 
\usepackage{thm-restate}

\usepackage{fullpage}
\usepackage{graphicx}
\usepackage{tikz}
\usepackage{bm}
\usepackage{times}

\usetikzlibrary{patterns}
\usetikzlibrary{decorations.pathreplacing}
\usetikzlibrary{arrows.meta}

\definecolor{myred}{RGB}{163, 51, 61}
\definecolor{myblue}{RGB}{90, 107, 127}
\definecolor{mygreen}{RGB}{133, 177, 168}

\usepackage{verbatim}

\makeatletter
\theoremstyle{plain}
\newtheorem{thm}{\protect\theoremname}
\theoremstyle{remark}

\theoremstyle{plain}
\declaretheorem[name=Lemma,sibling=thm]{lem}
\ifx\proof\undefined
\newenvironment{proof}[1][\protect\proofname]{\par
	\normalfont\topsep6\p@\@plus6\p@\relax
	\trivlist
	\itemindent\parindent
	\item[\hskip\labelsep\scshape #1]\ignorespaces
}{%
	\endtrivlist\@endpefalse
}
\providecommand{\proofname}{Proof}
\fi
\theoremstyle{definition}
\newtheorem{defn}[thm]{\protect\definitionname}
\theoremstyle{plain}
\newtheorem{prop}[thm]{\protect\propositionname}

\makeatother

\usepackage{xspace}
\usepackage[textsize=tiny]{todonotes}

\newcommand{\red}[1]{\textcolor{red}{#1\xspace}}

\ifdefined\DEBUG

\newcommand{\lr}[1]{\cyan{#1}}
\newcommand{\aw}[1]{\blue{#1}}
\newcommand{\lrtodo}[1]{\todo[fancyline, color=green!40]{{LR: #1}}}
\newcommand{\awtodo}[1]{\todo[fancyline]{AW: #1}}

\else

\newcommand{\lr}[1]{#1}
\newcommand{\aw}[1]{#1}
\newcommand{\lrtodo}[1]{}
\newcommand{\awtodo}[1]{}

\fi

\renewcommand{\epsilon}{\varepsilon}

\usepackage{babel}
\providecommand{\claimname}{Claim}
\providecommand{\definitionname}{Definition}

\providecommand{\theoremname}{Theorem}
\providecommand{\propositionname}{Proposition}

\begin{document}
\global\long\def\P{\mathcal{P}}%
\global\long\def\OPT{\mathrm{OPT}}%
\global\long\def\E{\mathcal{E}}%
\global\long\def\S{\mathcal{S}}%
\global\long\def\Seg{\mathrm{Seg}}%
\global\long\def\N{\mathbb{N}}%
\global\long\def\T{\mathcal{T}}%
\global\long\def\red{\mathrm{red}}%
\global\long\def\opt{\mathrm{opt}}%
\global\long\def\dist{\mathrm{dist}}%
\global\long\def\crit{\mathrm{crit}}%
\global\long\def\next{\mathrm{next}}%
\global\long\def\Rel{\mathrm{Rel}}%
\global\long\def\C{\mathcal{C}}%
\global\long\def\Q{\mathcal{Q}}%
\global\long\def\I{\mathcal{I}}%
\global\long\def\l{\mathrm{large}}%
\global\long\def\Orel{\mathrm{off}_{x}}%
\global\long\def\Oproc{\mathrm{off}_{y}}%
\global\long\def\R{\mathcal{R}}%
\global\long\def\Rs{\mathcal{R}_{\mathrm{small}}}%
\global\long\def\Rl{\mathcal{R}_{\mathrm{large}}}%
\global\long\def\sm{\mathrm{small}}%
\global\long\def\l{\mathrm{large}}%
\global\long\def\Jh{J_{\mathrm{high}}}%
\global\long\def\Je{J_{\mathrm{easy},g}}%
\global\long\def\Jh{J_{\mathrm{hard},g}}%
\global\long\def\len{\mathrm{len}}%
\global\long\def\cellb{\mathrm{beg}}%
\global\long\def\celle{\mathrm{end}}%
\global\long\def\Rir{\R^{\mathrm{ir}}}%
\global\long\def\G{\mathcal{G}}%
\global\long\def\r{\mathrm{round}}%

\input{fig-geometric}

\date{}
\pagenumbering{gobble}

\title{A $(2+\epsilon)$-approximation algorithm for preemptive weighted
flow time on a single machine}
\author{Lars Rohwedder\footnote{EPFL, Switzerland, \href{mailto:lars.rohwedder@epfl.ch}{lars.rohwedder@epfl.ch},
  supported by the Swiss National Science Foundation project 200021-184656}
\and Andreas Wiese\footnote{Universidad de Chile, Chile, \href{mailto:awiese@dii.uchile.cl}{awiese@dii.uchile.cl}, partially
supported by the ANID Fondecyt Regular grant 1200173.}
}
\maketitle

\begin{abstract}
Weighted flow time is a fundamental and very well-studied objective
function in scheduling. In this paper, we study the setting of a single
machine with preemptions. The input consists of a set of jobs,
characterized by their processing times, release times, and weights
and we want to compute a (possibly preemptive) schedule for them.
The objective is to minimize the sum of the weighted flow times of
the jobs, where the flow time of a job is the time between its release
date and its completion time.

It had been a long-standing open problem to find a polynomial time
$O(1)$-approximation algorithm for this setting. In a recent break-through
result, Batra, Garg, and Kumar (FOCS 2018) found such an algorithm
if the input data are polynomially bounded integers, and Feige, Kulkarni,
and Li (SODA 2019) presented a black-box reduction to this setting.
The resulting approximation ratio is a (not explicitly stated) constant
which is at least $10.000$. In this paper we improve this ratio to
$2+\epsilon$. The algorithm by Batra, Garg, and Kumar (FOCS
2018) reduces the problem to \textsc{Demand MultiCut on trees} and
solves the resulting instances via LP-rounding and a dynamic program.
Instead, we first reduce the problem to a (different) geometric problem
while losing only a factor $1+\epsilon$, and then solve its resulting
instances up to a factor of $2+\epsilon$ by a dynamic program.
In particular, our reduction ensures certain structural properties,
thanks to which we do not need LP-rounding methods.

We believe that our result makes substantial progress towards
finding a PTAS for weighted flow time on a single machine.
\end{abstract}
\newpage\pagenumbering{arabic}

\setcounter{page}{1}

\section{Introduction}

Weighted flow time is a fundamental and well studied objective in
the scheduling literature, e.g., ~\cite{DBLP:conf/focs/Batra0K18,DBLP:conf/soda/FeigeKL19,DBLP:conf/latin/BansalP04,DBLP:conf/stoc/BansalP03,DBLP:journals/siamcomp/BansalP14,DBLP:journals/siamcomp/KellererTW99,azar2018improved,DBLP:conf/stoc/ChekuriK02,chekuri2001algorithms}.
We are given
a set of jobs $J$ where each job $j\in J$ is characterized by a
release time $r_{j}\in\N$, a processing time $p_{j}\in\N$, and a
weight $w_{j}\in\N$. In a computed schedule, the flowtime $F_{j}$
of a job $j$ is the difference between its completion time and its
release date $r_{j}$. The goal is to minimize $\sum_{j\in J}w_{j}F_{j}$.

In this paper, we study the setting of a single machines in which
we allow to preempt jobs (and resume them later). Note that without
preemptions the problem cannot even be approximated with a factor
of $O(n^{1/2-\epsilon})$ for any $\epsilon>0$~\cite{DBLP:journals/siamcomp/KellererTW99}.
It is known by the work of Chekuri and Khanna~\cite{DBLP:conf/stoc/ChekuriK02}
that for every $\epsilon>0$ there is a $(1+\epsilon)$-approximation
in quasi-polynomial time (QPTAS), assuming quasi-polynomially bounded
input data. In contrast to this, it had been a long-standing
important open problem whether a constant factor approximation can
be computed in \emph{polynomial} time~\cite{schuurman1999polynomial}.
In a breakthrough result, Batra,
Garg, and Kumar~\cite{DBLP:conf/focs/Batra0K18} presented such an
algorithm with pseudopolynomial running time. While for many
scheduling problems one can assume the input data to be polynomially
bounded via straight-forward rounding of the input etc., this is
not the case for weighted flow time. However, Feige, Kulkarni, and
Li~\cite{DBLP:conf/soda/FeigeKL19} gave a \lr{non-trivial} black-box reduction to
this setting which completely settles the mentioned long-standing
open question (and also yields a QPTAS for arbitrary input data).

The algorithm in~\cite{DBLP:conf/focs/Batra0K18} first reduces a
given problem instance to a clean graph problem, the \textsc{Demand
MultiCut problem on trees}. This reduction loses a factor of $32$
in the approximation ratio. Then, the authors present an approximation
algorithm for the resulting instance of \textsc{Demand MultiCut}.
To this end, they split it into two subinstances and solve the first
one by rounding a linear program (LP) and the second one with a dynamic
program (DP). Their approximation ratio for the first subinstance
is $24+8\beta$ where $\beta=O(1)$ is the approximation ratio of
an algorithm by Chan, Grant, K{\"o}nemann, and Sharpe~\cite{DBLP:conf/soda/ChanGKS12}
(which is invoked as a subroutine); the constant $\beta$ is not explicitly
stated in~\cite{DBLP:conf/soda/ChanGKS12}. The DP for the second
subinstance crucially exploits the hierarchical structure given by
the tree. Its approximation ratio is a constant which is not explicitly
stated in~\cite{DBLP:conf/focs/Batra0K18}, but is at least $512$.
Hence, the overall approximation ratio is at least $32\cdot(536+8\beta)\ge10.000$.
While one could try to optimize this constant, it is not
clear how to avoid to lose substantial factors in several parts of
the algorithm, e.g., the factor $32$ in the reduction to \textsc{Demand
MultiCut}, further constant factors when solving the two subinstances
mentioned above, and also the dependence on $\beta$.

\subsection{Our contribution}

In this paper, we present a polynomial time $(2+\epsilon)$-approximation
algorithm for weighted flow time on a single machine. We first reduce
the problem to a geometric problem (rather than \textsc{Demand MultiCut}).
Then we solve the resulting instance of this problem by a dynamic
program. Our reduction is almost loss-less, i.e., it loses only a
factor of $1+\epsilon$, and our DP has an approximation ratio
of only $2+\epsilon$ which leads to an approximation ratio of
$2+\epsilon$ overall.

In our geometric problem, the input consists of a set of non-overlapping
axis-parallel rectangles of unit height and a set of rays that are
all vertical and oriented downwards, see Figure~\ref{fig:geometric-problem}.
Each rectangle has a cost and a capacity, each ray has a demand. The
goal is to select rectangles of minimium total cost such that for
each ray, the total capacity of the selected rectangles intersecting
it is at least the demand of the ray. For technical reasons there
are some local dependencies between rectangles, that is, some rectangles
can only be selected when another rectangle of the same
size \lr{directly to its left} is selected as well.

In the instances obtained by our reduction, the rectangles are arranged
in a hierarchical structure given by a hierarchical decomposition
of the $x$-axis. More precisely, the projection of each rectangle
to the $x$-axis concides with a cell of this hierarchical decomposition.
Moreover, when we traverse each ray from its respective initial point
on, the widths of the rectangles hit by the ray are monotone (non-increasing).
This hierarchical structure is crucial for our dynamic program (similarly
to the tree-structure in~\cite{DBLP:conf/focs/Batra0K18}). In particular,
we manage to obtain this important structure while losing only a factor
of $1+\epsilon$ in the reduction.

\begin{figure}
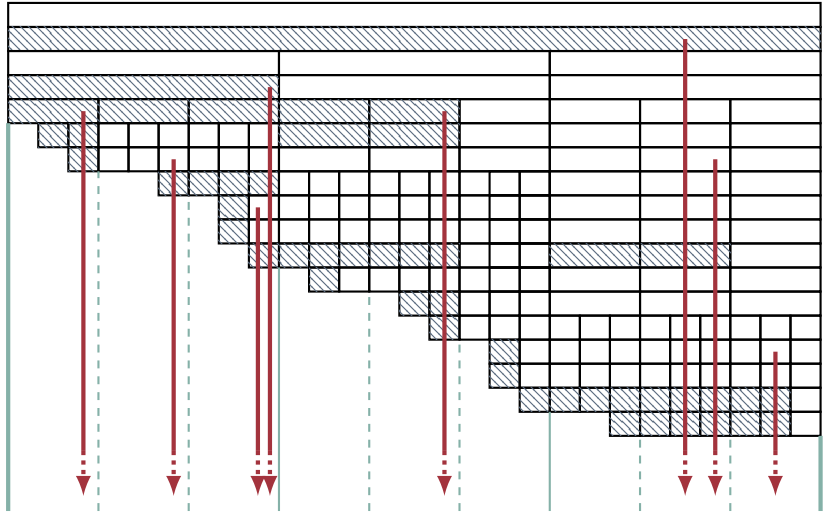

\centering
\figgeometric{}
\caption{An instance of the geometric problem to which we reduce weighted flow
time. The rays are depicted in red and the hierachical decomponsition is visualized in green.
The hatched rectangles form an example solution.
The capacities and costs of the rectangles and the demands
of the rays are not depicted, and neither the mentioned local dependencies
between some adjacent rectangles of the same sizes.}
\label{fig:geometric-problem} 
\end{figure}

Importantly, in contrast to \cite{DBLP:conf/focs/Batra0K18}
we can solve our instances of this geometric problem by dynamic programming
only, and do not require the LP rounding algorithm from~\cite{DBLP:conf/soda/ChanGKS12}
or a similar procedure (with additional constant factor losses). The
intuitive reason is that, translated to our geometric visualization,
the instances of \textsc{Demand MultiCut} described in~\cite{DBLP:conf/focs/Batra0K18}
introduce vertical line segments, rather than rays, and the algorithm
in \cite{DBLP:conf/focs/Batra0K18} needs LP-rounding for a certain
type of (intuitively short) line segments, which we can completely
avoid. In our DP, we translate some ideas from \cite{DBLP:conf/focs/Batra0K18}
to our geometric problem. However, our routine is significantly more
involved than the DP in \cite{DBLP:conf/focs/Batra0K18} due to the
higher complexity of our geometric problem (compared to \textsc{Demand
MultiCut on trees}), and since it is designed to optimize the
approximation ratio of $2+\epsilon$ incurred by it.

Our algorithm has pseudo-polynomial running time. With the black-box
reduction in~\cite[Section 4]{DBLP:conf/soda/FeigeKL19} we turn
it into a polynomial time algorithm, while losing only a factor of
$1+\epsilon$.
\begin{thm}
There is a polynomial time $(2+\epsilon)$-approximation algorithm
for the problem of minimizing weighted flow time on a single machine
in the preemptive setting. 
\end{thm}

We believe that our result is a crucial step forward in the search
of a PTAS for weighted flow time on a single machine. In particular,
a possible approach for constructing a PTAS could be to use our reduction
to the geometric problem above and develop a $(1+\epsilon)$-approximation
algorithm for the resulting instances.

\subsection{Other related work}

Prior to the results in~\cite{DBLP:conf/focs/Batra0K18,DBLP:conf/soda/FeigeKL19},
Bansal and Pruhs~\cite{DBLP:journals/siamcomp/BansalP14} presented
a $O(\log\log P)$-approximation algorithm for weighted flow time
(we denote by $P$ and $W$ the ratios between the largest and smallest
processing times and weights in the input, respectively), and even
more general for the General Scheduling problem in which each job
incurs a cost, depending on its completion time, and this cost is
given by a job-dependent cost function. They reduce this problem to
a geometric covering problem (which, however, is substantially
different from the geometric problem that we reduce to). For the
special cases where $w_{j}=1/p_{j}$ for each job $j$ (i.e. the stretch
metric) or if $P=O(1)$ there is a PTAS known~\cite{DBLP:conf/stoc/ChekuriK02,DBLP:journals/scheduling/BenderMR04}.
The best complexity result for weighted flow time on a single machine
with preemption is strong NP-hardness~\cite{lenstra1977complexity},
which leaves open whether a PTAS exists for the problem.

Weighted flow time has been studied in the online setting. Bansal
and Dhamdhere~\cite{bansal2007minimizing} presented a $O(\log W)$-competitive
algorithm and a semi-online $O(\log nP)$-competitive algorithm. Also,
Chekuri, Khanna, and Zhu~\cite{chekuri2001algorithms} gave a semi-online
$O(\log^{2}P)$-approximation algorithm. These results were improved
by Azar and Touitou~\cite{azar2018improved} who gave a $\min(\log W,\log P,\log D)$-competitive
algorithm, where $D$ is the ratios of the largest and smallest job
densities, being defined as $w_{j}/p_{j}$ for each job $j$. On the
other hand, there can be no online $O(1)$-competitive algorithm,
due to a result by Bansal and Chan~\cite{bansal2009weighted}. However,
if the online algorithm is given machines of speed $1+\epsilon$ then
$O(1)$-competitive algorithms exist, as shown by Bansal and Pruhs~\cite{DBLP:conf/stoc/BansalP03,DBLP:conf/latin/BansalP04}.


\section{Reduction to geometric problem}
\label{sec:geo} We start with some standard transformations to simplify
the instance of weighted flow time. We assume w.l.o.g.\ that $1/\epsilon\in\N$
and $\min_{j}r_{j}=0$. Moreover, we can assume that $\max_{j}r_{j}\le\sum_{j}p_{j}$,
since otherwise we can split the given instance into independent subinstances.
Recall that $P$ is defined as the ratio $\max_j p_j / \min_j p_j$.
\aw{By scaling the input values and rounding,}
we can also assume that $\min_j p_j =1$,
$\max_{j}p_{j}=P$,\awtodo{TODO: prove in appendix} and $1\le w_{j}\le O_{\epsilon}(n^{2}P)$ for
each job $j$, while losing only a factor of $1+\epsilon$ in the approximation ratio \aw{and increasing $P$ by only polynomial factors} (see Appendix~\ref{subsec:Bounded-weights} for details).
We define $T:=\max_{j}r_{j}+\sum_{j}p_{j}\le2nP.$ Hence, we can assume
w.l.o.g.~that each job finishes within $[0,T)$.

Then the problem is modeled by the following integer program
that we denote by (IP). Intuitively, for each job $j\in J$ and each
time $t\in\N$, we introduce a variable $x_{j,t}$ such that $x_{j,t}=1$
if in the corresponding solution job $j$ has not yet finished by
time $t$. For each interval $[s,t]$ we introduce a constraint modeling
that among the jobs released during $[s,t]$, only jobs with a total
processing time of $t-s$ can complete during $[s,t]$ (which is clearly
a necessary condition for feasibility).
\begin{align*}
\min\sum_{j\in J}\sum_{t\aw{\ge}r_{j}} & w_{j}x_{j,t}\\
\sum_{\substack{j\in J\\
s\le r_{j}\le t
}
}x_{j,t}\cdot p_{j} & \ge\sum_{\substack{j\in J\\
s\le r_{j}\le t
}
}p_{j}-(t-s) & \,\,\,\,\,\,\forall s\le t\le T\\
x_{j,t} & \ge x_{j,t+1} & \forall j\in J,t>r_{j}\\
x_{j,t} & \in\{0,1\} & \forall j\in J,t\in\{r_{j},\dotsc,T\}
\end{align*}
Given a feasible schedule, one can easily obtain a feasible solution
to (IP) with the same cost following the intuition for the variables
$x_{j,t}$ above.
Also, one can show that any feasible solution to (IP) can be translated
to a feasible schedule with the same cost. 
\begin{thm}[\cite{DBLP:conf/focs/Batra0K18}]
Suppose that $\left\{ x_{j,t}\right\} _{j,t}$ is a feasible solution
to (IP). Then, there is a schedule for which the total weighted flow-time
is equal to the cost of the solution $\left\{ x_{j,t}\right\} _{j,t}$. 
\end{thm}

One interpretation of (IP) is that for each job $j$ there are segments
$[r_{j},r_{j}+1),[r_{j}+1,r_{j}+2),\dotsc,[T-1,T)$, and we need to select
a prefix of these segments (modeled by the variables $x_{j,t}$ and
the constraints $x_{j,t}\ge x_{j,t+1}$ for each $t$). If we select
a segment $[t-1,t)$ for a job $j$ (i.e., $x_{j,t}=1$), then
this helps us to satisfy the constraint for each interval $[s,t]$
with $s\le r_{j}$. Figure~\ref{fig:visual-IP} provides a visualization
of these constraints: we first sort the jobs non-decreasingly by their
release dates, breaking ties arbitrarily. Denote by $\prec$ the obtained
(fixed) order of the jobs and suppose that the jobs are labeled $1,\dotsc,n$
according to $\prec$. For each job $j$ and each variable $x_{j,t}$
we introduce a square $[t-1,t)\times[j,j+1)$. For each interval
$I=[s,t]$ we define $j(I)$ to be the job $j$ with minimum $r_{j}$
such that $s\le r_{j}$; we introduce a vertical ray $L(I):=\{t-\frac{1}{2}\}\times[j(I)+\frac{1}{2},\infty)$
corresponding to $I$. Then one can show easily that $L(I)$ intersects
the square of a variable $x_{j,t}$ if and only if the variable $x_{j,t}$
appears in the left-hand side of the constraint corresponding to $I$.
Hence, intuitively, the capacity of the square for a variable $x_{j,t}$
is $p_{j}$, the demand of a ray $L(I)$ is the right-hand side of
the constraint in (IP) corresponding to $I$, i.e., $\sum_{j\in J:s\le r_{j}\le t}p_{j}-(t-s)$,
and our goal is to select squares such that each ray $L(I)$ intersects
with selected squares whose total capacity are at least the demand
of $L(I)$.

\begin{figure}
\centering \begin{tikzpicture}[scale=0.8]
  \def\xL{-2}
  \def\y{10}
  \foreach \x in {{\xL},...,15}
  { \draw[thick] (0.5 * \x, \y) rectangle (0.5 * \x + 0.5, \y - 0.5); }
  \node at (8.5, \y - 0.25) {$j_1$};
  \draw[thick] (0.5 * \xL, 6) -- (0.5 * \xL, 5.75) node[pos=1, below] {$r_{j_1}$};
  \fill[pattern=north west lines, pattern color=myblue] (0.5 * \xL, \y) rectangle (0.5 * -2 + 0.5, \y - 0.5);

  \def\xL{0}
  \def\y{9.5}
  \foreach \x in {{\xL},...,15}
  { \draw[thick] (0.5 * \x, \y) rectangle (0.5 * \x + 0.5, \y - 0.5); }
  \node at (8.5, \y - 0.25) {$j_2$};
  \draw[thick] (0.5 * \xL, 6) -- (0.5 * \xL, 5.75) node[pos=1, below] {$r_{j_2}$};
  \fill[pattern=north west lines, pattern color=myblue] (0.5 * \xL, \y) rectangle (0.5 * 7 + 0.5, \y - 0.5);

  \def\xL{3}
  \def\y{9}
  \foreach \x in {{\xL},...,15}
  { \draw[thick] (0.5 * \x, \y) rectangle (0.5 * \x + 0.5, \y - 0.5); }
  \node at (8.5, \y - 0.25) {$j_3$};
  \draw[thick] (0.5 * \xL, 6) -- (0.5 * \xL, 5.75) node[pos=1, below] {$r_{j_3}$};
  \fill[pattern=north west lines, pattern color=myblue] (0.5 * \xL, \y) rectangle (0.5 * 4 + 0.5, \y - 0.5);

  \def\xL{7}
  \def\y{8.5}
  \foreach \x in {{\xL},...,15}
  { \draw[thick] (0.5 * \x, \y) rectangle (0.5 * \x + 0.5, \y - 0.5); }
  \node at (8.5, \y - 0.25) {$j_4$};
  \draw[thick] (0.5 * \xL, 6) -- (0.5 * \xL, 5.75) node[pos=1, below] {$r_{j_4}$};
  \fill[pattern=north west lines, pattern color=myblue] (0.5 * \xL, \y) rectangle (0.5 * 12 + 0.5, \y - 0.5);

  \def\xL{9}
  \def\y{8}
  \foreach \x in {{\xL},...,15}
  { \draw[thick] (0.5 * \x, \y) rectangle (0.5 * \x + 0.5, \y - 0.5); }
  \node at (8.5, \y - 0.25) {$j_5$};
  \draw[thick] (0.5 * \xL, 6) -- (0.5 * \xL, 5.75) node[pos=1, below] {$r_{j_5}$};
  \fill[pattern=north west lines, pattern color=myblue] (0.5 * \xL, \y) rectangle (0.5 * 12 + 0.5, \y - 0.5);

  \draw[myred, ultra thick] (3.75, \y - 0.75) -- (3.75, 9.25);
  \draw[-latex, myred, dotted, ultra thick] (3.75, \y - 0.75) -- (3.75, \y - 1.5);

  \draw[myred, ultra thick] (6.25, \y - 0.75) -- (6.25, 8.25);
  \draw[-latex, myred, dotted, ultra thick] (6.25, \y - 0.75) -- (6.25, \y - 1.5);

  \draw[myred, ultra thick] (2.25, \y - 0.75) -- (2.25, 8.75);
  \draw[-latex, myred, dotted, ultra thick] (2.25, \y - 0.75) -- (2.25, \y - 1.5);
    \draw[thick, -latex] (-2, \y - 2) -- (10, \y - 2) node[pos=1, right] {time};
\end{tikzpicture}
\caption{Geometric visualization of (IP).
\lr{The rays are depicted in red.
The hatched rectangles form an example solution.
The capacities and costs of the rectangles and the demands
of the rays are not depicted}}
\label{fig:visual-IP} 
\end{figure}
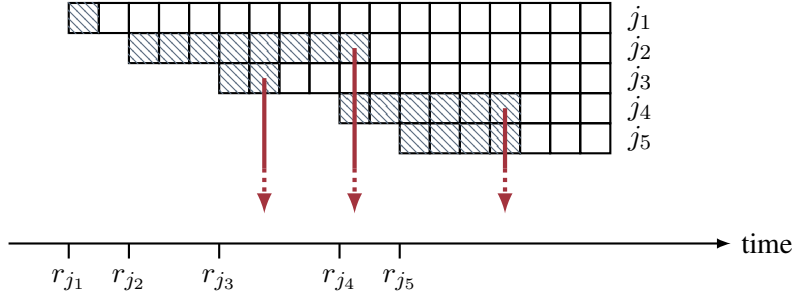
It is not clear how to approximate (IP) directly in polynomial time
and how to make use of the visualization above for this. Therefore,
we will give a randomized reduction of (IP) to a different (and in
particular more structured) integer program (IP2) with the following
relation. 
\begin{enumerate}
\item Any solution for (IP) can be transformed into a solution for (IP2)
such that the objective value increases at most by a factor $1+O(\epsilon)$
in expectation. 
\item Any solution for (IP2) can be transformed into a solution for (IP)
with the same objective value. 
\end{enumerate}
In particular, in (IP2) we will define rectangles for each
job $j$ which intuitively result from merging certain sets of
adjacent squares of $j$. Morever, these rectangles will be aligned
with a hierarchical grid which will help us later to compute a cheap
solution to (IP2) by a dynamic program.

\paragraph{Hierarchical grid.}
Our hierarchical grid has $O_{\epsilon}(\log T)=O_{\epsilon}(\log nP)$
levels. Each grid cell $C$ of some level $\ell$ corresponds to some
interval $[t_{1},t_{2})$ with $t_{1},t_{2}\in\N$. We define
$\cellb(C):=t_{1}$, $\celle(C):=t_{2}$ , and $\len(C):=t_{2}-t_{1}$.
Each cell $C$ has $K:=(2/\epsilon)^{1/\epsilon}$ children
cells of level $\ell+1$, unless $\ell$ is the maximum level $\ell_{\max}$
of the hierarchy in which case $C$ does not have any children cells.
There will be exactly one grid cell of level $0$. The grid is parametrized
by two random variables $\Orel,\Oproc$. Intuitively, we give the
grid a horizontal shift with some random offset $\Orel$. Also, we
choose the size of the unique cell of level $0$ randomly via an offset
$\Oproc$.

Formally, we define $\ell_{\max}$ to be the minimal value $k$
such that $K^{k-2}\ge T$ which will ensure that later the grid cells
$C$ of level $\ell_{\max}$ satisfy that $\len(C)\in[1,K)$. We
choose both $\Oproc\in\{(2/\epsilon)^{0},(2/\epsilon)^{1},\dotsc,(2/\epsilon)^{1/\epsilon-1}\}$
and $\Orel\in\{-K^{\ell_{\max}-1}+1,-K^{\ell_{\max}-1}+2,\dotsc,0\}$
uniformly at random. We define that the unique grid cell of
level $0$ corresponds to the interval $[\Orel,\Orel+\Oproc K^{\ell_{\max}})$
which contains $[0,T)$ (since $T\le K^{\ell_{\max}-2}\le K^{\ell_{\max}}-K^{\ell_{\max}-1}\le\Oproc K^{\ell_{\max}}+\Orel$).
Thus, we can assume w.l.o.g.\ that no job is processed outside $[\Orel,\Orel+\Oproc K^{\ell_{\max}})$.
Inductively, for each grid cell $C$ with $\len(C)\ge K\cdot\Oproc$
of some level $\ell$, we introduce $K$ child grid cells of level
$\ell+1$, one for each interval 
\[
\left[\cellb(C)+\frac{i}{K}\len(C),\ \cellb(C)+\frac{i+1}{K}\len(C)\right),\quad i=0,\dotsc,K-1.
\]
By construction, the interval of each grid cell of level $\ell$ has
length $\Oproc\cdot K^{\ell_{\max}-\ell}$. Denote by $\C$ the set
of all grid cells (of all levels). For each $C\in\C$, denote by $\ell(C)$
its level. It follows that $\ell_{\max}=\max_{C\in\C}\ell(C)$.

\paragraph{Segments of jobs.}

For each job $j$ we want to define a set of $O_{\epsilon}(\log nP)$
segments $\Seg(j)$ which form a partition of $[r_{j},T)$,
see Figure~\ref{fig:cells}. We will associate each segment
$S\in\Seg(j)$ with some grid cell $C\in\C$ such that $S\subseteq C$
and denote by $\Seg(j,C)$ the segments in $\Seg(j)$ associated
with $C$. We will ensure that all segments in $\Seg(j,C)$ are aligned
with the grid cells of level $\ell(C)+2$ and in particular all have
the same size. We will also ensure that for each $C\in\C$,
the union of the segments in $\Seg(j,C)$ forms an interval that is
\emph{right-aligned with }$C$, i.e., it holds that $\bigcup_{S\in\Seg(j,C)}S=[s,\celle(C))$
for some $s\in C$.

Formally, consider a job $j$. We construct a sequence of cells $C_{\ell_{\max}},C_{\ell_{\max}-1},\dotsc,C_{0}$
in levels $\ell_{\max},\ell_{\max}-1,\dotsc,0$ such that the union
of these cells contains $[r_{j},T)$. The cells are chosen as follows.
Cell $C_{\ell_{\max}}$ is identical to the cell of of level $\ell_{\max}$
that contains $r_{j}$. Suppose we have chosen cells $C_{\ell_{\max}},\dotsc,C_{k}$.
Then we define $C_{k-1}$ as the cell of level $k-1$
that contains $\celle(C_{k})$ (see Figure~\ref{fig:cells}); observe
that this implies $\celle(C_{k})<\celle(C_{k-1})$. For each $k\in\{\ell_{\max}-1,\dotsc,1\}$
consider the interval $[\celle(C_{k+1}),\celle(C_{k}))$, and $[r_{j},\celle(C_{\ell_{\max}}))$
for $k=\ell_{\max}$. The length of this interval must be an integer
multiple of $\len(C_{k+1})$,
or 1 if $k=\ell_{\max}$. We subdivide this interval into segments
of length $\len(C_{k+2})$ each if $k\le\ell_{\max}-2$, and
of length~1 each if $k>\ell_{\max}-2$.
We define that these segments form the set $\Seg(j,C_{k})$.
It follows by construction that each of them coincides with a cell
of level $k+2$ if $k\le\ell_{\max}-2$, and otherwise has length~1.

For each job $j$ we define $\Seg(j):=\bigcup_{C\in\C}\Seg(j,C)$.
In the following lemma, we will prove some properties of these segments.
For this, for each job $j$ let $F_{j}^{*}$ denote the flowtime of
$j$ in $\OPT$ and let $C_{j}^{*}$ denote the cell such that $\Seg(j,C_{j}^{*})$
contains a segment $S$ with $r_{j}+F_{j}^{*}-1\in S$. We will use
this notation in the (technical) fifth property below \lr{that} will be crucial
later to prove that the reduction to (IP2) loses only a factor $1+O(\epsilon)$.
The last point states intuitively that the segments of a job are coarser
if the job is released earlier, see also Figure~\ref{fig:cells}.

\begin{restatable}{lem}{definesegments}
\label{lem:define-segments}For each job $j\in J$ the segments \aw{$\Seg(j)$}
and $\left\{ \Seg(j,C)\right\} _{C\in\C}$ have the following properties:
\begin{enumerate}
\item $\Seg(j)$ forms a partition of $[r_{j},T)$,%
\item for each $C\in\C$ and each $S\in\Seg(j,C)$ it holds that $S\subseteq C$
and $S=C'$ for some cell $C'$ of level $\ell(C)+2$ if $\ell(C)\le\ell_{\max}-2$,
and $S=[t,t+1)$ for some $t\in\N$ if $\ell(C)>\ell_{\max}-2$,
\item for each $C\in\C$ the interval $\bigcup_{S\in\Seg(j,C)}S$ is right-aligned
with $C$, $\left|\Seg(j,C)\right|\le K^{2}$, and all segments in
$\Seg(j,C)$ have the same size, \label{en:define-segments-no-segs} 
\item for two segments $S\in\Seg(j,C)$, $S'\in\Seg(j,C')$ where $S$ lies
on the left of $S'$ and $C\ne C'$, it holds that \lr{$\len(S')=\len(S)=1$
or }$\len(S')=\len(S)\cdot K^{i}$ for some integer $i\ge 1$, and 
\item with probability at least $1-O(\epsilon)$ we have that $F_{j}^{*}\ge\len(C_{j}^{*})/(\epsilon K)$.\label{lem:define-segments-F*} 
\end{enumerate}
Moreover, for two jobs $j,j'$ with $r_{j}\le r_{j'}$ it holds that
for each segment $S'\in\Seg(j')$ there is a segment $S\in\Seg(j)$
with $S'\subseteq S$. 
\end{restatable}
\begin{proof}
\lr{
By construction the first four properties follow immediately. The non-trivial property we need to show
is the fifth.}
First we will show that probability at least
$1-3\epsilon$ it holds that 
\begin{equation}
F_{j}^{*}\notin[\epsilon/2\cdot\Oproc K^{k},2/\epsilon\cdot\Oproc K^{k}]\quad\text{for all }k\in\mathbb{Z}.\label{eq:F*-in-range}
\end{equation}
Note that if (\ref{eq:F*-in-range}) is not true for $k$ then $F_{j}^{*}$
and $\Oproc K^{k}$ differ by a factor of at most $2/\epsilon$.
In other words, it suffices to show that with the mentioned probability
we have 
\[
|\log_{2/\epsilon}F_{j}^{*}-\log_{2/\epsilon}(\Oproc K^{k})|>1\quad\text{for all }k\in\mathbb{Z}.
\]
Notice that $\log_{2/\epsilon}(\Oproc K^{k})\in\mathbb{Z}$. Hence,
the statement above is implied by 
\begin{align*}
\lfloor\log_{2/\epsilon}F_{j}^{*}\rfloor & \neq\log_{2/\epsilon}(\Oproc)+k/\epsilon-1,\\
\lfloor\log_{2/\epsilon}F_{j}^{*}\rfloor & \neq\log_{2/\epsilon}(\Oproc)+k/\epsilon,\phantom{-1}\quad\text{and}\\
\lfloor\log_{2/\epsilon}F_{j}^{*}\rfloor & \neq\log_{2/\epsilon}(\Oproc)+k/\epsilon+1\quad\text{for all }k\in\mathbb{Z}.
\end{align*}
This is equivalent to 
\begin{align*}
\lfloor\log_{2/\epsilon}F_{j}^{*}\rfloor & \not\equiv\log_{2/\epsilon}(\Oproc)-1\mod1/\epsilon,\\
\lfloor\log_{2/\epsilon}F_{j}^{*}\rfloor & \not\equiv\log_{2/\epsilon}(\Oproc)\phantom{,-1}\mod1/\epsilon,\text{and}\\
\lfloor\log_{2/\epsilon}F_{j}^{*}\rfloor & \not\equiv\log_{2/\epsilon}(\Oproc)+1\mod1/\epsilon.
\end{align*}
The distribution of $\log_{2/\epsilon}(\Oproc)$ is uniform over
$\{0,\dotsc,1/\epsilon-1\}$. Hence (\ref{eq:F*-in-range}) holds
with probability at least $1-3\epsilon$. We condition on the event
above which implies that there is some $k\in\mathbb{Z}$ with $2/\epsilon\cdot\Oproc K^{k-1}<F_{j}^{*}<\epsilon/2\cdot\Oproc K^{k}$.
Because of $F_{j}^{*}\ge1\ge\Oproc/K$ it must hold that $\Oproc K^{k}\ge2/\epsilon$
and $k\ge0$. Moreover, since $F_{j}^{*}\le r_{j}+F_{j}^{*}\le T\le\Oproc K^{\ell_{\max}-2}$
we have that $k<\ell_{\max}-1$.

Let $C_{\ell_{\max}},\dotsc,C_{0}$ be the cells constructed in the
definition of $\Seg(j)$. Recall that $\Orel$ is chosen uniformly
at random from $\{-K^{\ell_{\max}-1}+1,\dotsc,0\}$. The number $K^{\ell_{\max}-1}$
is an integer multiple of $\Oproc K^{k}$. Thus, the distribution
of $\Orel\mod\Oproc K^{k}$ is uniform. With probability at least
$1-2\epsilon$ we have 
\[
r_{j}\not\equiv\Orel,\Orel-1,\dotsc,\Orel-\epsilon\Oproc K^{k}\mod\Oproc K^{k},
\]
which means the grid cells are aligned such that $r_{j}$ lies inside
a cell $C$ of level $\ell_{\max}-k$, that is, $\len(C)=\Oproc K^{k}$,
and $r_{j}<\celle(C)-\epsilon\len(C)$. Together with $(*)$ this
event has a probability of at least $1-5\epsilon$. We now prove
that the event implies $C=C_{j}^{*}$, which finishes the proof since
\[
F_{j}^{*}>2/\epsilon\cdot\Oproc K^{k-1}\ge\len(C)/(\epsilon K).
\]
First, we prove that $C=C_{\ell_{\max}-k}$. If $k=0$, this follows
from $r_{j}\in C$. Otherwise, it follows from 
\[
\cellb(C)\le r_{j}\le\celle(C_{\ell_{\max}-k+1})
\]
and 
\begin{multline*}
\celle(C_{\ell_{\max}-k+1})\le r_{j}+\sum_{i=\ell_{\max}-k+1}^{\ell_{\max}}\len(C_{i})=r_{j}+\sum_{i=0}^{k-1}\Oproc K^{i}<r_{j}+2\Oproc K^{k-1}\\
<\celle(C)-\epsilon\Oproc K^{k}+2\Oproc K^{k-1}=\celle(C).
\end{multline*}
Finally, $C=C_{j}^{*}$ since 
\[
r_{j}+F_{j}^{*}-1<r_{j}+\frac{\epsilon}{2}\Oproc K^{k}\le\celle(C)-\epsilon\cdot\Oproc K^{k}+\frac{\epsilon}{2}\Oproc K^{k}<\celle(C)
\]
and for $k\ge1$ we have 
\begin{multline*}
r_{j}+F_{j}^{*}-1>r_{j}+\frac{2}{\epsilon}\Oproc K^{k-1}-1>r_{j}+\frac{1}{\epsilon}\sum_{i=0}^{k-1}\Oproc K^{i}-1\\
>r_{j}+\sum_{i=\ell_{\max}-k+1}^{\ell_{\max}}\len(C_{i})\ge\celle(C_{\ell_{\max}-k+1})\qedhere
\end{multline*}
\end{proof}

\begin{figure}
\centering \begin{tikzpicture}[scale=0.8]
  \draw[ultra thick, mygreen] (5, -0.9) rectangle (5, 2.5) node[pos=1, yshift=7pt] {$r_j$};
  \draw[ultra thick, mygreen] (6.5, -0.9) rectangle (6.5, -1.9) node[pos=1, yshift=-7pt] {$r_{j'}$};
  \node at (4-0.2, 2.0) {$C_{\ell_{\max}\phantom{-0}}$};
  \draw[thick] (5, 2.0 - 0.15) rectangle (5.3, 2.0 + 0.15);
  \node at (4-0.2, 1.5) {$C_{\ell_{\max}-1}$};
  \draw[thick] (4.7, 1.5 - 0.15) rectangle (5.6, 1.5 + 0.15);
  \node at (4-0.2, 1.0) {$C_{\ell_{\max}-2}$};
  \draw[thick] (5.6, 1.0 - 0.15) rectangle (8.3, 1.0 + 0.15);
  \node at (4-0.2, 0.5) {$C_{\ell_{\max}-3}$};
  \draw[thick] (5.6, 0.5 - 0.15) rectangle (13.7, 0.5 + 0.15);

  \node[thick] at (4, -0.4) {$\Seg(j)$};
  \draw[thick, myblue] (5, -0.4 - 0.15) rectangle (5.1, -0.4 + 0.15);
  \draw[thick, myblue] (5.1, -0.4 - 0.15) rectangle (5.2, -0.4 + 0.15);
  \draw[thick, myblue] (5.2, -0.4 - 0.15) rectangle (5.3, -0.4 + 0.15);
  \draw[thick, myblue] (5.3, -0.4 - 0.15) rectangle (5.4, -0.4 + 0.15);
  \draw[thick, myblue] (5.4, -0.4 - 0.15) rectangle (5.5, -0.4 + 0.15);
  \draw[thick, myblue] (5.5, -0.4 - 0.15) rectangle (5.6, -0.4 + 0.15);
  \draw[thick, myblue] (5.6, -0.4 - 0.15) rectangle (5.9, -0.4 + 0.15);
  \draw[thick, myblue] (5.9, -0.4 - 0.15) rectangle (6.2, -0.4 + 0.15);
  \draw[thick, myblue] (6.2, -0.4 - 0.15) rectangle (6.5, -0.4 + 0.15);
  \draw[thick, myblue] (6.5, -0.4 - 0.15) rectangle (6.8, -0.4 + 0.15);
  \draw[thick, myblue] (6.8, -0.4 - 0.15) rectangle (7.1, -0.4 + 0.15);
  \draw[thick, myblue] (7.1, -0.4 - 0.15) rectangle (7.4, -0.4 + 0.15);
  \draw[thick, myblue] (7.4, -0.4 - 0.15) rectangle (7.7, -0.4 + 0.15);
  \draw[thick, myblue] (7.7, -0.4 - 0.15) rectangle (8.0, -0.4 + 0.15);
  \draw[thick, myblue] (8.0, -0.4 - 0.15) rectangle (8.3, -0.4 + 0.15);
  \draw[thick, myblue] (8.3, -0.4 - 0.15) rectangle (9.2, -0.4 + 0.15);
  \draw[thick, myblue] (9.2, -0.4 - 0.15) rectangle (10.1, -0.4 + 0.15);
  \draw[thick, myblue] (10.1, -0.4 - 0.15) rectangle (11.0, -0.4 + 0.15);
  \draw[thick, myblue] (11.0, -0.4 - 0.15) rectangle (11.9, -0.4 + 0.15);
  \draw[thick, myblue] (11.9, -0.4 - 0.15) rectangle (12.8, -0.4 + 0.15);
  \draw[thick, myblue] (12.8, -0.4 - 0.15) rectangle (13.7, -0.4 + 0.15);

  \node[thick] at (4, -1.4) {$\Seg(j')$};
  \draw[thick, myblue] (6.5, -1.4 - 0.15) rectangle (6.6, -1.4 + 0.15);
  \draw[thick, myblue] (6.6, -1.4 - 0.15) rectangle (6.7, -1.4 + 0.15);
  \draw[thick, myblue] (6.7, -1.4 - 0.15) rectangle (6.8, -1.4 + 0.15);
  \draw[thick, myblue] (6.8, -1.4 - 0.15) rectangle (6.9, -1.4 + 0.15);
  \draw[thick, myblue] (6.9, -1.4 - 0.15) rectangle (7.0, -1.4 + 0.15);
  \draw[thick, myblue] (7.0, -1.4 - 0.15) rectangle (7.1, -1.4 + 0.15);
  \draw[thick, myblue] (7.1, -1.4 - 0.15) rectangle (7.2, -1.4 + 0.15);
  \draw[thick, myblue] (7.2, -1.4 - 0.15) rectangle (7.3, -1.4 + 0.15);
  \draw[thick, myblue] (7.3, -1.4 - 0.15) rectangle (7.4, -1.4 + 0.15);
  \draw[thick, myblue] (7.4, -1.4 - 0.15) rectangle (7.7, -1.4 + 0.15);
  \draw[thick, myblue] (7.7, -1.4 - 0.15) rectangle (8.0, -1.4 + 0.15);
  \draw[thick, myblue] (8.0, -1.4 - 0.15) rectangle (8.3, -1.4 + 0.15);
  \draw[thick, myblue] (8.3, -1.4 - 0.15) rectangle (9.2, -1.4 + 0.15);
  \draw[thick, myblue] (9.2, -1.4 - 0.15) rectangle (10.1, -1.4 + 0.15);
  \draw[thick, myblue] (10.1, -1.4 - 0.15) rectangle (11.0, -1.4 + 0.15);
  \draw[thick, myblue] (11.0, -1.4 - 0.15) rectangle (11.9, -1.4 + 0.15);
  \draw[thick, myblue] (11.9, -1.4 - 0.15) rectangle (12.8, -1.4 + 0.15);
  \draw[thick, myblue] (12.8, -1.4 - 0.15) rectangle (13.7, -1.4 + 0.15);
\end{tikzpicture} \caption{Example cell and segment construction with $K=3$}
\label{fig:cells} 
\end{figure}
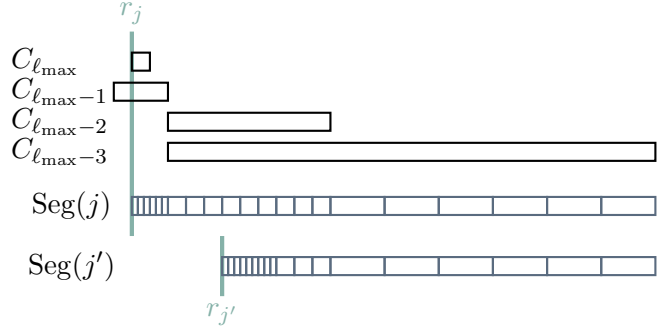

Based on the segments in the sets $\Seg(j,C)$ we define (IP2) where
we introduce a variable $y_{j,S}$ for each job $j$ and each segment
$S\in\Seg(j)$. This variable $y_{j,S}$ models whether we select
the segment $S$ for job $j$ which implies that we allow job $j$
to finish after time $\celle(S)$ (and are willing to pay for this).
This is similar to the variables $x_{j,t}$ in (IP). Like in (IP),
we have a constraint for each interval $[s,t]$. For each set $\Seg(j,C)$
we define that the first segment $S\in\Seg(j,C)$ has cost $c_{j,S}=w_{j}(\celle(S)-r_{j})$
and each other segment $S'\in\Seg(j,C)$ has cost $c_{j,S'}=w_{j}\len(S')$.
Moreover, we require that from each set $\Seg(j,C)$ a prefix of its
segments is selected. Thus, intuitively, if we select the first segment
$S$ of a set $\Seg(j,C)$ then we pay the full price for not processing
job $j$ until $\celle(S)$, and for each other segments $S'\in\Seg(j,C)$
we pay the price for delaying $j$ by $\len(S')$ more.
\begin{align}
\min\sum_{j\in J}\sum_{S\in\Seg(j)} & c_{j,S}y_{j,S}\nonumber \\
\sum_{\substack{j\in J\\
s\le r_{j}\le t
}
}\sum_{\substack{S\in\Seg(j)\\
t\in S
}
}y_{j,S}\cdot p_{j} & \ge\sum_{\substack{j\in J\\
s\le r_{j}\le t
}
}p_{j}-(t-s) &  & \,\,\,\,\,\,\forall s\le t\le T\nonumber \\
y_{j,S} & \ge y_{j,S'} &  & \begin{aligned} & \forall j\in J,C\in\C\ \forall S,S'\in\Seg(j,C)\\
 & \quad\text{with }\celle(S)<\celle(S')
\end{aligned}
\label{eq:prefix}\\
y_{j,S} & \in\{0,1\} &  & \forall j\in J\forall S\in\Seg(j)\text{ }\nonumber 
\end{align}
We prove now that
by reducing (IP) to (IP2) we lose only a factor of $1+O(\epsilon)$
in expectation. We define $\OPT^{\mathrm{(IP)}}$ and $\OPT^{\mathrm{(IP2)}}$
to be the costs of the optimal solutions to (IP) and (IP2), respectively.
Note that $\OPT^{\mathrm{(IP2)}}$ depends on $\Orel$ and $\Oproc$. 
\begin{restatable}{lem}{ipleip}
\label{lem:IP-le-IP2}
For all possible values for the offsets $\Orel,\Oproc$ it
holds that $\OPT^{\mathrm{(IP)}}\le\OPT^{\mathrm{(IP2)}}$. 
\end{restatable}
\begin{proof}
Consider some solution for (IP2). Let
$j$ be a job and let $S\in\Seg(j)$ be the rightmost segment $S$
with $y_{j,S}=1$, i.e., the segment with maximal $\celle(S)$ such
that $y_{j,S}=1$. In (IP) we set $x_{j,t}=1$ for each $t<\celle(S)$
and $x_{j,t}=0$ for each $t\ge\celle(S)$. This forms a feasible
solution of at most the same cost: For feasibility consider the covering
constraint in (IP) and (IP2) for some fixed $s\le t\le T$. As the
right-hand side is equal in both integer programs, it suffices to
show that the left-hand side of (IP) is at least as big as that in
(IP2), that is, 
\begin{equation}
\sum_{\substack{j\in J\\
s\le r_{j}\le t
}
}x_{j,t}\cdot p_{j}\ge\sum_{\substack{j\in J\\
s\le r_{j}\le t
}
}\sum_{\substack{S\in\Seg(j)\\
t\in S
}
}y_{j,S}\cdot p_{j}\label{eq:left-hands}
\end{equation}
The elements of the sums in (\ref{eq:left-hands}) correspond to jobs.
For each job $j$ with $s\le r_{j}\le t$ the left sum contains $p_{j}$
if $x_{j,t}=1$ and $0$, otherwise. The right sum contains $p_{j}$,
if $y_{j,S^{(t)}}=1$ for the segment $S^{(t)}\in\Seg(j)$ that contains
$t$ and $0$, otherwise. By definition of $x_{j,t}$, however, we
know that if $y_{j,S^{(t)}}=1$, then $x_{j,t}=1$ as well. Hence,
(\ref{eq:left-hands}) follows.

For the cost of the solution we will consider each job independently,
that is, we show that for each job $j$ it holds that 
\[
\sum_{t>r_{j}}w_{j}x_{j,t}\le\sum_{S\in\Seg(j)}c_{j,S}y_{j,S}.
\]
Let $C$ be the largest cell such that $y_{j,S}=1$ for some $S\in\Seg(j,C)$
and let $S_{1},S_{2},\dotsc,S_{k}\in\Seg(j,C)$ be the maximal prefix
of segments with $y_{j,S_{i}}=1$ for $i=1,\dotsc,k$. By definition
of $x_{j,t}$ we know that $x_{j,t}=1$ if and only if $t<\celle(S_{k})$.
Thus, 
\begin{align*}
\sum_{S\in\Seg(j)}c_{j,S}y_{j,S}\ge\sum_{i=1}^{k}c_{j,S_{i}} & =w_{j}(\celle(S_{1})-r_{j})+\sum_{i=2}^{k}w_{j}\len(S_{i})\\
 & =w_{j}(\celle(S_{1})-r_{j})+\sum_{i=2}^{k}w_{j}(\celle(S_{i})-\celle(S_{i-1}))\\
 & =w_{j}(\celle(S_{k})-r_{j})=\sum_{t\ge r_{j}}w_{j}x_{j,t}\qedhere
\end{align*}
\end{proof}

On the other hand, we prove that in expectation $\OPT^{\mathrm{(IP2)}}$
is not much more expensive than $\OPT^{\mathrm{(IP)}}$. Given an
optimal solution to (IP), we define a solution to (IP2) which incurs
for each job $j$ a cost of at most $(1+O(\epsilon))F_{j}^{*}w_{j}$
if the fifth condition of Lemma~\ref{lem:define-segments} is satisfied
for $j$ (which happens with probability $1-O(\epsilon)$). On
the other hand, we show that even if this condition is \emph{not}
satisfied for $j$, then the cost of $j$ in (IP2) is at most $O(F_{j}^{*}w_{j})$,
which yields a cost of at most $(1+O(\epsilon))F_{j}^{*}w_{j}$ in
expectation. Taking the sum over all jobs $j$ yields the following
lemma.
\begin{restatable}{lem}{costincreasesmall}
\label{lem:expected-cost-increase-small}It holds that $\mathbb{E}\left[\OPT^{\mathrm{(IP2)}}\right]\le(1+O(\epsilon))\OPT^{\mathrm{(IP)}}$. 
\end{restatable}
\begin{proof}
Let $F_{j}^{*}$ denote the flow time
in an optimal solution for (IP), that is, the optimal solution is
defined with $x_{j,t}=1$ if and only if $t<r_{j}+F_{j}^{*}$. For
each job $j$ and we set $y_{j,S}=1$ for all segments $S\in\Seg(j)$
that intersect with $[r_{j},r_{j}+F_{j}^{*})$ and $y_{j,S}=0$, otherwise.

For feasibility consider the covering constraint in (IP) and (IP2)
for some fixed $s\le t\le T$. As the right-hand side is equal in
both integer programs, it suffices to show that the left-hand side
of (IP2) is at least as big as that in (IP), that is, 
\begin{equation}
\sum_{\substack{j\in J\\
s\le r_{j}\le t
}
}x_{j,t}\cdot p_{j}\le\sum_{\substack{j\in J\\
s\le r_{j}\le t
}
}\sum_{\substack{S\in\Seg(j)\\
t\in S
}
}y_{j,S}\cdot p_{j}\label{eq:left-hands2}
\end{equation}
Let $j\in J$ with $s\le r_{j}\le t$. We argue that if $x_{j,t}=1$
then also $y_{j,S^{(t)}}=1$ for the segment $S^{(t)}\in\Seg(j)$
with $t\in S^{(t)}$. Indeed, this follows from the definition of
$y_{j,S}$, since $S^{(t)}$ intersects with $[r_{j},r_{j}+F_{j}^{*})$
(both contain $t$). Thus (\ref{eq:left-hands2}) holds.

For the cost of the solution we consider each job individually, that
is, we show that $\sum_{S\in\Seg(j)}c_{j,S}y_{j,S}$ is at most $(1+O(\epsilon))w_{j}F_{j}^{*}$
in expectation. More precisely, we first argue that it never exceeds
$O(w_{j}F_{j}^{*})$; then we show that with probability $1-O(\epsilon)$
it does not exceed $(1+O(\epsilon))w_{j}F_{j}^{*}$. To this end,
we fix a job $j$.

Let $C_{\ell_{\max}},\dotsc,C_{0}$ be the sequence of cells in the
construction of $\Seg(j)$. Let $k\in\N$ such that $C_{k}=C_{j}^{*}$,
that is, there is a segment $S^{*}\in\Seg(j,C_{k})$ with $r_{j}+F_{j}^{*}-1\in S^{*}$.
Observe that the costs of segments are chosen in a way that for each
$C_{i}$, $i>k$, we have 
\[
\sum_{S\in\Seg(j,C_{i})}c_{j,S}y_{j,S}=\sum_{S\in\Seg(j,C_{i})}c_{j,S}=w_{j}(\celle(C_{i})-r_{j}).
\]
Further, for cell $C_{k}$ we have 
\begin{equation}
\sum_{S\in\Seg(j,C_{k})}c_{j,S}y_{j,S}=w_{j}(\celle(S^{*})-r_{j})\le w_{j}F_{j}^{*}+w_{j}\len(S^{*}).\label{eq:seg-k}
\end{equation}
We first bound (\ref{eq:seg-k}) by $2w_{j}F_{j}^{*}$. If $k\in\{\ell_{\max},\ell_{\max}-1\}$
then this holds trivially, because $\len(S^{*})=1\le F_{j}^{*}$.
Otherwise, we have that $F_{j}^{*}\ge\len(C_{k})/K^{2}=\len(S^{*})$.
The first inequality holds because $[\celle(C_{k+2}),\celle(C_{k+1}))$
is contained in $[r_{j},r_{j}+F_{j}^{*})$ and its length is an integer
multiple of $\len(C_{k+2})=\Oproc K^{\ell_{\max}-(k+2)}=\len(C_{k})/K^{2}$.
It follows that 
\begin{align*}
\sum_{S\in\Seg(j)}c_{j,S}y_{j,S} & =w_{j}(\celle(S^{*})-r_{j})+\sum_{i=k+1}^{\ell_{\max}}w_{j}(\celle(C_{i})-r_{j})\\
 & \le2w_{j}(\celle(S^{*})-r_{j})+\sum_{i=k+2}^{\ell_{\max}}w_{j}(\celle(C_{i})-r_{j})\\
 & \le4w_{j}F_{j}^{*}+\sum_{i=k+2}^{\ell_{\max}}w_{j}\sum_{\ell=i}^{\ell_{\max}}\len(C_{\ell})
\end{align*}
Moreover, 
\begin{multline*}
\sum_{i=k+2}^{\ell_{\max}}w_{j}\sum_{\ell=i}^{\ell_{\max}}\len(C_{\ell})=\sum_{i=k+2}^{\ell_{\max}}w_{j}\sum_{\ell=i}^{\ell_{\max}}\Oproc K^{\ell_{\max}-\ell}\\
\le2\Oproc\sum_{i=k+2}^{\ell_{\max}}w_{j}K^{\ell_{\max}-i}\le4w_{j}\Oproc K^{\ell_{\max}-(k+2)}\le4w_{j}F_{j}^{*}.
\end{multline*}
We conclude that for all $\Oproc,\Orel$ it holds that 
\[
\sum_{S\in\Seg(j)}c_{j,S}y_{j,S}\le8w_{j}F_{j}^{*}.
\]
It remains to prove that with probability $1-O(\epsilon)$ the
selected segments have cost at most $(1+O(\epsilon))w_{j}F_{j}^{*}$.
To this end, assume we are in the case of Lemma~\ref{lem:define-segments}:\ref{lem:define-segments-F*}.
In other words, $F_{j}^{*}\ge\len(C_{k})/(\epsilon K)$. This implies
\begin{align*}
\sum_{S\in\Seg(j)}c_{j,S}y_{j,S} & =w_{j}(\celle(S^{*})-r_{j})+w_{j}\sum_{i=k+1}^{\ell_{\max}}(\celle(C_{i})-r_{j})\\
 & \le w_{j}F_{j}^{*}+w_{j}\len(S^{*})+w_{j}\sum_{i=k+1}^{\ell_{\max}}\sum_{\ell=i}^{\ell_{\max}}\len(C_{\ell}).
\end{align*}
Furthermore, $\len(S^{*})=\len(C_{k})/K^{2}\le\epsilon F_{j}^{*}/K\le\epsilon F_{j}^{*}$
and 
\begin{multline*}
w_{j}\sum_{i=k+1}^{\ell_{\max}}\sum_{\ell=i}^{\ell_{\max}}\len(C_{\ell})=w_{j}\sum_{i=k+1}^{\ell_{\max}}\sum_{\ell=i}^{\ell_{\max}}\Oproc K^{\ell_{\max}-\ell}\le2w_{j}\sum_{i=k+1}^{\ell_{\max}}\Oproc K^{\ell_{\max}-i}\\
\le4w_{j}\Oproc K^{\ell_{\max}-(k+1)}=4w_{j}\len(C_{k})/K\le4\epsilon w_{j}F_{j}^{*}.
\end{multline*}
This concludes the proof. 
\end{proof}

Note that Lemma~\ref{lem:expected-cost-increase-small} implies that
there exist values for $\Oproc,\Orel$ such that $\OPT^{\mathrm{(IP2)}}\le(1+O(\epsilon))\OPT^{\mathrm{(IP)}}$;
since the number of combinations for $\Oproc,\Orel$ is bounded by
$O_{\epsilon}(T)$ we simply guess these values.

\subsection{Geometric visualization}
\begin{figure}
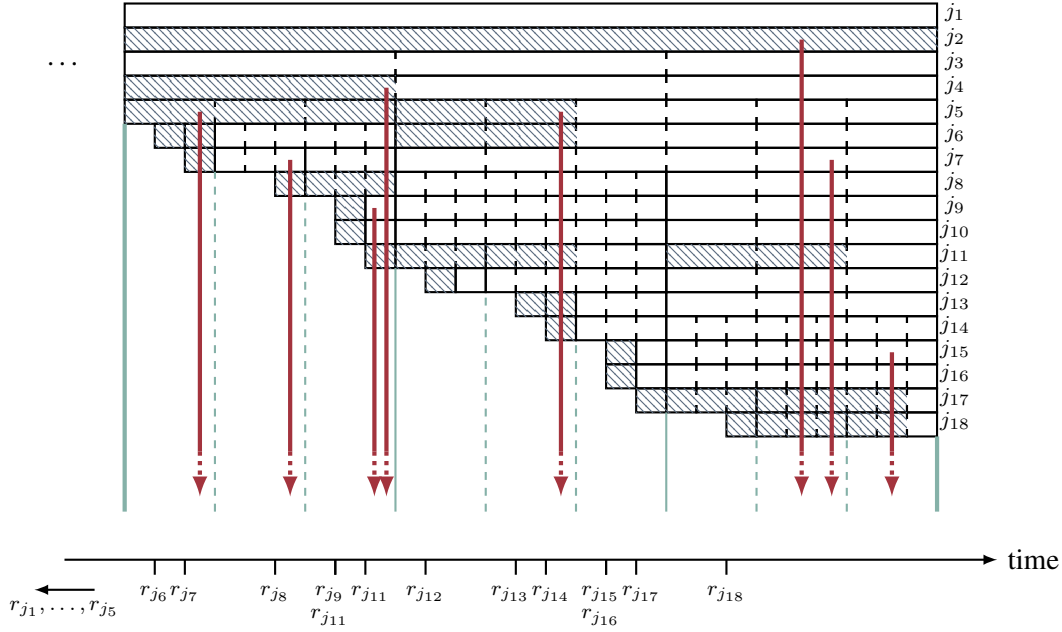

\centering
\figgeometric{full}
\caption{Geometric visualization of (IP2).
The rays are depicted in red and the hierachical decomponsition is visualized in green.
The hatched rectangles form an example solution.
The line between adjacent rectangles is interrupted, if they belong to the same set $\R(j, C)$,
i.e., these are the sets of which a solution needs to select a prefix.
The capacities and costs of the rectangles and the demands
of the rays are not depicted.}
\label{fig:geometric-problem2} 
\end{figure}

We now visualize (IP2) in a similar way as (IP) before (see Figure~\ref{fig:geometric-problem2}.
Again, we assume that the jobs are labeled $1,\dotsc,n$
according to $\prec$. For each job $j$ and each segment $S\in\Seg(j)$
we introduce a rectangle $R(j,S)=[\cellb(S),\celle(S))\times[j,j+1)$.
For each job $j$ and each cell $C\in\C$ we define $\R(j,C):=\{R(j,S) \mid S\in\Seg(j,C)\}$,
for each cell $C$ let $\R(C):=\bigcup_{j}\R(j,C)$, and additionally
we define $\R:=\bigcup_{C\in\C}\R(C)$. For each interval $I=[s,t]$
we define $j(I)$ to be the job $j$ with minimum $r_{j}$ such that
$s\le r_{j}$; we introduce a vertical ray $L(I):=\{t+\frac{1}{2}\}\times[j(I)+\frac{1}{2},\infty)$
corresponding to $I$. Then $L(I)$ intersects a rectangle $R(j,S)$
if and only if the variable $y_{j,S}$ appears in the left-hand side
of the constraint corresponding to $I$. 
\begin{lem}
\label{lem:ray-intersect} Let $I=[s,t]$. The ray $L(I)$ intersects
a rectangle $R(j,S)$ corresponding to a segment $S$ if and only
if $s\le r_{j}\le t$ and $t\in S$. 
\end{lem}

\begin{proof}
Suppose that $L(I)$ and $R(j,S)$ intersect. Then $\cellb(S)\le t+\frac{1}{2}<\celle(S)$
and $r_{j(I)}\le r_{j}$. Note that $s\le r_{j(I)}$ and hence $s\le r_{j}$.
Also, $t+\frac{1}{2}\ge\cellb(S)\ge r_{j}$ holds. On the other hand,
assume that $s\le r_{j}\le t$ and $t\in S$. Then $j(I)\prec j$
and thus $[j(I)+\frac{1}{2},\infty)\cap[j,j+1)\ne\emptyset$. Therefore,
$L(I)$ and $R(j,S)$ intersect. 
\end{proof}
For each ray $L(I)$ corresponding to an interval $I=[s,t]$ we define
a demand of $d(I):=\sum_{j\in J:s\le r_{j}\le t}p_{j}-(t-s)$ (which
is the right-hand side of the constraint corresponding to $I$ in
(IP2)). For each rectangle $R=R(j,S)$ we define a capacity $p_{R}:=p_{j}$
and a cost $c_{R}:=c_{j,S}$. This yields a geometric covering problem
in which our goal is to select rectangles (respecting the prefix constraints~\eqref{eq:prefix})
of minimum total cost such that each ray $L(I)$ intersects selected
rectangles with a total capacity of at least $d(I)$.

We will solve this problem approximately with a dynamic program. In
our DP, we will take advantage of the hierarchical structure induced
by the cells $\C$. To this end, note that for a cell $C$ with two
children cells $C_{1},C_{2}\subseteq C$, the rectangles in $\mathcal{R}(C_{1})$
and $\mathcal{R}(C_{2})$ are independent, in the sense that if a
ray $L(I)$ intersects a rectangle in $\mathcal{R}(C_{1})$ then it
does not intersect any rectangle in $\mathcal{R}(C_{2})$ and vice
versa.

\section{Computing an approximate solution}

Assume that we are given the cells $\C$ and the rectangles $\R$
as defined above. We want to compute a set $\R'\subseteq\R$ of small
total cost that represents a feasible\emph{ }solution to (IP2),
i.e., such that if we set $y_{j,S}:=1$ for each rectangle $R(j,S)\in\R'$
then we satisfy (IP2).

The cells $\C$ induce a tree $G=(V,E)$ as follows. For each cell
$C\in\C$ we introduce a vertex $v_{C}$ in $V$. We connect two vertices
$v_{C},v_{C'}$ by an edge $\{v_{C},v_{C'}\}$ if $C$ is the
parent cell of $C'$ in the hierarchy, i.e., if $C'\subseteq C$ and
$C$ is of level $\ell$ and $C'$ is of level $\ell+1$ for some
$\ell\in\N$.
We define that the root of $G$ is the vertex that corresponds to
the unique cell of level 0.

Let $\Q$ denote the set of all paths in $G$ for which one of the
endpoints is the root of $G$. For convenience, for a path $Q\in\Q$
we write $C\in Q$ if $v_{C}\in Q$ (i.e., abusing notation we also
interpret $Q$ as a set of cells). For each path $Q\in\Q$ we define
$\R(Q):=\cup_{C\in Q}\R(C)$, i.e., all rectangles assigned to cells
on $Q$. Let $\I$ denote the set of all intervals $I=[s,t]$ with
$0\le s\le t\le T$. For each interval $I\in\I$ we define $\R(I)$
as a set of all rectangles $R(j,S)$ such that $R(j,S)$ and $L(I)$
intersect. 
\begin{restatable}{lem}{coverinterval}
\label{lem:path-cover-interval}For each interval $I=[s,t]\in\I$
there is a path $Q\in\Q$ such that $\R(I)\subseteq\R(Q)$. 
\end{restatable}
\begin{proof}
By Lemma~\ref{lem:ray-intersect} the set $\R(I)$ contains exactly
those rectangles $R(j,S)$ where $s\le r_{j}\le t$ and $t\in S$.
Let $C_{0},C_{1},\dotsc,C_{\ell_{\max}}$ be the cells of level $0,1,\dotsc,\ell_{\max}$
that contain $t$. Note that for each level there is exactly one such
cell and $v_{C_{i+1}}$ must be the child of $v_{C_{i}}$ for each
$i$. This is precisely the path $Q$ such that $\R(I)\subseteq\R(Q)$:
Let $R(j,S)\in\R(I)$. Then $t\in S$. This means the cell $C$ which
$S$ is assigned to must also contain $t$ (since $S\subseteq C$).
Hence $C$ is in the path $Q$ and $R(j,S)\in\R(C)\subseteq\R(Q)$. 
\end{proof}

Next, we define a framework for approximating our problem by a dynamic
program; in~\cite{DBLP:conf/focs/Batra0K18} a similar framework
was implicitly used. We will define a global solution $\R'\subseteq\R$
and for each path $Q\in\Q$ a subset $\R'_{Q}\subseteq\R'\cap\R(Q)$.
We will ensure that for each $Q\in\Q$ and each interval $I\in\I$
with $\R(I)\subseteq\R(Q)$ the rectangles in $\R'_{Q}$ are sufficient
to satisfy the demand of $L(I)$. Also, we want the subsets $\left\{ \R'_{Q}\right\} _{Q\in\Q}$
to be consistent in the sense that for two paths $Q,Q'\in\Q$ with
$Q\supseteq Q'$ (i.e., $Q$ is an extension of $Q'$), the set $\R'_{Q}$
can contain only those rectangles from cells in $Q'$ that are contained
in $\R'_{Q'}$ (but possibly $\R'_{Q}$ does not contain all of them).
Moreover, we want that for each set $\R'_{Q}$ there are only polynomially
many candidates. Therefore, we will require for each $Q\in\Q$
that $\R'_{Q}\in\chi_{Q}$ for a family of sets $\chi_{Q}$ that we
can compute in time $(nP)^{O_{\epsilon}(1)}$, and hence in particular
$|\chi_{Q}|\le(nP)^{O_{\epsilon}(1)}$. These properties will be
useful for our dynamic program later.

Formally, we require $(\R',\left\{ \R'_{Q}\right\} _{Q\in\Q})$ to
be a consistent solution as defined below.
For any set of rectangles $\R'\subseteq\R$ we define $p(\R'):=\sum_{R(j,S)\in\R'}p(j,S)$
and $c(\R'):=\sum_{R(j,S)\in\R'}c(j,S)$. 
\begin{defn}
\label{def:framework} Let \emph{$\left\{ \chi_{Q}\right\} _{Q\in\Q}$}
be a family with $\chi_{Q}\subseteq2^{\R(Q)}$ for each $Q\in\Q$.
Let $\R'\subseteq\R$ and $\left\{ \R'_{Q}\right\} _{Q\in\Q}$ with
$\R'_{Q}\in\chi_{Q}$ for each $Q\in\Q$. We say that $(\R',\left\{ \R'_{Q}\right\} _{Q\in\Q})$
forms a \emph{consistent solution for $\left\{ \chi_{Q}\right\} _{Q\in\Q}$}
if 
\begin{enumerate}
\item $\R'_{Q}\subseteq\R'$ for each $Q\in\Q$, \label{en:framework-subset} 
\item for each $I\in\I,Q\in\Q$ with $\R(I)\subseteq\R(Q)$, we have that
$\R'_{Q}$ covers $I$, i.e., $p(\R'_{Q}\cap\R(I))\ge d(I)$, \label{en:framework-cover} 
\item for any two paths $Q,Q'\in\Q$ with $Q\supseteq Q'$ we have that
$\R'_{Q}\cap\R(Q')\subseteq\R'_{Q'}$. \label{en:framework-paths} 
\end{enumerate}
We define $c(\R')$ to be the \emph{cost }of $(\R',\left\{ \R'_{Q}\right\} _{Q\in\Q})$.

It should be noted that by Definition~\ref{def:framework} the rectangles
in $\R'$ form a feasible solution. This follows from the first and
second property and Lemma~\ref{lem:path-cover-interval}: For every
interval $I\in\I$ there is a $Q\in\Q$ with $\R(I)\subseteq\R(Q)$.
By the second property we have $p(\R'_{Q}\cap\R(I))\ge d(I)$ and
since $\R'_{Q}\subseteq\R'$ we also have $p(\R'\cap\R(I))\ge d(I)$.
We can compute the cheapest consistent solution for a given family
$\left\{ \chi_{Q}\right\} _{Q\in\Q}$ with an easy dynamic program. 
\end{defn}

\begin{restatable}{lem}{dynprog}
\label{lem:DP}Given a family $\left\{ \chi_{Q}\right\} _{Q\in\Q}$,
we can compute the cheapest consistent solution for $\left\{ \chi_{Q}\right\} _{Q\in\Q}$
in time $\left(|\I|\cdot|\Q|\cdot\max_{Q\in\Q}|\chi_{Q}|\right)^{O(1)}$. 
\end{restatable}
\begin{proof}
We build a dynamic programming table that contains an entry for each
pair $(v,\S')$, where $v$ is a vertex and $\S'\in\chi_{Q}$
for $Q$ which we define to be the path from the root to $v$.
This entry stores a set of rectangles $\S\subseteq\bigcup_{Q':Q'\supseteq Q}\R(Q')$,
that is, $\S$ contains rectangles that belong to cells that are either
descendants of $v$ or on the path from $v$ to the root. The set
$\S$ is chosen such that $\S\cap\R(Q)=\S'$ and $p(\S\cap\R(I))\ge d(I)$
for every $Q'\supseteq Q$ and $I\in\I$ with $\R(I)\subseteq\R(Q)$.
If there is no such $\S$ then a special symbol indicates that the
choice of $\S'$ is infeasible, that is, the value $\R'_{Q}$ in a
consistent solution cannot be $\S'$.

We fill the table starting with the leafs and then compute each inner
node's entries using the previously computed childrens' entries. Suppose
that $v$ is a leaf and let $Q$ be the path from the root to $v$.
We check for each $\S'\in\chi_{Q}$ whether for all $I\in\I$ with
$\R(I)\subseteq\R(Q)$ it holds that $p(\S'\cap\R(I))\ge d(I)$. If
so, we store $\S'$ in the entry for $(v,\S')$. Otherwise, we insert
a special symbol indicating that the choice is infeasible.

Now let $v$ be an inner node and let again $Q$ be the path from
root to $v$. Let $\S'\in\chi_{Q}$. In the following we describe
how to compute the table entry for $(v,\S')$. Let $u_{1},\dotsc,u_{K}$
be the children of $v$ and $Q_{1},\dotsc,Q_{K}$ the extension of
$Q$ to each child. For each $i=1,\dotsc,K$ let $\S'_{i}\in\chi_{Q_{i}}$
be the set $\S'_{i}$ with
\begin{equation}
\S'_{i}\cap\R(Q)\subseteq\S'\label{eq:dyn-subset}
\end{equation}
for which the set $\S_{i}$ stored in $(u_{i},\S'_{i})$ minimizes
$c(\S_{i}\setminus R(Q))$. If for some $i$ no such $\S'_{i}$
exists, then we determine that the choice
$\S'$ is infeasible. Otherwise, we insert for $(v,\S')$ the entry
\[
\S=\S'\cup\S_{1}\cup\cdots\cup\S_{K}.
\]
Eventually, this dynamic program will compute various solutions for
the root $r$, namely one solution $\S$ for each $S'\in\chi_{\{r\}}$.
As an overall solution $\S^{*}$ we output the solution $\S$
that minimizes $c(\S)$. We define $\{\S_{Q}^{*}\}_{Q\in\Q}$ by memoization:
Recall that $\S^{*}=\S'\cup\S_{1}\cup\cdots\cup\S_{K}$ where $\S'\in\chi_{\{r\}}$,
$\S_{i}$ is the entry for $(u_{i},\S'_{i})$, $u_{i}$ is the $i$'th
child of $r$, and $\S'_{i}\in\chi_{\{r,u_{i}\}}$. We set $\S_{\{r\}}^{*}=\S'$.
Likewise, we set $\S_{\{r,u_{i}\}}^{*}=\S'_{i}$. Each $\S_{i}$ is
derived from selections at the children of $u_{i}$. In the same way
we recursively define $\S_{Q}^{*}$ on each path $Q$. Indeed, $(\S^{*},\{\S_{Q}^{*}\}_{Q\in\Q})$
forms a consistent solution:

Let $Q=\{v_{1},v_{2},v_{3},\dotsc,v_{k}\}\in\Q$ where $v_{1}$ is
the root. For all $i\le k$ let $Q^{(\le i)}=\{v_{1},\dotsc,v_{i}\}$.
By construction we have that $\S_{Q^{(\le i)}}^{*}\subseteq\S_{Q^{(\le i-1)}}^{*}$.
In particular, 
\[
\S_{Q^{(\le k)}}^{*}\subseteq\S_{Q^{(\le k-1)}}^{*}\subseteq\cdots\subseteq\S_{Q^{(\le1)}}^{*}=\S^{*}.
\]
This proves (\ref{en:framework-subset}) of Definition~\ref{def:framework}.
Moreover, by Equation~(\ref{eq:dyn-subset}) the dynamic program
ensures that 
\[
\S_{Q^{(\le k)}}^{*}\cap\R(Q^{(\le k-1)})\subseteq\S_{Q^{(\le k-1)}}^{*}.
\]
It follows for all $i<k$ that 
\[
\S_{Q^{(\le k)}}^{*}\cap\R(Q^{(\le i)})=(\S_{Q^{(\le k)}}^{*}\cap\R(Q^{(\le k-1)}))\cap\R(Q^{(\le i)})\subseteq\S_{Q^{(\le k-1)}}^{*}\cap\R(Q^{(\le i)}).
\]
Iterating this argument we obtain 
\[
\S_{Q^{(\le k)}}^{*}\cap\R(Q^{(\le i)})\subseteq\S_{Q^{(\le i)}}^{*}\cap\R(Q^{(\le i)})=\S_{Q^{(\le i)}}^{*}.
\]
and thus (\ref{en:framework-paths}) of Definition~\ref{def:framework}
holds. Finally, we prove (\ref{en:framework-cover}) of Definition~\ref{def:framework}.
To this end let $I\in\I,Q\in\Q$ with $\R(I)\subseteq\R(Q)$. We need
to show that $p(\S_{Q}^{*}\cap\R(I))\ge d(I)$. Let $Q'\supseteq Q$
be any extension of $Q$ ending in a leaf. Then also $\R(I)\subseteq\R(Q)\subseteq\R(Q')$.
The way we define the dynamic program on leafs it holds that $p(\S_{Q'}^{*}\cap\R(I))\ge d(I)$.
Moreover, since we already showed (\ref{en:framework-paths}) it follows
that $\S_{Q'}^{*}\cap\R(Q)\subseteq\S_{Q}^{*}$. Hence, 
\[
p(\S_{Q}^{*}\cap\R(I))\ge p((\S_{Q'}^{*}\cap\R(Q))\cap\R(I))=p(\S_{Q'}^{*}\cap\R(I))\ge d(I).
\]
It remains to check that this dynamic program indeed computes the
cheapest consistent solution. To this end let $(\R',\{\R'_{Q}\}_{Q\in\Q})$
be the cheapest consistent solution. We show inductively that for
each path $Q$ from the root to a vertex $v$ the entry $\S$ computed
for $(v,\R'_{Q})$ satisfies $c(\S)\le c(\bigcup_{Q'\supseteq Q}\R'_{Q})$.
It follows that $\S^{*}$ is of minimal cost, because $c(\S^{*})$
is at most the cost of the entry computed for $(r,\R'_{\{r\}})$ which
is at most 
\[
c(\bigcup_{Q\supseteq\{r\}}\R'_{Q})=c(\bigcup_{Q\in\Q}\R'_{Q})\le c(\R').
\]
If $v$ is a leaf the claim is satisfied by definition, since the
entry of the dynamic table is $\R'_{Q}=\bigcup_{Q'\supseteq Q}\R'_{Q'}$.
Now assume that $v$ is not a leaf. Let---as in the definition of
the dynamic program---$u_{1},\dotsc,u_{K}$ be the children of $v$
and $Q_{1},\dotsc,Q_{K}$ the extensions of $Q$ to each child. Let
$\hat{\S}_{i}$ be the entry computed for $(u_{i},\R'_{Q_{i}})$,
$i=1,\dotsc,K$. By induction hypothesis we have for each $i=1,\dotsc,K$
that $c(\hat{\S}_{i})\le c(\bigcup_{Q'\supseteq Q_{i}}\R'_{Q'})$.
Since the rectangle sets in both sides contain the same rectanges
from $\R(Q)$, namely $\R'_{Q_{i}}\cap\R(Q)$, we also have 
\[
c(\hat{\S}_{i}\setminus\R(Q))\le c(\bigcup_{Q'\supseteq Q_{i}}\R'_{Q'}\setminus\R(Q)).
\]
The rectangles at entry $(v,\R'_{Q})$ were chosen as $\S=\R'_{Q}\cup\S_{1}\cup\cdots\cup\S_{K}$
where $\S_{i}$ minimizes $c(\S_{i}\setminus\R(Q))$ over all entries
$\S_{i}$ at $(u_{i},\S'_{i})$ with $\S'_{i}\in\chi_{Q_{i}}$ and
$\S'_{i}\cap\R(Q)\subseteq\R'_{Q}$. Since $\R'_{Q_{i}}\cap\R(Q)\subseteq\R'_{Q}$
by property (\ref{en:framework-paths}) of Definition~\ref{def:framework},
$\hat{S}_{i}$ is among these candidates and in particular $c(\S_{i}\setminus\R(Q))\le c(\hat{\S}_{i}\setminus\R(Q)$.
Hence, 
\begin{multline*}
c(\S)\le c(\R'_{Q}\cup\S_{1}\cup\cdots\cup\S_{K})\\
\le c(\R'_{Q})+c(\bigcup_{Q'\supseteq Q_{1}}\R'_{Q'}\setminus\R(Q))+\cdots+c(\bigcup_{Q'\supseteq Q_{K}}\R'_{Q'}\setminus\R(Q))\le c(\bigcup_{Q'\supseteq Q}\R'_{Q'}).
\end{multline*}
This finishes the proof that $(\S^{*},\{S_{Q}^{*}\}_{Q\in\Q})$ is
a consistent solution of minimal cost.

The claimed running time follows because there are $|\Q|\max_{Q\in\Q}|\chi_{Q}|$
entries in the dynamic table, computing each leaf's entry requires
$O(|\I|\cdot\max_{Q\in\Q}|\chi_{Q}|)$ operations, and computing each
inner vertex's entry requires $O(K\max_{Q\in\Q}|\chi_{Q}|)\le|\Q|\max_{Q\in\Q}|\chi_{Q}|$
operations. 
\end{proof}

The hard part is to show that in polynomial time we can compute a
polynomial size family $\left\{ \chi_{Q}\right\} _{Q\in\Q}$ that
admits a consistent solution of small cost. We will prove the following
lemma in Section~\ref{subsec:quasi-polynomial-size-consistent}
and Section~\ref{subsec:Existence-cheap-consistent-solution}. 
\begin{lem}
\label{lem:cheap-consistent-solution}In time $(nP)^{O_{\epsilon}(1)}$
we can compute a family $\left\{ \chi_{Q}\right\} _{Q\in\Q}$ with
$\max_{Q\in\Q}|\chi_{Q}|\le(nP)^{O_{\epsilon}(1)}$ for which there
exists a consistent solution of cost at most $(2+\epsilon)\OPT^{\mathrm{(IP2)}}$. 
\end{lem}
Then Lemmas~\ref{lem:DP} and \ref{lem:cheap-consistent-solution}
yield a $(2+\epsilon)$-approximation algorithm with a running time
of $(nP)^{O_{\epsilon}(1)}$. The black-box reduction in \cite[Section 4]{DBLP:conf/soda/FeigeKL19}
then implies our main result. 
\begin{thm}
There exists a polynomial time $(2+\epsilon)$-approximation algorithm
for weighted flow time on a single machine when preemptions are allowed. 
\end{thm}

\section{\label{subsec:quasi-polynomial-size-consistent}Quasi-polynomial
size consistent solution}

In this section, we prove a weaker variant of Lemma~\ref{lem:cheap-consistent-solution}
which already introduces several of our key techniques and leads to
a quasi-polynomial $(2+\epsilon)$-approximation. More precisely,
in this section we relax the condition in Lemma~\ref{lem:cheap-consistent-solution}
on the size of each set $\chi_{Q}$ with $Q\in\Q$ to $|\chi_{Q}|\le(nP)^{O_{\epsilon}(\log^{2}(nP))}$
and also the running time to $(nP)^{O_{\epsilon}(\log^{2}(nP))}$.

For each rectangle $R$ we define a \emph{density} $\rho_{R}$ which
approximately describes its cost-efficiency $c_{R}/p_{R}$. Instead
of using this ratio directly, we define $\rho_{R}:=(1+\epsilon)^{k}$
for the value $k\in\mathbb{Z}$ with $(1+\epsilon)^{k}\le c_{R}/p_{R}<(1+\epsilon)^{k+1}$.
In this way, $\rho_{R}$ differs from $c_{R}/p_{R}$ only by a small
factor of $1+\epsilon$, but we ensure that there are only $O_{\epsilon}(\log nP)$
different densities overall. Recall that we defined the set $\R(j,C)$
for combinations of a job $j$ and cell $C\in\C$ (which contains
all rectangles in $\R(C)$ corresponding to $j$). By construction,
almost all of these rectangles have the same cost $c_{R}$,
apart from the leftmost rectangle in $\R(j,C)$ whose cost might be
higher. Thus, we can describe the densities of the rectangles
in $\R(j,C)$ by only two values that we denote by $\rho_{j,C},\rho'_{j,C}$.
Formally, let $R,R'\in\R(j,C)$ be the leftmost and second leftmost
rectangles in $\R(j,C)$, respectively. We define $\rho_{j,C}:=\rho_{R}$
and $\rho'_{j,C}:=\rho_{R'}$; in case that $\left|\R(j,C)\right|\le1$
we define $\rho'_{j,C}:=\infty$ and if $\left|\R(j,C)\right|=0$
we define also $\rho{}_{j,C}:=\infty$.
Using these values $\rho_{j,C}$ and $\rho'_{j,C}$ we classify the
sets $\left\{ \R(j,C)\right\} _{j\in J,C\in\C}$ into types. 
\begin{defn}
For a job $j$ and a cell $C$ with $\R(j,C)$ we say that $\R(j,C)$
is of \emph{type }$\tau=(\rho,\rho',s)$ if $|\R(j,C)|=s$, $\rho_{j,C}=\rho$,
and $\rho'_{j,C}=\rho'$. 
\end{defn}

It turns out that there are only $O_{\epsilon}(\log(nP))$ different
types $\tau=(\rho,\rho',s)$ arising in the input, since in each set
$\R(j,C)$ the costs of the leftmost and second leftmost rectangles
differ only by a factor $O_{\epsilon}(1)$, all rectangles have the
same capacity, and $s=O_{\epsilon}(1)$. 
\begin{restatable}{lem}{notypes}\label{lem:no-types}
There are at most $O_{\epsilon}(\log(nP))$ different types $\tau$
for which there exists a set $\R(j,C)$ of type $\tau$. Moreover,
for each $\rho$ there are only $O_{\epsilon}(1)$ many pairs
$\rho',s$ for which there is a set $\R(j,C)$ of type $\tau=(\rho,\rho',s)$. 
\end{restatable}

\begin{proof}
Consider a type $\tau=(\rho,\rho',s)$ for which there exists a set
$\R(j,C)$ of type $\tau$. By Property~\ref{en:define-segments-no-segs}
of Lemma~\ref{lem:define-segments} we have that $s\in\{1,2,\dotsc,K^{2}\}$.
Moreover, the costs of different rectangles
within the same set $\R(j,C)$ can only differ by a factor of $K^{2}$:
\lr{
Recall, the rectangles in $\R(j, C)$ (and their costs) are derived from segments $\S(j, C)$.
Let $S\in\Seg(j, C)$. If $S$ is not the leftmost segment in $\Seg(j,C)$, then $c_{S}=w_{j}\len(S) \le w_j \len(C)$
and $\len(S)$ is either $\len(C)/K^{2}$ or $1$. The latter applies
if $\ell(C)\in\{\ell_{\max},\ell_{\max}-1\}$ and therefore $\len(C)\le\Oproc K\le K^{2}$.
In both cases we can bound the cost of the rectangle from below by $c_S \ge w_j \len(C) \ge w_j \len(C) / K^2$. Now suppose
that $S$ is the leftmost segment. Then $c_{S}=w_{j}(\celle(S)-r_{j})$.
Since this is at least $w_{j}\len(S)$, the lower bound holds as before.
Finally, notice that
\[
\celle(S)-r_{j}\le\len(S)+\cellb(C)-r_{j}\le\len(S)+\sum_{i=1}^{\infty}\len(C)/K^{i}\le\len(S)+\len(C)\cdot2/K\le\len(C).
\]
It follows that $c_{S}\le w_{j}\len(C)$. 
Hence, each rectangle in $\R(j, C)$ has a cost between $w_j \len(C) / K^2$ and $w_j \len(C)$.
}
This implies that $\rho/(K^{2}(1+\epsilon))\le\rho'\le(1+\epsilon)K^{2}\rho$.
The number of powers of $(1+\epsilon)$ in $[\rho/(K^{2}(1+\epsilon)),(1+\epsilon K^{2}\rho]$
is only 
\[
O(\log_{1+\epsilon}((1+\epsilon^{2}K^{4}))\le O_{\epsilon}(1).
\]
Hence for a fixed $\rho$, there are only $K^{2}\cdot O_{\epsilon}(1)=O_{\epsilon}(1)$
types. We will show that $1/((1+\epsilon)P)\le\rho\le O_{\epsilon}(n^{3}P^{2})$.
The number of powers of $(1+\epsilon)$ in $[1/P,O_{\epsilon}(n^{3}P^{2})]$
is 
\[
O(\log_{1+\epsilon}(O_{\epsilon}((1+\epsilon)n^{3}P^{3})))\le O_{\epsilon}(\log(nP)).
\]
Therefore there are only $O_{\epsilon}(\log(nP))$ possibilities
for $\rho$ and consequently $O_{\epsilon}(\log(nP))$ relevant
types overall.

Let us now prove the claimed bounds for $\rho$. Let $R$ be the rectangle
corresponding to $\rho$. Recall that the capacity $p_{R}$ is bounded
by $P$. Moreover, the cost $c_{R}$ is defined as $w_{j}(t_{2}-t_{1})$
for some interval $[t_{1},t_{2}]$. The right border $t_{2}$ is bounded
by 
\[
\Orel+\Oproc K^{\ell_{\max}}\le2K^{\ell_{\max}+1}\le K^{3}T\le2K^{3}(\max_{j}r_{j}+\sum_{j}p_{j})\le4K^{3}\sum_{j}p_{j}\le4K^{3}nP\le O_{\epsilon}(nP).
\]
Recall by preprocessing we have $1\le w_{j}\le O_{\epsilon}(n^{2}P)$.
Hence $1\le c_{R}\le O_{\epsilon}(n^{3}P^{2})$. This means that
\[
\rho\ge1/(1+\epsilon)\cdot c_{R}/p_{R}\ge1/(1+\epsilon)\cdot1/P.
\]
On the other hand 
\[
\rho\le c_{R}/p_{R}\le c_{R}\le O_{\epsilon}(n^{3}P^{2}).\qedhere
\]
\end{proof}

Let $\R^{*}\subseteq\R$ be the rectangles corresponding to $\OPT^{\mathrm{(IP2)}}$.
We define quantities that describe how much of its budget the solution
$\R^{*}$ spends within each cell $C$ for sets $\R(j,C)$
of each type $\tau=(\rho,\rho',s)$, and how much of this amount it
spends on jobs for which it buys exactly the first $s'$ rectangles,
for each $s'\in\{1,\dotsc,s\}$. Formally, for each cell $C$, each type
$\tau=(\rho,\rho',s)$, and each $s'\in\{1,\dotsc,s\}$ let $\R^{*}(C,\tau,s')$
be the set of all rectangles $R\in\R^{*}$ for which there is
a job $j$ such that $R\in\R(j,C)$ and
$\R^{*}$ contains exactly the first $s'$ rectangles from $\R(j,C)$.
We define $B^{\opt}(C,\tau,s'):=\sum_{R\in\R^{*}(C,\tau,s')}c_{R}$.

We define now our solution $\R'$. For a cell $C$ and a type
$\tau$ the solution $\bigcup_{s'}\R^{*}(C,\tau,s')$ can be very
complicated. Instead, we construct an algorithm GreedySelect which,
intuitively, computes a simple solution of total cost at most $(2+\epsilon)\sum_{s'}B^{\opt}(C,\tau,s')$
that covers as much of each ray $L(I)$ as the rectangles in $\bigcup_{s'}\R^{*}(C,\tau,s')$.
For computing it, we need to know only $\{B^{\opt}(C,\tau,s')\}_{s'}$.
Then $\R'$ will consist of the union of all these simple solutions
for all cells $C\in\Q$ and types $\tau$ and for each path $Q\in\Q$
we will simply define $\R'_{Q}:=\R'\cap\R(Q)$. Then there are only
$(nP)^{O_{\epsilon}(\log^{2}(nP))}$ options for $\R'_{Q}$ since
it depends only on the $O_{\epsilon}(\log^{2}(nP))$ budgets $\{B^{\opt}(C,\tau,s')\}_{C,s':C\in Q}$.

\paragraph{Procedure GreedySelect.}

Formally, the input of GreedySelect consists of a cell $C$, a type
$\tau=(\rho,\rho',s)$, and for each $s'\in\{1,\dotsc,s\}$ a budget
$B(s')$ (for the purpose of this section we can think this value
as $B^{\opt}(C,\tau,s')$). It selects rectangles of total cost at
most $\aw{(2+\epsilon)}\sum_{s'}B(s')$ from the sets in $\left\{ \R(j,C)\right\} _{j}$
that are of type $\tau$. We will denote by GreedySelect$(C,\tau,\left\{ (B(s')\right\} _{s'})$
the computed rectangles.

Note that for two jobs $j,j'$ for which $\R(j,C)$ and $\R(j',C)$
are of the same type $\tau$, the rectangles in these two sets look
identical, up to a vertical shift (and they might have different costs).
We first define a fractional
solution greedily. For each job $j$ for which $\R(j,C)$ is of type
$\tau$ and each $r\in\{1,\dotsc,s\}$ we define a value $x_{j,r}$
which denotes the fractional extent to which we select the $r$-th
rectangle in $\R(j,C)$. Initially, we define $x_{j,r}=0$ for each
such variable.
For each $s'=s,s-1,\dotsc,1$ we start a phase in which we consider
the jobs for which there is a $\tilde{s}\ge s'$ such that buying
the first $\tilde{s}$ rectangles in $\R(j,C)$ costs at most $B(\tilde{s})$.
Observe that through the phases more and more jobs satisfy this condition
and hence if a job is available in one phase then it will also be
available in all future phases. We sort the corresponding jobs decreasingly
by $\prec$ (so in particular non-increasingly by their release dates)
and we consider them in this order. Note that in our graphical
visualization this orders the jobs from bottom to top. When we consider
a job $j$, for all $r\le s'$ we increase $x_{j,r}$ simultaneously
by the same amount until either $x_{j,r}=1$ for each such $r$ or
we paid exactly $(1+\epsilon)B_{s'}$ in this phase (fractionally).
Hence, the fractional cost is $(1+\epsilon)\sum_{s'}B(s')$ by construction.

We chose our ordering for the jobs since our rays are vertical
and downwards oriented and, hence, if a rectangle is located further
down, it intersects with more rays whose demands it helps to satisfy.
In particular, here we crucially exploit that for each interval $I$
the corresponding object $L(I)$ is a vertical ray, rather than e.g.,
a line segment. Using this, we will show that if $B(s')\ge B^{\opt}(C,\tau,s')$
for each $s'$ then our fractional solution $\left\{ x_{j,r}\right\} _{j,r}$
covers as much from each ray $I\in\I$ as the rectangles in $\R^{*}\cap\R(C)$
of type $\tau$. In the output of GreedySelect we select for each
job $j$ and each $r\in\{1,\dotsc,s\}$ the $r$-th rectangle in $\R(j,C)$
if $x_{j,r}>0$, i.e., intuitively we round up each variable $x_{j,r}$
with $x_{j,r}>0$. We will show that for each $r$ there is at most
one job $j$ such that $0<x_{j,r}<1$ and hence we pay additionally
at most $\sum_{s'}B(s')$ due to the rounding.

\begin{restatable}{lem}{greedyselect}
\label{lem:greedyselect} Suppose that for each $s'$ it holds that
$B(s')\ge B^{\opt}(C,\tau,s')$. Then for each interval $I\in\I$
it holds that 
$p(\mathrm{GreedySelect}(C,\tau,\left\{ (B(s')\right\} _{s'})\cap\R(I))\ge p\left(\bigcup_{s'}\R^{*}(C,\tau,s')\cap\R(I)\right)$
and additionally $c(\mathrm{GreedySelect}(C,\tau,\left\{ (B(s')\right\} _{s'}))\le(2+\epsilon)\sum_{s'}B(s')$. 
\end{restatable}

\begin{proof}
First, we claim that already the fractional solution covers as much
from each $I\in\I$ as $\R^{*}$. To this end, let $I\in\I$ and $r\in\N$
such that $L(I)$ intersects the $r$-th rectangle of each job of
type $\tau$ in $\R(C)$. We want to show that the fractional solution
covers $L(I)$ at least as much as $\R^{*}$.

We say that a job $j$ is of \emph{kind $k$ }if $k$ is the largest
value $k'$ such that buying the first $k'$ rectangles of $\R(j,C)$
costs at most $B(k')$ and one of the first $k'$ rectangles of $\R(j,C)$
intersects $L(I)$. Observe that some jobs are of no kind at all,
however, such segments are not selected in $\R^{*}$ (since $B(s')\ge B^{\opt}(C,\tau,s')$
for each $s'$) or do not intersect $L(I)$. In other words, it suffices
to show that GreedySelect covers $L(I)$ at least as much as $\R^{*}$
with rectangles that are of some kind $k$. Intuitively, if $j$ is
of kind $k$ then $k$ is the earliest round in which we might have
selected rectangles of $\R(j,C)$. 

Let $\hat{s}'$ be the minimal value $s'$ such that after iteration
$s'$ (which refers to the value of $s'$ in this iteration; recall
that these values decrease through the iterations) the algorithm has
selected (possibly partially in previous iterations) the first $s'\ge r$
rectangles of all jobs of kind $k\ge s'$. Regarding rectangles of
kinds $k$ with $r\le k<\hat{s}'$, we know that in each iteration
$k<\hat{s}'$ GreedySelect spends by a factor of $1+\epsilon$ more
budget than $\R^{*}$ (since $B(k)\ge B^{\opt}(C,\tau,k)$). Since
we consider rectangles $R$ of the same type $\tau$, their ratios
$c_{R}/p_{R}$ can only differ by a factor of $(1+\epsilon)$. In
particular, the rectangles of all kinds $k<\hat{s}'$ selected by
the algorithm have a total (fractional) size that is at least as large
as the corresponding rectangles in $\R^{*}$. Also, we sort the corresponding
jobs decreasingly by $\prec$ and the rays $L(I)$ are vertical and
downward oriented. This implies that $p(\mathrm{GreedySelect}(C,\tau,\left\{ (B(s')\right\} _{s'})\cap\R(I))\ge p\left(\bigcup_{s'}\R^{*}(C,\tau,s')\cap\R(I)\right)$
for each $I\in\I$.

By construction, the cost of the fractional solution is at most $(1+\epsilon)\sum_{s'}B(s')$.
We argue that the cost increases by at most another $\sum_{s'}B(s')$
when we round up the fractional solution. To this end, we claim that
in the fractional solution for each kind $k$ and each $r\le k$ there
can be at most one job $j$ of kind $k$ such that for its $r$-th
rectangle it holds that $0<x_{j,r}<1$. The claim is clearly true
before the first iteration. Suppose that it is true after the iteration
that corresponds to some value $k'$. Suppose that in the next iteration
corresponding to $k'-1$ a value $x_{j,r}$ is increased such that
before $x_{j,r}=0$. Let $k$ be the kind of the corresponding job
$j$. Assume by contradiction that there is some other variable $x_{j',r}$
with $0<x_{j',r}<1$ corresponding to some other job $j'$ of kind
$k$. If $j\prec j'$ then the algorithm would have increased $x_{j',r}$
instead of $x_{j,r}$ in this iteration $k'-1$. If $j'\prec j$ then
in the previous iteration in which $x_{j',r}$ was increased, it would
have increased $x_{j,r}$ instead. Also, note that if a job $j$ is
of kind $k$ then always $x_{j,r}=0$ for each $r>k$.

When we round up the fractional solution, then for each kind $k$
and each $r\le k$ we round up at most one variable $x_{j,r}$ for
a job $j$ of kind $k$. The total cost of rounding up all these rectangles
for this kind $k$ is bounded by $B(k)$. Hence, the total cost of
rounding up is bounded by $\sum_{s'}B(s')$ which yields a total cost
of $(1+\epsilon)\sum_{s'}B(s')+\sum_{s'}B(s')$ as claimed.
\end{proof}

\paragraph{Definition of consistent solution.}

As mentioned above, we define the set of all rectangles $\R'$ in our solution by
$\R':=\bigcup_{C}\bigcup_{\tau}\mathrm{GreedySelect}(C,\tau,\left\{ B^{\opt}(C,\tau,s')\right\} _{s'})$.
For each path $Q\in\Q$ we define $\R'_{Q}:=\R'\cap\bigcup_{C\in Q}\R(Q)$
and observe that $\R'_{Q}$ can be computed with GreedySelect once
we know all $O_{\epsilon}(\log^{2}(nP))$ budgets $\left\{ B^{\opt}(C,\tau,s')\right\} _{\tau,s',C:C\in Q}$.
Each of them is an integer, bounded by $O_{\epsilon}(nT\cdot\max_{j}w_{j})\le O_{\epsilon}(n^{4}P^{2})$,
which yields only $(nP)^{O_{\epsilon}(\log^{2}(nP))}$ possibilities
overall. We define $\chi_{Q}$ to contain each of these possibilities.
This proves Lemma~\ref{lem:cheap-consistent-solution} if we relax
the condition on the size of each set $\chi_{Q}$ to $|\chi_{Q}|\le(nP)^{O_{\epsilon}(\log^{2}(nP))}$
and allow a running time of $(nP)^{O_{\epsilon}(\log^{2}(nP))}$.

\section{\label{subsec:Existence-cheap-consistent-solution}Polynomial size
consistent solution}

In this section we prove Lemma~\ref{lem:cheap-consistent-solution}
(without any relaxations of its statement). We start by defining the
solution $\R'$ and the sets $\left\{ \R'_{Q}\right\} _{Q\in\Q}$.
Afterwards, we define the family $\left\{ \chi_{Q}\right\} _{Q\in\Q}$.

In the approach in Section~\ref{subsec:quasi-polynomial-size-consistent}
we guessed the values $\{B^{\opt}(C,\tau,s')\}_{C,\tau,s'}$ and recovered
an approximate solution using only them. For a path $Q$ it seems
unlikely to be able to guess all $O_{\epsilon}(\log^{2}(nP))$ values
$B^{\opt}(C,\tau,s')$ corresponding to cells $C\in Q$ (or sufficiently
strong approximate variants of them) in polynomial time. However,
consider the values $\{B^{\opt}(C)\}_{C\in\C}$, where $B^{\opt}(C)=\sum_{\tau,s'}B^{\opt}(C,\tau,s').$
\aw{Note that for each path $Q\in\Q$ only the $O_{\epsilon}(\log(nP))$
values in $\{B^{\opt}(C)\}_{C\in Q}$ are relevant. Our first step
is to construct suitable substitutes for these quantities} that can
be guessed efficiently. Whenever we say that we ``guess a value $x$
in time $O(f(n))$'' for some function $f$ we mean that in time
$O(f(n))$ we can compute a set (which hence has size $O(f(n))$)
that contains $x$. In the proof of the following lemma we use smoothing
techniques due to~\cite{DBLP:conf/focs/Batra0K18}. 
\begin{lem}
\label{lem:guessA} There are values $\left\{ B^{\mathrm{round}}(C)\right\} _{C\in\C}$
with the following properties 
\begin{itemize}
\item $B^{\mathrm{round}}(C)\ge B^{\opt}(C)$ for all $C$, 
\item $\sum_{C\in\C}B^{\mathrm{round}}(C)\le(1+\epsilon)\sum_{c\in\C}B^{\opt}(C)$, 
\item for each path $Q\in\Q$ we can guess in time $(nP)^{O_{\epsilon}(1)}$
\aw{all values $\left\{ B^{\mathrm{round}}(C)\right\} _{C\in Q}$}. 
\end{itemize}
\end{lem}

\begin{proof}
We set 
\[
B'(C)=B^{\opt}(C)+\hspace{-2em} \sum_{C':v_{C'}\text{ is ancesor of }v_{C}}\hspace{-2em} B^{\opt}(C')\cdot\left(\frac{\epsilon}{K}\right)^{\mathrm{dist}(v_{C},v_{C'})} .
\]
Then for each $C\in\C$ we set $B^{\mathrm{round}}(C)=(1+\epsilon)^{k}$
where $k\in\mathbb{Z}$ with $(1+\epsilon)^{k-1}<B'(C)\le(1+\epsilon)^{k}$.
By the construction it is obvious that $B^{\mathrm{round}}(C)\ge B^{\opt}(C)$
for all $C\in\C$. Moreover, note that each cell $C$ has at most
$K^{i}$ descendants $C'$ with $\mathrm{dist}(C,C')=i$. Hence 
\[
\sum_{C\in\C}B'(C)\le\sum_{C\in\C}B^{\opt}(C)\cdot\sum_{i=0}^{\infty}K^{i}\left(\frac{\epsilon}{K}\right)^{i}\le\sum_{C\in\C}B^{\opt}(C)\cdot\sum_{i=0}^{\infty}\epsilon^{i}\le(1+2\epsilon)\sum_{C\in\C}B^{\opt}(C).
\]
Finally, 
\[
\sum_{C\in\C}B^{\mathrm{round}}(C)\le(1+\epsilon)\sum_{C\in\C}B'(C)\le(1+4\epsilon)\sum_{C\in\C}B^{\opt}(C).
\]
Let $\{C_{1},\dotsc,C_{k}\}=Q\in\Q$ with $C_{1}$ being the root.
Guessing $B^{\mathrm{round}}(C_{i})$ is equivalent to guessing $\lfloor\log_{1+\epsilon}(B'(C_{i}))\rfloor$.
Notice that we have $B'(C_{i})\ge\epsilon B'(C_{i-1})$ for all
$i=2,\dotsc,k$. It follows that 
\[
\log_{1+\epsilon}(B'(C_{i}))\ge\log_{1+\epsilon}(B'(C_{i-1}))+\log_{1+\epsilon}(\epsilon)\ge\log_{1+\epsilon}(B'(C_{i-1}))-O_\epsilon(1).
\]
In other words, we want to guess values $q_{1},\dotsc,q_{k}=O_{\epsilon}(\log(nP))$
with $k=O_{\epsilon}(\log(nP))$ such that for some $c \le O_{\varepsilon}(1)$ it holds that
$q_{i} > q_{i-1}-c$ for all $i$.
By transforming $p_{i}=q_{i}+i\cdot c$ we obtain the problem of
guessing values $p_{1},\dotsc,p_{k}=O_{\epsilon}(\log(nP))$ with
$k=O_{\epsilon}(\log(nP))$ such that $p_{i}>p_{i-1}$. This can
be done in time $2^{O_{\epsilon}(\log(nP))}=(nP)^{O_{\epsilon}(1)}$:
We guess for all values $v\in\{1,2,\dotsc,O_{\epsilon}(\log(nP))\}$
whether $p_{i}=v$ for some $i$. After guessing, the first such $v$
must be $p_{1}$, the second must be $p_{2}$, etc., because $p_{1}<p_{2}<\cdots<p_{k}$.
Hence the values of $p_{1},p_{2},\dotsc,p_{k}$ are fully determined. 
\end{proof}
Informally speaking, by this lemma we can assume that our algorithm
knows all values $\left\{ B^{\mathrm{round}}(C)\right\} _{C\in Q}$
when we consider a path $Q$. 

Next, we divide all rectangles into small and large according to their
cost compared to the total budget $B^{\mathrm{round}}(C)$ of their
cell $C$. Let $\delta>0$ be a constant (depending only on $\epsilon$)
to be defined later. Consider a cell $C$ and a job $j$ with $\R(j,C)\ne\emptyset$.
We say that $j$ \emph{is large for $C$} if for the leftmost rectangle
$R\in\R(j,C)$ it holds that $c_{R}>\delta\cdot B^{\mathrm{round}}(C)$,
and $j$ \emph{is small for }$C$ if $c_{R}\le\delta\cdot B^{\mathrm{round}}(C)$.
We define $\Rl:=\bigcup_{C}\bigcup_{j:j\,\mathrm{is\,large\,for\,}C}\R(j,C)$
and $\Rs:=\R\setminus\Rl$.
\lr{Note that since the leftmost rectangle is always the most expensive one,
 for every small job $j$, all rectangles $R\in\R(j, C)$ satisfy $c_R\le \delta B^{\mathrm{round}}(C)$.}

Intuitively, we will prove Lemma~\ref{lem:cheap-consistent-solution}
separately for $\Rl$ and $\Rs$ and argue afterwards that this yields
the complete proof of Lemma~\ref{lem:cheap-consistent-solution}.
More precisely, for these sets we will provide families $\left\{ \chi_{\l,Q}\right\} _{Q\in\Q}$,
$\left\{ \chi_{\sm,Q}\right\} _{Q\in\Q}$ for which there exist consistent
solutions that dominate \emph{$\R^{*}$ }on $\Rl,\Rs$ according to
the next definition. 
\begin{defn}
\label{def:framework-split} Let $\R_{\mathrm{subset}}\subseteq\R$
and let \emph{$\left\{ \chi_{Q}\right\} _{Q\in\Q}$} be a family with
$\chi_{Q}\subseteq2^{\R(Q)}$ for each $Q\in\Q$. \lr{We say $\R'\subseteq\R$ is a}
\emph{solution for $\left\{ \chi_{Q}\right\} _{Q\in\Q}$} \emph{that dominates
$\R^{*}$ on $\R_{\mathrm{subset}}$} if
\begin{enumerate}
\item $\R'_{Q}\subseteq\R'$ for each $Q\in\Q$, \label{en:framework-subset-1} 
\item for each $I\in\I,Q\in\Q$ with $\R(I)\subseteq\R(Q)$, we have that
$\R'_{Q}\cap\R_{\mathrm{subset}}$ covers as much of $I$ as $\R^{*}\cap\R_{\mathrm{subset}}$,
i.e., $p(\R'_{Q}\cap\R_{\mathrm{subset}}\cap\R(I))\ge p(\R^{*}\cap\R_{\mathrm{subset}}\cap\R(I))$,
\label{en:framework-cover-1} 
\item for any two paths $Q,Q'\in\Q$ with $Q\supseteq Q'$ we have that
$\R'_{Q}\cap\R(Q')\subseteq\R'_{Q'}$. \label{en:framework-paths-1} 
\end{enumerate}
We define $c(\R')$ to be the \emph{cost} of $(\R',\left\{ \R'_{Q}\right\} _{Q\in\Q})$. 
\end{defn}

Notice that in the definition of the cost, we do not take the intersection
with $\R_{\mathrm{subset}}$. The reason is intuitively that
$\Rl$ and $\Rs$ are not known upfront, since they depend on the
unknown values $\left\{ B^{\mathrm{round}}(C)\right\} _{C\in\C}$.
Hence, when we select a set $\R'_{Q}$ then we need to pay for all
its rectangles and cannot, e.g., take the intersection with $\Rl$
or $\Rs$. We will prove the following two lemmas in Sections~\ref{subsec:small-rectangles}
and~\ref{subsec:large-rectangles}. 
\begin{lem}
\label{lem:small-jobs}In time $(nP)^{O_{\epsilon}(1)}$ we can
compute a family $\left\{ \chi_{\sm,Q}\right\} _{Q\in\Q}$ with $\max_{Q\in\Q}|\chi_{\sm,Q}|\le(nP)^{O_{\epsilon}(1)}$
for which there exists a consistent solution $(\R'_{\sm},\left\{ \R'_{\sm,Q}\right\} _{Q\in\Q})$
that that dominates $\R^{*}$ on $\R_{\sm}$ and has cost at most
\lr{$(2+O(\epsilon))\cdot c(\R^{*}\cap\Rs) + O(K^8\delta + \varepsilon) \cdot c(\R^*)$.}
\end{lem}
\begin{lem}
\label{lem:large-jobs}In time $(nP)^{O_{\epsilon,\delta}(1)}$ we
can compute a family $\left\{ \chi_{\l,Q}\right\} _{Q\in\Q}$ with
$\max_{Q\in\Q}|\chi_{\l,Q}|\le(nP)^{O_{\epsilon,\delta}(1)}$ for
which there exists a consistent solution $(\R'_{\l},\left\{ \R'_{\l,Q}\right\} _{Q\in\Q})$
that dominates $\R^{*}$ on $\Rl$ and has cost at most $2\cdot c(\R^{*}\cap\Rl)+\lr{O(\epsilon)}\cdot c(\R^{*})$. 
\end{lem}

Together these lemmas imply Lemma~\ref{lem:cheap-consistent-solution}. 
\begin{proof}[Proof of Lemma~\ref{lem:cheap-consistent-solution}]
We define $\delta:=\epsilon/K^{8}$. We first compute the families
$\left\{ \chi_{\l,Q}\right\} _{Q\in\Q}$, $\left\{ \chi_{\sm,Q}\right\} _{Q\in\Q}$
in time $(nP)^{O_{\epsilon}(1)}$ according to Lemmas~\ref{lem:large-jobs}
and~\ref{lem:small-jobs}. For each $Q\in\Q$ we define that $\chi_{Q}$
contains the set $\R'_{\l,Q}\cup\R'_{\sm,Q}$ for each combination
of a set $\R'_{\l,Q}\in\chi_{\l,Q}$ and a set $\R'_{\sm,Q}\in\chi_{\sm,Q}$.
Then $|\chi_{Q}|\le(nP)^{O_{\epsilon}(1)}$ as required. Then Lemmas~\ref{lem:large-jobs}
and \ref{lem:small-jobs} imply that for $\left\{ \chi_{Q}\right\} _{Q\in\Q}$
there exists a consistent solution of cost at most $(2+O(\epsilon))c(\R^{*})=(2+O(\epsilon))\OPT^{\mathrm{(IP2)}}$. 
\end{proof}

\subsection{\label{subsec:small-rectangles}Consistent solution for small rectangles}

This section is dedicated to proving Lemma~\ref{lem:small-jobs}.
Recall that in the quasi-polynomial construction, for defining a set
$\R'_{Q}$ we guessed the $O_{\epsilon}(\log^{2}(nP))$ values $\{B^{\opt}(C,\tau,s')\}_{C,\tau,s':C\in Q}$.
Since in this section we focus on small rectangles, let $B_{\sm}^{\opt}(C,\tau,s')$
denote the budgets that $\R^{*}$ spends on buying the first $s'$
rectangles in $\R(C)$ of type $\tau$. In our polynomial time procedure
we want to guess only $O_{\epsilon}(\log(nP))$ of them for each path
$Q$. The strategy for this is to increase the budgets slightly. Intuitively,
in each cell $C$ we want to spend more budget than $\R^{*}$ on rectangles
from sets $\R(j,C)$ with good densities $\rho_{j,C},\rho'_{j,C}$
and hence select more such rectangles. It will turn out that these
additional rectangles cover as much of each line segment $L(I)$ as
the rectangles from sets $\R(j',C')$ with \emph{bad }cost-efficiencies
$\rho_{j',C'},\rho'_{j',C'}$ where $v_{C'}$ is an ancestor of $v_{C}$.
Therefore, in $\R'_{Q}$ we do not need rectangles from the sets $\R(C',\tau,s')$
for of types $\tau$ with such bad cost-efficiencies. Thus we do not
need to guess the corresponding budget $\{B_{\sm}^{\opt}(C',\tau,s')\}_{s'}$
when we define $\R'_{Q}$. Additionally, for each cell $C$ the selection
of very efficient rectangles will simplify drastically. Namely, we
select all rectangles that have a very low density $\rho$ until some
given budget is exhausted. Thus, for the types $\tau$ corresponding
to these low densities we do need to guess all corresponding budgets
$\{B_{\sm}^{\opt}(C',\tau,s')\}_{s'}$ but only one single value that
describes the total budget used for all of them. We remark that a
similar strategy was used in~\cite{DBLP:conf/focs/Batra0K18}.

To this end, for each cell $C$ we define an additional budget $B^{\mathrm{add}}(C)$.
Intutively, each cell $C$ donates a budget of $\epsilon\cdot B_{\sm}^{\opt}(C)$
to its descendent cells and the received amount of each descendent
cell drops exponentially. From these donations, each cell $C$ receives
an additional budget of 
\begin{equation}
B^{\mathrm{add}}(C):=\hspace{-2em}\sum_{C':v_{C'}\,\mathrm{is\,ancestor\,of}\,v_{C}}\hspace{-2em}B^{\mathrm{round}}(C')\cdot\left(\frac{\epsilon}{K}\right)^{\dist(v_{C},v_{C'})},\label{eq:B-add-definition}
\end{equation}
where $\dist(v_{C},v_{C'})$ is the distance of $v_{C}$ and $v_{C'}$.
In fact, we used a similar procedure in the proof of Lemma~\ref{lem:guessA}. 
\begin{lem}
We have that $\sum_{C\in\C}B^{\mathrm{add}}(C)\le2\epsilon\sum_{C\in\C}B^{\mathrm{round}}(C)$. 
\end{lem}

\begin{proof}
Notice that each cell $C$ has at most $K^{i}$ descendants $C'$
with $\mathrm{dist}(C,C')=i$. Hence 
\[
\sum_{C\in\C}B^{\mathrm{add}}(C)\le\sum_{C\in\C}B^{\mathrm{round}}(C)\cdot\sum_{i=1}^{\infty}K^{i}\left(\frac{\epsilon}{K}\right)^{i}\le\sum_{C\in\C}B^{\mathrm{round}}(C)\cdot\sum_{i=1}^{\infty}\epsilon^{i}\le2\epsilon\sum_{C\in\C}B^{\mathrm{round}}(C).\qedhere
\]
\end{proof}
For each cell $C$ we partition $B^{\mathrm{add}}(C)$ equally among
the pairs $s,s'\in\{1,\dotsc,K^{2}\}$ with $s'\le s$. To this end,
we define $B^{\mathrm{add}}(C,s,s'):=B^{\mathrm{add}}(C)/K^{4}$ for
each such pair $s,s'$.

We sort the density pairs $(\rho,\rho')$ lexicographically and we
write $(\rho,\rho')<_{L}(\hat{\rho},\hat{\rho}')$ if $(\rho,\rho')$
is lexicographically smaller than $(\hat{\rho},\hat{\rho}')$. Our
strategy is to define a \emph{critical density pair }$\gamma(C,s,s')$
for each cell $C\in\C$ and each pairs of values $s,s'\in\{1,\dotsc,K^{2}\}$
with $s'\le s$. In our solution $\R'_{\sm}$, intuitively for each
cell $C\in\C$ and each value $s\in\{1,\dotsc,K^{2}\}$ select rectangles
as follows:\awtodo{Added this intuition} 
\begin{enumerate}
\item for each $s'\le s$ and each type $\tau=(\rho,\rho',s)$ with $(\rho,\rho')<_{L}\gamma(C,s,s')$
we select the first $s'$ rectangles from \emph{each} set $\R(j,C)$
of type $\tau$ such that $j$ is small for $C$; one can imagine
that we have infinite budget for these types $\tau$ and these values
of $s'$ when we run GreedySelect, 
\item all other rectangles of each type $\tau=(\rho,\rho',s)$ are selected
with GreedySelect and budgets $B_{\sm}^{\opt}(C,\tau,s')$; one special
case is the type $\tau=(\rho,\rho',s)$ with $(\rho,\rho')=\gamma(C,s,s')$
for each value $s'$, where we use an increased budget of $B^{\mathrm{round}}(C,\tau,s')+B^{\mathrm{add}}(C,s,s')$
in order to select some additional rectangles of relatively good densities. 
\end{enumerate}
Note that for the selections due to step 1, for a cell $C$ we need
to know only the $O_{\epsilon}(1)$ values $\gamma(C,s,s')$. We will
ensure that for a given path $Q\in\Q$ we can guess these values for
each cell $C\in Q$ in time $(nP)^{O_{\epsilon}(1)}$. For selecting
rectangles due to step 2, we need to know the budgets $B_{\sm}^{\opt}(C,\tau,s')$.
We will not guess them exactly, but sufficiently good estimates that
we will denote by $B_{\sm}^{\r}(C,\tau,s')$. For a path $Q\in\Q$
these can be up to $O_{\epsilon}(\log^{2}(nP))$ values which is too
much. Therefore, when we define $\R'_{Q}$ we omit the rectangles
of some types $\tau$ in some cells $C\in Q$. More precisely, we
will make sure that in total there are only $O_{\epsilon}(\log(nP))$
combinations of a cell $C\in Q$ and a type $\tau$ for which we add
rectangles to $\R'_{Q}$ in step~2. Intuitively, these omitted rectangles
will be compensated by additional rectangles with relatively good
densities that we select when we use the additional budget of $B^{\mathrm{add}}(C,s,s')$
above. We will ensure that we can guess the $O_{\epsilon}(\log(nP))$
needed values $B_{\sm}^{\r}(C,\tau,s')$ in time $(nP)^{O_{\epsilon}(1)}$.
Note that we can guess the needed values $B^{\mathrm{add}}(C,s,s')$
easily by Lemma~\ref{lem:guessA} and \eqref{eq:B-add-definition}.

\subsubsection{Definition of critical density pairs}

For defining the critical density pairs $\gamma(C,s,s')$, we define
a procedure GreedyIncrease. Given a cell $C$, a value $s$, the budgets
$\{B_{\sm}^{\opt}(C,(\rho,\rho',s),s')\}_{\rho,\rho',s'}$, and the
additional budgets $B^{\mathrm{add}}(C,s,s')$, GreedyIncrease sorts
the density pairs $(\rho,\rho')$ in lexicographically increasing
order. For each of these pairs $(\rho,\rho')$ it runs a slight variation
of GreedySelect in order to select rectangles of type $(\rho,\rho',s)$.
To this end, we define a procedure GreedySelectSmall$(C,B^{\mathrm{round}}(C),\tau,\{B_{S}(s')\}_{s'})$:
this procedure first identifies all jobs $j$ that are small for $C$
according to the given value $B^{\mathrm{round}}(C)$, and then on
them it runs GreedySelect$(C,\tau,\{B_{S}(s')\}_{s'})$ as defined
in Section~\ref{subsec:quasi-polynomial-size-consistent}. 

In the iteration for the density pair $(\rho,\rho')$ we call GreedySelectSmall
as follows: for each $s'\le s$ let $\hat{B}_{s'}$ denote the total
cost spent on buying exactly the first $s'$ rectangles of jobs. For
each $s'\le s$, 
\begin{enumerate}
\item if $\hat{B}_{s'}<B^{\mathrm{add}}(C,s,s')+\sum_{(\hat{\rho},\hat{\rho}')<_{L}(\rho,\rho')}B_{\sm}^{\opt}(C,(\hat{\rho},\hat{\rho}',s),s')$
(i.e., in this case we have not spent much more than $\R^{*}$ on
buying the first $s'$ rectangles of the jobs of densities $(\hat{\rho},\hat{\rho}')<_{L}(\rho,\rho')$)
then we define $B(s'):=B_{\sm}^{\opt}(C,(\rho,\rho',s),s')+B^{\mathrm{add}}(C,s,s')$,
\item if $\hat{B}_{s'}\ge B^{\mathrm{add}}(C,s,s')+\sum_{(\hat{\rho},\hat{\rho}')<_{L}(\rho,\rho')}B_{\sm}^{\opt}(C,(\hat{\rho},\hat{\rho}',s),s')$
then we define $B(s'):=B_{\sm}^{\opt}(C,(\rho,\rho',s),s')$.
\end{enumerate}
Then we call GreedySelectSmall$(C,B^{\mathrm{round}}(C),(\rho,\rho',s),\{B_{s'}\}_{s'})$.
For each $s'\le s$ we define $\gamma(C,s,s')$ to be the lexicographically
largest density pair $(\rho,\rho'$) such that case~1 applies for
$s'$ for type $(\rho,\rho',s)$;\lr{ in the corner case that
case~2 never applies we set $\gamma(C,s,s') = (\infty,\infty)$.}
Note that in the iteration of each
type $(\rho,\rho',s)$ with $(\rho,\rho')\le_{L}\gamma(C,s,s')$ case~1 applies for $s'$,
and in the iteration of each type $(\rho,\rho',s)$ with $(\rho,\rho')>_{L}\gamma(C,s,s')$
case~2 applies for $s'$.

 Now an important observation is that in the calls to GreedySelectSmall
we could change the value of $B(s')$ from $B_{\sm}^{\opt}(C,(\rho,\rho',s),s')+B^{\mathrm{add}}(C,s,s')$
to $\infty$ if $(\rho,\rho')<_{L}\gamma(C,s,s')$ without affecting
the returned rectangles. Using this idea it will turn out that we
do not require to know the respective value $B_{\sm}^{\opt}(C,(\rho,\rho',s),s')$
in those cases. As we still need to guess some of the remaining values,
we will define a guessing scheme for density pairs $(\rho,\rho')\ge_{L}\gamma(C,s,s')$.
More precisely, instead of a value $B_{\sm}^{\opt}(C,\tau,s')$ we
will use an estimate $B_{\sm}^{\mathrm{round}}(C,\tau,s')$ which
we will guess, using the following lemma.
\begin{lem}
\label{lem:guessB} Let $C\in\C$. Define $\Gamma(C)=\{\gamma(C,s,s'):s,s'\}$.
There are values $B_{\sm}^{\mathrm{round}}(C,\tau,s')$ for each $\tau$,
and $s'$ with the following properties
\begin{itemize}
\item for each $C,\tau,s'$ it holds that $B_{\sm}^{\mathrm{round}}(C,\tau,s')\ge B_{\sm}^{\opt}(C,\tau,s')$, 
\item for each $C,s,s'$ and each $(\rho_{0},\rho'_{0})\in\Gamma(C)$ it
holds that 
\[
\sum_{(\rho,\rho')\ge_{L}(\rho_{0},\rho'_{0})}\hspace{-2em}B_{\sm}^{\mathrm{round}}(C,(\rho,\rho',s),s')\le\frac{\epsilon}{K^{4}}B^{\mathrm{round}}(C)+(1+\epsilon)\hspace{-2em}\sum_{(\rho,\rho')\ge_{L}(\rho_{0},\rho'_{0})}\hspace{-2em}B_{\sm}^{\opt}(C,(\rho,\rho',s),s'),
\]
\item given $B_{\sm}^{\mathrm{round}}(C)$, $(\rho_{0},\rho'_{0})\in\Gamma(C)$,
and $\ell\in\mathbb{N}$, we can guess the values $B_{\sm}^{\mathrm{round}}(C,(\rho,\rho',s),s')$
for the first $\ell$ many (with respect to lexicographic order) density
pairs $(\rho,\rho')$ satisfying $(\rho,\rho')\ge_{L}(\rho_{0},\rho'_{0})$
in time $2^{O_{\epsilon}(\ell)}$. 
\end{itemize}
\end{lem}

\begin{proof}
Let $(\rho^{(1)},{\rho'}^{(1)}),(\rho^{(2)},{\rho'}^{(2)}),\dotsc,(\rho^{(\lr{m}-1)},\rho'^{(\lr{m}-1)})$
be the density pairs in $\Gamma(C)$ in lexicographically increasing
order. Moreover, let $(\rho^{(0)},{\rho'}^{(0)})=(-\infty,-\infty)$
and $(\rho^{(m)},{\rho'}^{(m)})=(\infty,\infty)$. Intuitively
we will prove the lemma independently on each set of density
pairs that for some given $h$ contains all densities $(\rho,\rho')$
with $(\rho^{(h)},{\rho'}^{(h)})\le_{L}(\rho,\rho')<_{L}(\rho^{(h+1)},{\rho'}^{(h+1)})$.

More precisely, for each $h$ we will prove that there are
values $B_{\sm}^{\mathrm{round}}(C,(\rho,\rho',s),s')$ for all $s,s'$
and $(\rho,\rho')$ with $(\rho^{(h)},{\rho'}^{(h)})\le_{L}(\rho,\rho')<_{L}(\rho^{(h+1)},{\rho'}^{(h+1)})$
such that: 
\begin{enumerate}
\item for each $s,s'$ and $(\rho,\rho')$, with $(\rho^{(h)},{\rho'}^{(h)})\le_{L}(\rho,\rho')<_{L}(\rho^{(h+1)},{\rho'}^{(h+1)})$
it holds that
\begin{equation*}
  B_{\sm}^{\mathrm{round}}(C,(\rho,\rho',s),s')\ge B_{\sm}^{\opt}(C,(\rho,\rho',s),s'),
\end{equation*}
\item for each $s,s'$ it holds that 
\[
\sum_{(\rho^{(h)},{\rho'}^{(h)})\le_{L}(\rho,\rho')<_{L}(\rho^{(h+1)},{\rho'}^{(h+1)})}\hspace{-5em}B_{\sm}^{\mathrm{round}}(C,(\rho,\rho',s),s')\le\frac{\epsilon}{K^{8}}B^{\mathrm{round}}(C)+(1+\epsilon)\hspace{-5em}\sum_{(\rho^{(h)},{\rho'}^{(h)})\le_{L}(\rho,\rho')<_{L}(\rho^{(h+1)},{\rho'}^{(h+1)})}\hspace{-5em}B_{\sm}^{\opt}(C,(\rho,\rho',s),s'),
\]
\item given $B_{\sm}^{\mathrm{round}}(C)$ and $\ell\in\mathbb{N}$, we
can guess the values $B_{\sm}^{\mathrm{round}}(C,(\rho,\rho',s),s')$
for the first $\ell$ many (with respect to lexicographic order) density
pairs $(\rho,\rho')$ satisfying $(\rho^{(h)},{\rho'}^{(h)})\le_{L}(\rho,\rho')<_{L}(\rho^{(h+1)},{\rho'}^{(h+1)})$ in time $2^{O_{\epsilon}(\ell)}$.
\end{enumerate}
In the last property we assume that the $\ell$-th density pair is
still lexicographically smaller than $(\rho^{(h+1)},{\rho'}^{(h+1)})$.
Using $|\Gamma(C)|\le K^{4}$ it is not hard to see that the claim
above implies the lemma.

Fix some $h$. Let $(\rho_{0},\rho'_{0}),(\rho_{1},\rho'_{1}),\dotsc,(\rho_{k},\rho'_{k})$
denote the density pairs $(\rho,\rho')$ with $(\rho^{(h)},{\rho'}^{(h)})\le_{L}(\rho,\rho')<_{L}(\rho^{(h+1)},{\rho'}^{(h+1)})$
in lexicographically increasing order starting with $(\rho_{0},\rho'_{0})=(\rho^{(h)},{\rho'}^{(h)})$.
For each $i=0,\dotsc,k$ we set 
\[
B'(C,(\rho_{i},\rho'_{i},s),s')=\left(\frac{\epsilon}{4K^{8}}\right)^{i+1}B^{\mathrm{round}}(C)+\sum_{j=0}^{i}\left(\frac{\epsilon}{4}\right)^{i-j}B_{\sm}^{\opt}(C,(\rho_{i},\rho'_{i},s),s').
\]
Moreover, define $B_{\sm}^{\mathrm{round}}(C,(\rho_{i},\rho'_{i},s),s')=(1+\epsilon/4)^{k}$
where $k\in\mathbb{Z}$ with 
\[
(1+\epsilon/4)^{k-1}<B'(C,(\rho_{i},\rho'_{i},s),s')\le(1+\epsilon/4)^{k}.
\]
Clearly $B_{\sm}^{\mathrm{round}}(C,(\rho_{i},\rho'_{i},s),s')\ge B_{\sm}^{\opt}(C,(\rho_{i},\rho'_{i},s),s')$
for all $i$. Moreover, we have 

\begin{align*}
\sum_{i=0}^{k}B'(C,(\rho_{i},\rho'_{i},s),s') & \le B^{\mathrm{round}}(C)\sum_{i=0}^{k}\left(\frac{\epsilon}{4K^{8}}\right)^{i+1}+\sum_{i=0}^{k}B_{\sm}^{\opt}(C,(\rho_{i},\rho'_{i},s),s')\sum_{j=0}^{\infty}\left(\frac{\epsilon}{4}\right)^{j}\\
 & \le\frac{\epsilon}{2K^{8}}B^{\mathrm{round}}(C)+(1+\epsilon/2)\sum_{i=0}^{k}B_{\sm}^{\opt}(C,(\rho_{i},\rho'_{i},s),s').
\end{align*}
This implies that 
\begin{align*}
\sum_{i=0}^{k}B_{\sm}^{\mathrm{round}}(C,(\rho_{i},\rho'_{i},s),s') & \le(1+\epsilon/4)\sum_{i=0}^{k}B'(C,(\rho_{i},\rho'_{i},s),s')\\
 & \le\frac{\epsilon}{K^{8}}B^{\mathrm{round}}(C)+(1+\epsilon)\sum_{i=0}^{k}B_{\sm}^{\opt}(C,(\rho_{i},\rho'_{i},s),s').
\end{align*}
Guessing the first $\ell$ values of $B^{\mathrm{round}}(C,(\rho_{i},\rho'_{i},s),s')$
is equivalent to guessing the first $\ell$ values of 
\[
\lfloor\log_{1+\epsilon/4}(B'(C,(\rho_{i},\rho'_{i},s),s')\rfloor.
\]
Notice that 
\[
B'(C,(\rho_{i},\rho'_{i},s),s')\le\epsilon/4\cdot B^{\mathrm{round}}(C)+B_{\sm}^{\opt}(C,(\rho_{i},\rho'_{i},s),s')\le(1+\epsilon/4)B^{\mathrm{round}}(C)
\]
for all $i=1,\dotsc,\ell$ and $B'(C,(\rho_{0},\rho'_{0},s),s')\ge\epsilon/(4K^{4})\cdot B^{\mathrm{round}}(C)$.
Moreover, from the definition of $B'$ it follows easily that $B'(C,(\rho_{i+1},\rho'_{i+1},s),s')\ge\epsilon B'(C,(\rho_{i},\rho'_{i},s),s')$
for all $i=1,\dotsc,k-1$. This implies that 
\begin{align*}
\lfloor\log_{1+\epsilon/4}(B'(C,(\rho_{i+1},\rho'_{i+1},s),s')\rfloor & \ge\lfloor\log_{1+\epsilon/4}(B'(C,(\rho_{i},\rho'_{i},s),s')+\log_{1+\epsilon}(\epsilon)\rfloor\\
 & \ge\lfloor\log_{1+\epsilon/4}(B'(C,(\rho_{i},\rho'_{i},s),s')\rfloor-O_{\epsilon}(1).
\end{align*}
\lr{
Define $q_{i}:=\lfloor\log_{1+\epsilon/4}(B'(C,(\rho_{i},\rho'_{i},s),s')\rfloor$ for each $i$.
Then there exists some $c = O_\epsilon(1)$ such that $q_{i+1} > q_{i}-\lr{c}$ for each $i$.
Moreover, $q_1 \ge \log(B^{\mathrm{round}}(C)) - O_{\epsilon}(1)$
and $q_i \le \log(B^{\mathrm{round}}(C)) + O_{\epsilon}(1) \le q_1 + O_{\epsilon}(1)$ for each $i$.
Guessing each $q_i$ is equivalent to guessing $q_1$ and each $p_{i}:=q_{i} - q_1 +c\cdot i$.
The former can be done in time $O_{\epsilon}(1)$.
For the latter observe that $p_{i+1}>p_{i}$
and $0 \le p_{i}\le O_{\epsilon}(\ell)$ for each $i$. Therefore, we can
guess all values $p_{i}$ in time $2^{O_{\epsilon}(\ell)}$ as follows.
First, we guess in time $2^{O_{\epsilon}(\ell)}$ which values in
$\{0,\dotsc,O_{\epsilon}(\ell)\}$ are attained by some $p_{i}$.
The values of $p_{1},p_{2},\dotsc,p_{\ell}$ are then fully defined
then: the first value that is attained must be $p_{1}$, the second
$p_{2}$, etc. Then the values $\left\{ p_{i}\right\} _{i}$ imply
the values $B^{\mathrm{round}}(C,(\rho,\rho',s),s')$.}
\end{proof}
Our strategy now is the following: we define a solution $\R'_{\sm}$
with a procedure that is very similar to GreedyIncrease above. The
main difference is that we use the values $B_{\sm}^{\mathrm{round}}(C,\tau,s')$
instead of the values $B_{\sm}^{\opt}(C,\tau,s')$. Then, for each
$Q\in\Q$ we define a solution $\R'_{\sm,Q}$ for which we first guess
the values $B^{\mathrm{round}}(C)$ and $\gamma(C,s,s')$ for each
$C\in Q$ and then define a value $\ell_{C}$ for each cell $C\in Q$
such that $\sum_{C\in Q}\ell(C)\le O_{\epsilon}(\log nP)$. Then we
will guess each value $B_{\sm}^{\mathrm{round}}(C,\tau,s')$ for each
cell $C\in Q$, each $s,s'$, and the first $\ell_{C}$ types $\tau=(\rho,\rho',s)$
(i.e., the $\ell_{C}$ lexicographically smallest types) with $(\rho,\rho')\ge_{L}\gamma(C,s,s')$.
For these types we add rectangles to $\R'_{\sm,Q}$ via GreedySelectSmall.
Due to Lemma~\ref{lem:guessB} we can guess all needed quantities
in time $\prod_{C\in Q}2^{O_{\epsilon}(\ell_{C})}\le(nP)^{O_{\epsilon}(1)}$. 

\subsubsection{Definition of solution $\protect\R'_{\protect\sm}$}

We now define the solution $\R'_{\sm}$ based on the values $B^{\mathrm{round}}(C)$,
$\gamma(C,s,s')$, and $B_{\sm}^{\mathrm{round}}(C,\tau,s')$ defined
above. We will produce a solution similar to the one returned by GreedyIncrease;
the main difference being that we will use the budgets $B_{\sm}^{\mathrm{round}}(C,\tau,s')$
instead of the budgets $B_{\sm}^{\opt}(C,\tau,s')$. Recall that in
GreedyIncrease, for a value $s$, a density pair $(\rho,\rho')$,
and a value $s'$ with $(\rho,\rho')<_{L}\gamma(C,s,s')$, we could
change the budget $B(s')$ to $\infty$ without affecting the returned
rectangles.

Formally, for each cell $C$, each $s$, and each density pair $(\rho,\rho')$
and we add to $\R'_{\sm}$ the rectangles returned by GreedySelectSmall$(C,B^{\mathrm{round}}(C),(\rho,\rho',s),B(s'))$
where 
\[
B(s')=\begin{cases}
\infty & \text{if }(\rho,\rho')<_{L}\gamma(C,s,s'),\\
B_{\sm}^{\mathrm{round}}(C,(\rho,\rho',s),s')+B^{\mathrm{add}}(C,s,s') & \text{if }(\rho,\rho')=\ \ \gamma(C,s,s'),\\
B_{\sm}^{\mathrm{round}}(C,(\rho,\rho',s),s') & \text{if }(\rho,\rho')>_{L}\gamma(C,s,s').
\end{cases}
\]

\begin{lem}
It holds that $c(\R'_{\sm})\le(2+O(\epsilon))c(\R^{*}\cap\R_{\sm})+O(K^{\aw{8}}\delta+\epsilon)c(\R^{*})$. 
\end{lem}

\begin{proof}
For a cell $C$, all $s,s'$, and all $(\rho,\rho')<_{L}\gamma(C,s,s')$
we define $B^{\mathrm{actual}}(C,(\rho,\rho',s),s')$ as the cost
of rectangles in $\R'_{\sm}$ which belong to some $\R(j,C)$ of density
pair $(\rho,\rho')$, with $|\R(j,C)|=s$, and of which exactly the
first $s'$ rectanges are selected. By the definition of $\gamma(C,s,s')$
we have that
\[
\sum_{(\rho,\rho')<_{L}\gamma(C,s,s')}\hspace{-2em}B^{\mathrm{actual}}(C,(\rho,\rho',s),s')<\hspace{-2em}\sum_{(\rho,\rho')<_{L}\gamma(C,s,s')}\hspace{-2em}B_{\sm}^{\opt}(C,(\rho,\rho',s),s')+B^{\mathrm{add}}(C,s,s')
\]
for all $C,s,s'$. Moreover, from the definition of GreedySelectSmall,
one can observe easily that changing the value $B(s')$ from
$\infty$ to $B^{\mathrm{actual}}(C,(\rho,\rho',s),s')$ does not
change the outcome in the construction of $\R'_{\sm}$: If it did,
then the call to GreedySelectSmall with the original parameters
would have spent more than $B^{\mathrm{actual}}(C,(\rho,\rho',s),s')$,
a contradiction. We can therefore calculate the cost of $\R'_{\sm}$
as if it was defined with these modified parameters. By an easy modification
of Lemma~\ref{lem:greedyselect} one can show that the cost for each
$C$,$\tau$, and $s$ that GreedySelectSmall spends is at most 
\[
\min\{(1+\epsilon)\sum_{s'}B_{S}(s')+K^{4}\delta B^{\mathrm{round}}(C),(2+\epsilon)\sum_{s'}B(s')\},
\]
where $B(s')$ are the budgets passed to GreedySelectSmall, since
the fractional solution computed by GreedySelectSmall has a cost of
$(1+\epsilon)\sum_{s'}B(s')$ and afterwards we round up at most $K^{4}$
variables corresponding to at most $K^{2}$ small jobs, which incurs
an additional cost of at most $K^{4}\delta B^{\mathrm{round}}(C)$.
This implies \begingroup \allowdisplaybreaks 
\begin{align*}
c(\R'_{\sm}) & \le\sum_{C,s,s'}\bigg(\sum_{(\rho,\rho')<_{L}\gamma(C,s,s')}\hspace{-2em}B^{\mathrm{actual}}(C,(\rho,\rho',s),s')\\
 & \qquad\qquad+(1+\epsilon)(B_{\sm}^{\mathrm{round}}(C,(\gamma(C,s,s'),s),s')+B^{\mathrm{add}}(C,s,s'))+K^{4}\delta B^{\mathrm{round}}(C)\\
 & \qquad\qquad+(2+\epsilon)\hspace{-2em}\sum_{(\rho,\rho')>_{L}\gamma(C,s,s')}\hspace{-2em}B_{\sm}^{\mathrm{round}}(C,(\rho,\rho',s),s')\bigg).\\
 & \le\sum_{C,s,s'}\bigg(\sum_{(\rho,\rho')<_{L}\gamma(C,s,s')}B^{\mathrm{actual}}(C,(\rho,\rho',s),s')\\
 & \qquad\qquad+(2+\epsilon)\sum_{(\rho,\rho')\ge_{L}\gamma(C,s,s')}B_{\sm}^{\mathrm{round}}(C,(\rho,\rho',s),s')\\
 & \qquad\qquad+(1+\epsilon)B^{\mathrm{add}}(C,s,s')+K^{4}\delta B^{\mathrm{round}}(C)\bigg)\\
 & \le\sum_{C,s,s'}\bigg(\sum_{(\rho,\rho')<_{L}\gamma(C,s,s')}B_{\sm}^{\opt}(C,(\rho,\rho',s),s')\\
 & \qquad\qquad+(1+\epsilon)(2+\epsilon)\sum_{(\rho,\rho')\ge_{L}\gamma(C,s,s')}\hspace{-2em}B_{\sm}^{\opt}(C,(\rho,\rho',s),s')\\
 & \qquad\qquad+(2+\epsilon)B^{\mathrm{add}}(C,s,s')+(K^{4}\delta+\epsilon/K^{4})B^{\mathrm{round}}(C)\bigg).\\
 & \le(2+4\epsilon)c(\R^{*}\cap\R_{\sm})+2\epsilon(2+\epsilon)\sum_{C}B^{\mathrm{round}}(C)+(K^{4}\delta+\epsilon/K^{4})\sum_{C,s,s'}B^{\mathrm{round}}(C)\\
 & \le(2+4\epsilon)c(\R^{*}\cap\R_{\sm})+(K^{8}\delta+6\epsilon)\sum_{C}B^{\mathrm{round}}(C)\\
 & \le(2+4\epsilon)c(\R^{*}\cap\R_{\sm})+(K^{8}\delta+6\epsilon)(1+\epsilon)c(\R^{*})\\
 & \le(2+O(\epsilon))c(\R^{*}\cap\R_{\sm})+O(K^{8}\delta+\epsilon)c(\R^{*}).\qedhere
\end{align*}
\endgroup 
\end{proof}

\subsubsection{Definition of solutions $\protect\R'_{\protect\sm,Q}$}

We now define the solutions $\left\{ \R'_{\sm,Q}\right\} _{Q\in\Q}$.
Let $Q\in\Q$. Let $C_{1},C_{2},\dotsc,C_{\ell}$ be the cells in
$Q$ ordered by increasing distance from the root. Intuitively,
the set $\R'_{\sm,Q}$ contains all rectangles in $\R'_{\sm}\cap\R(Q)$
that we will use in order to satisfy the demand of intervals $I=[s,t]$
such that $t\in C_{\ell}$. Note that then $t\in C$ for each $C\in Q$.
The trick is that $\R'_{\sm,Q}$ does not necessarily need to contain
all segments in $\R'_{\sm}\cap\R(Q)$ (thanks to the additional budgets
$B^{\mathrm{add}}(C)$ from which we selected additional rectangles
with relatively good cost-efficiencies) but potentially only
a subset for which there will be only few options. In this way, we
will ensure that $|\chi_{Q}|\le(nP)^{O(1)}$.

We start with the rectangles in $\R'_{\sm}\cap\R(Q)$ and omit some
of them in the following. 
Consider some cell $C\in Q$, which is not one of the two bottom-most
cells, i.e., $C\notin\{C_{\ell},C_{\ell-1}\}$. For each $s$ we
define $r(s,C,Q)\le s$ such that for
any set $\R(j,C)$ with $|\R(j,C)|=s$ the $r(s,C,Q)$-th rectangle
intersects with $C_{\ell}$; we define $r(s,C,Q)=0$ if no rectangle
in $\R(j,C)$ intersects with $C_{\ell}$. Observe that this
is the same value for each such set $\R(j,C)$ and this set can contain
only only rectangle that intersects with $C_{\ell}$.

Consider a job $j$. Assume that $\R(j,C)\cap\R'_{\sm}$ contains
exactly the first $s'$ rectangles in $\R(j,C)$. If $s'<r(s,C,Q)$
then we omit all rectangles in $\R(j,C)\cap\R'_{\sm}$; if $s'\ge r(s,C,Q)$
then we omit all but the first $r(s,C,Q)$ rectangles from $\R(j,C)\cap\R'_{\sm}$.
We do this for each job $j$.

Next, we will omit rectangles of high densities under certain circumstances;
intuitively, because they have become obsolute due to other rectangles
that we selected additionally with the additional budgets $B^{\mathrm{add}}(C)$.
For every $C_{j}$ with $j<\ell-1$ we define $\sigma^{-}(C_{j},Q)$
as the lexicographically minimal density pair $\gamma(C_{j},s,r(s,C,Q))$
over all $s$ with $r(s,C,Q)\neq0$. Consider some $C_{i}$ that is
a ancestor of of $C_{j}$, i.e., $i<j$. Recall that $\R'_{\sm}$
has spent an additional budget of $B^{\mathrm{add}}(C_{j})/K^{4}$
compared to the optimal solution on rectangles in $\R(C_{j})$ of
densities $(\rho,\rho')\le_{L}\sigma^{-}(C_{j},Q)$. This can be used
to compensate for omitting all rectangles in $\R'_{\sm}\cap\R(C_{i})$
of density $(\rho,\rho')$ where $(\rho,\rho')\ge_{L}(K^{7}/\epsilon)^{j-i}\sigma^{-}(C_{j},Q)$
(multiplication component-wise). \lr{Therefore, we define $H = (1 + \epsilon)^k$ where
\begin{equation*}
  (1 + \epsilon)^k \le K^{7}/\epsilon < (1 + \epsilon)^{k+1}
\end{equation*}
and let $\sigma^{+}(C_{i},Q)$ be the lexicographically minimal density pair
$H^{j-i}\sigma^{-}(C_{j},Q)$ over all $j\in\{i+1,i+2,\dotsc,\ell-2\}$.} We omit from $\R'_{\sm}$
all rectangles in $\R'_{\sm}\cap\R(C_{i})$ of density $(\rho,\rho')\ge_{L}\sigma^{+}(C_{i},Q)$.

We define that $\R'_{\sm,Q}$ contains all rectangles in $\R'_{\sm}$
that are not omitted by the rules above. We apply the construction
above for each $Q\in\Q$. We will argue that this fulfills the three
properties of Definition~\ref{def:framework-split} on the small
rectangles $\R_{\sm}$. The first and third property are fulfilled
by construction. The latter holds since if $Q\supseteq Q'$ then
$\sigma^{+}(C,Q)\le_{L}\sigma^{+}(C,Q')$ for each $C\in Q'$ as in
$Q$ there are more values $\sigma^{-}(C_{j},Q)$ that can affect
$\sigma^{+}(C,Q)$. 
\begin{prop}
For each $Q\in\Q$ it holds that $\R'_{\sm,Q}\subseteq\R'_{\sm}$.
Also, 
for any two paths $Q,Q'\in\Q$ with $Q\supseteq Q'$ we have that
$\R'_{\sm,Q}\cap\R(Q')\subseteq\R'_{\sm,Q'}$. 
\end{prop}

In the next lemma we show that also~ the second property holds.
\begin{lem}
Let $Q\in\Q$ and $I\in\I$ with $\R(I)\subseteq\R(Q)$. Then 
\[
p(\R'_{\sm,Q}\cap\R_{\sm}\cap\R(I))\ge d(R^{*}\cap\R_{\sm}\cap\R(I)).
\]
\end{lem}
\lr{
\begin{proof}
Consider the solution $\R'_{\sm}$ and let $C\in Q$.
We have by Lemma~\ref{lem:greedyselect} (with straightforward adaption to GreedySelectSmall) for each type $\tau$ that
\[
p(\R'_{\sm}\cap\R_{\sm}(C,\tau)\cap\R(I))\ge p(\R^{*}\cap\R_{\sm}(C,\tau)\cap\R(I)),
\]
where $\R_{\sm}(C,\tau)$ are the small rectangles of type $\tau$ in $\R(C)$.
In particular, if we sum over all types, the solution $\R'_{\sm}$ covers $L(I)$ at least
as much as $\R^*\cap\R_{\sm}$.
If $L(I)$ fully traverses $C$ (the ray begins above $C$) we have, in fact, a slightly stronger relation.
Let  $({\rho}^{-},{\rho'}^{-})=\sigma^{-}(C,Q)$.
If $\rho^- < \infty$ then there is some $s$ such that in $\R'_{\sm}$
we payed an extra budget of $B^{\mathrm{add}}(C,s,r(s,C,Q)) = B^{\mathrm{add}}(C)/K^{4}$ compared
to $\R^{*}\cap \R_{\sm}$ for buying the first $r(s,C,Q)$ rectangles of some jobs
$j$ with $|\R(j,C)|=s$ that have density at most $({\rho}^{-},{\rho'}^{-})$.
These jobs cover at least volume of 
\[
\frac{1}{s'\rho^{-}}\cdot\frac{B^{\mathrm{add}}(C)}{K^{4}}\ge\frac{1}{K^{6}\rho^{-}}B^{\mathrm{add}}(C).
\]
If $\rho^- = \infty$ then we interpret the term as $0$. Hence it still holds. Thus
\[
p(\R'_{\sm}\cap\R_{\sm}(C)\cap\R(I))\ge p(\R^{*}\cap\R_{\sm}(C)\cap\R(I))+\frac{B^{\mathrm{add}}(C)}{K^{6}\rho^{-}}.
\]
On the other hand, when we compare for a cell $C_{i}$ the rectangles
$\R'_{\sm,Q}\cap\R_{\sm}(C_i)\cap\R(I)$ and $\R'_{\sm}\cap\R_{\sm}(C_i)\cap\R(I)$
we have removed rectangles of the densities greater than $\sigma^{+}(C_{i},Q)$.
If $\sigma^+(C_i,Q) = (\infty,\infty)$, nothing is removed and we cover enough of $L(I)$.
Hence assume otherwise.
By definition of $\sigma^{+}(C_{i},Q)$ we know that there is a cell
$C_{j}$, $i<j<\ell-1$ with $\sigma^{+}(C_{i},Q)=H^{j-i}\sigma^{-}(C_{j},Q)$, where $H \ge K^7/\epsilon$.
Let $(\rho^{+},{\rho'}^{+})=\sigma^{+}(C_{i},Q)$. Notice that $B^{\mathrm{add}}(C_{j})/K^{4}\ge\sum_{k=1}^{j-1}(\epsilon/K)^{j-k}/K^{4}\cdot B^{\mathrm{round}}(C_{k})$
. We will charge the deleted rectangles in $C_{i}$ against the summand
for $k=i$ in the sum above. In cell $C_{j}$ we cover an additional
volume of 
\[
\frac{(\epsilon/K)^{j-i}}{K^{6}\rho^{-}}B^{\mathrm{round}}(C_{i})
\ge \frac{(\epsilon/K)^{j-i}}{K^{6} (K^7/\epsilon)^{i-j} \rho^{+}}B^{\mathrm{round}}(C_{i})
\ge\frac{1}{\rho^{+}}B^{\mathrm{round}}(C_{i}).
\]
The right-hand side is an upper bound for the volume that has been
deleted. Thus, the deleted volume is less than the added volume.
Notice that we can assume $L(I)$ fully traverses $C_j$. If it does not, then
$L(I)$ does not hit $C_i$ and the deleted volume in $C_i$ is irrelevant.
Since all deleted rectangles are higher than those added, the ray $L(I)$
is still covered. 
\end{proof}
}
\subsubsection{Definition of families $\left\{ \chi_{\protect\sm,Q}\right\} _{Q\in\protect\Q}$}

Let $Q\in\Q$; we want to define the set $\chi_{\sm,Q}$. Recall that
we want to ensure that $\left|\chi_{\sm,Q}\right|\le(nP)^{O_{\epsilon}(1)}$.
We argue that for $\R'_{\sm,Q}$ there are at most $(nP)^{O_{\epsilon}(1)}$
options and we define $\chi_{\sm,Q}$ to be the family of sets that
contains each of these possible options. Let $C_{1},C_{2},\dotsc,C_{\ell}$
be the cells in $Q$ from root to bottom. We will argue that $\R'_{\sm,Q}$
is completely defined once we know 
\begin{enumerate}
\item the budget $B^{\mathrm{round}}(C)$ for each cell $C\in Q$ (which
imply $B^{\mathrm{add}}(C)$ for each cell $C\in Q$), 
\item for each $i<\ell-1$ whether $\sigma^{+}(C_{i},Q)\ge\sigma^{-}(C_{i},Q)$
and if so the values of $\sigma^{+}(C_{i},Q)$ and $\sigma^{-}(C_{i},Q)$, 
\item for each $i<\ell-1$, and $s,s'$ with $s'\ge r(s,C_{i},Q)$
\begin{enumerate}
\item whether (i) $\gamma(C_{i},s,s')<\sigma^{-}(C_{i},Q)$, (ii) $\gamma(C_{i},s,s')>\sigma^{+}(C_{i},Q)$,
or (iii) $\sigma^{-}(C_{i},Q)\le\gamma(C_{i},s,s')\le\sigma^{+}(C_{i},Q)$, 
\item the budgets $B_{\sm}^{\mathrm{round}}(C_{i},(\rho,\rho',s),s')$ for
all $\sigma^{-}(C_{i},Q)\le_{L}(\rho,\rho')<_{L}\sigma^{+}(C_{i},Q)$, 
\item in case of (iii) also the density pair $\gamma(C_{i},s,s')$, 
\end{enumerate}
\item the density pair $\gamma(C_{i},s,s')$ and budgets $B_{\sm}^{\mathrm{round}}(C_{i},(\rho,\rho',s),s')$
for $i\in\{\ell-1,\ell\}$, $s,s'$, and $(\rho,\rho')\ge_{L}\gamma(C_{i},s,s')$
. 
\end{enumerate}

\paragraph{Reconstructing $\protect\R'_{\protect\sm,Q}$.}

We assume we are given the information above and argue that this suffices
to reconstruct $\R'_{\sm,Q}$. For cells $C_{\ell}$ and $C_{\ell-1}$
we know all parameters passed to GreedySelectSmall when constructing
the rectangles of the solution, hence these are easy. Now let $C_{i}$
be a cell with $i<\ell-1$. Let $(\rho,\rho')$ be a density pair
that is lexicographically smaller than $\sigma^{-}(C_{i},Q)$ and
let $j$ be a small job with this density pair and $s=|\R(j,C)|$.
Then we know that $\R'_{\sm,Q}$ selects exactly the first $r(s,C,Q)$
rectangles of $\R(j,C)$ (provided that $r(s,C,Q)\neq0$). On the
other hand, for density pairs $(\rho,\rho')$ that are lexicographically
bigger than $\sigma^{+}(C_{i},Q)$ we know that $\R'_{\sm,Q}$ does
not select any rectangles. Hence, we can focus on the density pairs
$(\rho,\rho')$ with $\sigma^{-}(C_{i},Q)\le(\rho,\rho')\le\sigma^{+}(C_{i},Q)$.
For such densities we again know all parameters that are passed to
GreedySelect when constructing the rectangles of the solution. 

\paragraph{Guessing relevant quantities.}

Let us now argue why the mentioned values can be guessed efficiently.
By Lemma~\ref{lem:guessA} we can guess in time $(nP)^{O_{\epsilon}(1)}$
all values in $\{B_{\sm}^{\mathrm{round}}(C)\}_{C\in Q}$. In time
$O_{\epsilon}(\log^{O_{\varepsilon}(1)}(nP))$ we can guess the values
$\gamma(C_{i},s,s')$ for $i\in\{\ell-1,\ell\}$
and all $s,s'$. Notice that $\gamma(C_{i},s,s')\in\Gamma(C_{i})$
and $\sigma^{-}(C_{i},Q)\in\Gamma(C_{i})$ for all $C_{i},s,s'$.
Here $\Gamma(C_{i})$ is defined as in Lemma~\ref{lem:guessB}, namely
\[
\Gamma(C_{i})=\{\gamma(C_{i},s,s'):s,s'\}.
\]
In particular, by Lemma~\ref{lem:guessB} we can guess in time $(nP)^{O_{\epsilon}(1)}$
the budgets $B_{\sm}^{\mathrm{round}}(C_{i},(\rho,\rho',s),s')$ for
$i\in\{\ell-1,\ell\}$, $s,s'$, and $(\rho,\rho')\ge_{L}\gamma(C_{i},s,s')$
. Next, we will guess the values $\{\sigma^{+}(C_{i},Q)\}_{i=1}^{\ell-2}$.
Here observe that for all $i=1,\dotsc,\ell-3$ 
\[
\sigma^{+}(C_{i},Q)\le_{L} H\cdot\sigma^{+}(C_{i+1},Q),
\]
\lr{where $H\le (1 + \epsilon) K^7/\epsilon \le O_{\varepsilon}(1)$.}
This is because $\sigma^{+}(C_{i+1},Q)$ is defined as the lexicographic
minimum density pair $H^{j-(i+1)}\sigma^{-}(C_{j},Q)$
\lr{over all $j\in\{i+2,i+3,\dotsc,\ell-2\}$}. Let
$C_{j}$ be the cell that achieves the minimum. Then it follows that
$\sigma^{+}(C_{i},Q)\le_{L} H^{j-i}\cdot\sigma^{-}(C_{j},Q) \lr{=} H\cdot\sigma^{+}(C_{i+1},Q)$.
This enables us to guess all values $\sigma^{+}(C_{i},Q)$ efficiently. 
\begin{lem}
In time $(nP)^{O_{\epsilon}(1)}$ we can guess all values $\{\sigma^{+}(C_{i},Q)\}_{i=1}^{\ell-2}$. 
\end{lem}

\begin{proof}
Let $(\rho_{i},\rho'_{i})=\sigma^{+}(C_{i},Q)$ for all $i=1,\dotsc,{\ell-2}$.
It suffices to guess all values $\rho_{1},\rho_{2},\dotsc,\rho_{\ell-2}$:
By Lemma~\ref{lem:no-types} we know that given these values there
are only $O_{\epsilon}(1)$ possible values for each $\rho'_{i}$.
Hence all these values can be guessed in $2^{O_{\epsilon}(\ell-2)}=(nP)^{O_{\epsilon}(1)}$
time. To guess the values $\rho_{i}$, we exploit that $\rho_{i+1}\ge\epsilon/K^{7}\cdot\rho_{i}$
for all $i=1,\dotsc,\ell-3$. In other words, $\log_{1+\epsilon}(\rho_{i+1})\ge\log_{1+\epsilon}(\rho_{i})+\log_{1+\epsilon}(\epsilon/K^{7})$.
Our task is to guess some non-negative integers $q_{1},q_{2},\dotsc,q_{\ell-2}\le O_{\epsilon}(\log(nP))$
such that $\ell-2\le O_{\epsilon}(\log(nP))$ and there is some $c = O_{\varepsilon}(1)$ with
$q_{i+1} > q_{i} - c$ for all $i$. By transforming to $p_{i}=q_{i}+i\cdot c$ we
get the equivalent problem of guessing values $p_{1},p_{2},\dotsc,p_{\ell-2}\le O_{\epsilon}(\log(nP))$
such that $\ell-2\le O_{\epsilon}(\log(nP))$ and $p_{i+1}>p_{i}$
for all $i$. This can be done in time $2^{O_{\epsilon}(nP)}=(nP)^{O_{\epsilon}(1)}$:
We guess for each $v\in\{0,1,\dotsc,O_{\epsilon}(nP)\}$ whether there
is some $i$ with $p_{i}=v$. After this the values $p_{1},p_{2},\dotsc,p_{\ell-2}$
are fully defined. The first such value must be $p_{1}$, the second
$p_{2}$, etc. 
\end{proof}
Now consider the values $\sigma^{-}(C,Q)$. Since these are only $O_{\epsilon}(\log(nP))$
many, we can guess in time $(nP)^{O_{\epsilon}(1)}$ which of them
satisfy $\sigma^{-}(C,Q)\ge_{L}\sigma^{+}(C,Q)$. For such cells
$C$ we do not need to guess any budgets $B_{\sm}^{\mathrm{round}}(C,(\rho,\rho',s),s')$.
On the other hand, consider some $C_{i}$ where $\sigma^{-}(C_{i},Q)$
does not satisfy the inequality above. Then it holds that 
\[
\sigma^{+}(C_{i},Q)>_{L}\sigma^{-}(C_{i},Q)\ge_{L} \frac{1}{H}\sigma^{+}(C_{i-1},Q).
\]
Intuitively, it is easy to guess $\sigma^{-}(C_{i},Q)$ unless $\sigma^{+}(C_{i},Q)$
is much larger than $\sigma^{+}(C_{i-1},Q)$. Although this is possible,
it cannot happen often as shown by the following lemma. 
\begin{lem}
\label{lma:no-density-pairs} For each $i=1,\dotsc,\ell-2$ let $\ell_{C_{i}}$
be the number of pairs $(\rho,\rho')$ with 
\[
\sigma^{+}(C_{i+1},Q)>_{L}(\rho,\rho')\ge_{L} \frac{1}{H}\sigma^{+}(C_{i},Q).
\]
Then $\sum_{i=1}^{\ell-2}\ell_{C_{i}}\le O_{\epsilon}(\log(nP))$. 
\end{lem}

\begin{proof}
Let $(\rho_{i},\rho'_{i})=\sigma^{+}(C_{i},Q)$ for each $i=1,\dotsc,\ell-2$.
Since $H \le O_{\varepsilon}(1)$ the number of values $\rho$ with $\rho_{i+1}\ge\rho\ge \lr{H^{-1}}\rho_{i}$
is at most $O_{\epsilon}(\max\{0,\log(\rho_{i+1}/(\lr{H^{-1}}\rho_{i}))\})=O_{\epsilon}(\max\{0,\log(\rho_{i+1}/\rho_{i})\})$.
By Lemma~\ref{lem:no-types} we have that the number of density
pairs $(\rho,\rho')$ such that $\rho_{i+1}\ge\rho\ge\epsilon/K^{7}\cdot\rho_{i}$
is also a most is at most $O_{\epsilon}(\max\{0,\log(\rho_{i+1}/\rho_{i})\})$.
This implies 
\begin{multline*}
\sum_{i=1}^{\ell}\ell_{C_{i}}\le\sum_{i=1}^{\ell-1}O_{\epsilon}\left(\max\left\{ 0,\log\frac{\rho_{i+1}}{\rho_{i}}\right\} \right)\le\sum_{i=1}^{\ell-1}O_{\epsilon}\left(\log\frac{\rho_{i+1}}{\rho_{i}}\right)-\sum_{i=1}^{\ell-1}O_{\epsilon}\left(\min\left\{ 0,\log\frac{\rho_{i+1}}{\rho_{i}}\right\} \right)\\
\le O_{\epsilon}(\ell)+O_{\epsilon}\left(\log\frac{\rho_{\ell}}{\rho_{1}}\right)+\sum_{i=1}^{\ell}O_{\epsilon}\left(\max\left\{ 0,\log\frac{\rho_{i}}{\rho_{i+1}}\right\} \right).
\end{multline*}
Notice that $O_{\epsilon}(\ell)\le O_{\epsilon}(\log(nP))$, $O_{\epsilon}(\log(\rho_{\ell}/\rho_{1}))\ge O_{\epsilon}(\log(nP))$,
and $\rho_{i+1}\ge\epsilon/K^{7}\cdot\rho_{i}$ for all $i$. It follows
that $\sum_{i=1}^{\ell}\ell_{C_{i}}\le O_{\epsilon}(\log(nP))$. 
\end{proof}
\begin{lem}
In time $(nP)^{O_{\epsilon}(1)}$ we can guess all values $\{\sigma^{-}(C_{i},Q)\}_{i=1}^{\ell-2}$
for which $\sigma^{-}(C,Q)<_{L}\sigma^{+}(C,Q)$. 
\end{lem}

\begin{proof}
By Lemma~\ref{lma:no-density-pairs} there are only $\ell_{C}$ candidates
for $\sigma^{-}(C,Q)$ for each $C\in Q$. Hence we guess each of
them in time $\ell_{C}$ yielding an overall time of 
\[
\prod_{i=1}^{\ell-2}\ell_{C_{i}}\le\prod_{i=1}^{\ell-2}2^{\ell_{C_{i}}}\le2^{\sum_{i=1}^{\ell-2}\ell_{C_{i}}}\le2^{O_{\epsilon}(\log(nP))}\le(nP)^{O_{\epsilon}(1)}.\qedhere
\]
\end{proof}
Regarding the budgets $\left\{ B_{\sm}^{\mathrm{round}}(C,(\rho,\rho',s),s')\right\} _{s'}$,
there are only $K^{4}\sum_{C\in Q}\ell_{C}\le O_{\epsilon}(\log nP)$
combinations of a cell $C$, a pair $(\rho,\rho')$ such that $\sigma^{-}(C,Q)\le_{L}(\rho,\rho')<_{L}\sigma^{+}(C,Q)$,
and values $s,s'$. Therefore, for all these combinations the budgets
$\left\{ B_{\sm}^{\mathrm{round}}(C,(\rho,\rho',s),s')\right\} _{s'}$
can be guessed in time $(nP)^{O_{\epsilon}(1)}$ by Lemma~\ref{lem:guessB}.
This completes the proof of Lemma~\ref{lem:small-jobs}.

\subsection{\label{subsec:large-rectangles}Consistent solution for large rectangles}

In this section we prove Lemma~\ref{lem:large-jobs}. Observe that
for each cell $C$ there can be at most $1/\delta$ jobs that are
large for $C$ and such that $\R(j,C)\cap\R^{*}\ne\emptyset$. In
particular, for a path $Q\in\Q$ there can be at most $|Q|/\delta\le O_{\epsilon}((\log nP)/\delta)$
jobs that are large for a cell $C\in Q$ and such that $\R(j,C)\cap\R^{*}\ne\emptyset$. 

First, for each cell $C$ we define $B_{\l}^{\opt}(C):=\sum_{R\in\R^{*}\cap\R(C)\cap\Rl}c_{R}$
and we define $B_{\l}^{\r}(C)$ to be the smallest multiple of $\epsilon\cdot B^{\r}(C)$
that is larger than $B_{\l}^{\opt}(C)$. Our strategy is to define
a solution $\R'_{\l}\subseteq\Rl$ that for each cell $C$ spends
at most $2\cdot B_{\l}^{\opt}(C)$ on rectangles of jobs that are
large for $C$. Then, for each path $Q\in\Q$ we define a solution
$\R'_{\l,Q}\subseteq\R'_{\l}$. 

\paragraph{Type groups.}

First, we \lr{form groups of job types}. We say that two types $(\rho,\rho',s)$,
$(\bar{\rho},\bar{\rho}',\bar{s})$ are in the same group if $\rho/\rho'=\bar{\rho}/\bar{\rho}'$
and $s=\bar{s}$. The intuition is that jobs of the same type behave
similarly w.r.t. $\rho$ and $\rho'$; therefore, we will consider
the jobs of each group separately. This is possible since there are
only constantly many groups. 
\begin{lem}
There are at most $O_{\epsilon}(1)$ different groups. 
\end{lem}

\begin{proof}
Consider a job $j$, a cell $C$, and a set of rectangles $\R(j,C)$
of a type $(\rho,\rho',s)$. Assume that $S$ is the segment corresponding
to the first rectangle $R\in\R(j,C)$, and assume that $S'$ is the
segment corresponding to the second rectangle $R'\in\R(j,C)$. Then
it holds that $c_{R}=w_{j}(\celle(S)-r_{j})$ and $c_{R'}=w_{j}(\celle(S')-\cellb(S'))$.

We first note that since $\left|\Seg(j,C)\right|\le K^{2}$ (see Lemma~\ref{lem:define-segments}),
it holds that $s\in\{1,\dotsc,K^{2}\}$ and hence there are only $K^{2}=O_{\epsilon}(1)$
options for $s$. We claim that for $\rho/\rho'$ there are also only
$O_{\epsilon}(1)$ options. Let $C(\ell_{\max}),C(\ell_{\max}-1),\dotsc,C(k)=C$
be the cells in the construction of $\Seg(j)$ (see Section~\ref{sec:geo}).
Then it holds that $\cellb(S)-r_{j}\le\sum_{i=k+1}^{\ell_{\max}}\len(C(i))\le\len(C(k))$
since the widths of the cells are geometrically increasing. Hence,
$\celle(S)-r_{j}\le2\cdot\len(C)$. On the other hand, the segment
corresponding to $R'$ equals the grid cell $C'$ of some level $\ell(C)+2$
(see Lemma~\ref{lem:define-segments}) and thus $\celle(S')-\cellb(S')=\len(C)K^{2}$.
Therefore,
\begin{equation*}
\frac{\rho}{\rho'}=\frac{c_{R}}{c_{R'}}=\frac{\celle(S)-r_{j}}{\celle(S')-\cellb(S')}\le\frac{2\len(C)}{\len(C)/K^{2}}=2K^{2}=O_{\epsilon}(1) . \qedhere
\end{equation*}
\end{proof}
We say that a set $\R(j, C)$ is \emph{in group $g$} if $j$ is of
a type $\tau$ in group $g$. 

\subsubsection{Definition of $\protect\R'_{\protect\l}$}

Let $C$ be a cell. We want to define $\R'_{\l}\cap\R(C)$ and to
this end, we consider separately each group $g$ and define which
rectangles we select from the sets $\R(j, C)$ of group $g$. Let $k(C,g)=:k\le1/\delta$
denote the number of jobs of group $g$ for which there is at least
one rectangle contained in $\R^{*}\cap\Rl(C)$. Denote by $j_{1},\dotsc,j_{k}$
the corresponding jobs. For each job $j\in\{j_{1},\dotsc,j_{k}\}$ denote
by $s'(j)$ the number of its rectangles that are contained in $\R^{*}$.
For each $k'\in\{1,\dotsc,k\}$ we define a budget $B_{\mathrm{large},g,k'}^{\mathrm{round}}(C)$
which intuitively (over-)estimates the part of $B_{\mathrm{large}}^{\mathrm{round}}(C)$
that is used for the rectangles of job $j_{k'}$. Formally, we define
that $B_{\mathrm{large},g,k'}^{\mathrm{round}}(C)$ is the smallest
integral multiple of $\epsilon\delta\cdot B_{\mathrm{large}}^{\mathrm{round}}(C)$
that is at least as large as $c\left(\R^{*}\cap\R(C,j_{k'})\right)$.
Note that therefore $\sum_{k'=1}^{k}B_{\mathrm{large},g,k'}^{\mathrm{round}}(C)\le(1+\epsilon)B_{\mathrm{large}}^{\mathrm{round}}(C)$.

Later, when we define the families $\left\{ \chi_{\l,Q}\right\} _{Q\in\Q}$,
we will not be able to guess the jobs $\{j_{1},\dotsc,j_{k}\}$ directly,
not even their exact processing times $\left\{ p_{j_{1}},\dotsc,p_{j_{k}}\right\} $.
However, intuitively we will be able to guess $(1+\epsilon)$-estimates
of their processing times. To this end, we define $\bar{p}_{j}:=(1+\epsilon)^h$ for each job $j$ where
\begin{equation*}
(1 + \epsilon)^{h} \le p_{j} < (1 + \epsilon)^{h+1}
\end{equation*}
and our goal later will be to guess the values $\left\{ \bar{p}_{j_{1}},\dotsc,\bar{p}_{j_{k}}\right\} $.
Then, for each value $k'$ we know that $\R^{*}$ selects the first
$s'(j_{k'})$ rectangles of some job $j$ with $\bar{p}_{j}=\bar{p}_{j_{k'}}$
and these rectangles cost at most $B_{\mathrm{large},g,k'}^{\mathrm{round}}(C)$
in total (and we can afford to spend $B_{\mathrm{large},g,k'}^{\mathrm{round}}(C)$).
Unfortunately, it is not clear how to find this job $j$. In particular,
there can be two candidate jobs $j,j'$ such that $j\prec j'$ and
hence the rectangles of $j$ potentially intersect with fewer rays
than the rectangles of $j'$, but on the other hand $p_{j}>p_{j'}$.
Even though $p_{j}\le(1+\epsilon)p_{j'}$, the difference $p_{j}-p_{j'}$
might be critical for whether a ray is completely covered or not.
Therefore, it is not clear which job we should select, $j$ or $j'$.
A similar situation can occur with $\omega(1)$ jobs, rather than
only $j$ and $j'$.

To remedy this issue, our strategy is that for each $k'$ we find
\emph{two }jobs with processing time at least $\bar{p}_{j_{k'}}$
and for which selecting the first $s'(j_{k'})$ rectangles costs at
most $B_{\mathrm{large},g,k'}^{\mathrm{round}}(C)$ each. In this
way, we pay at most $2\cdot B_{\mathrm{large},g,k'}^{\mathrm{round}}(C)$
for the rectangles of these two jobs and we will ensure that together
they cover as much as the rectangles of $j_{k'}$ in $\R^{*}$. More
precisely, observe that for any two jobs $j^{(1)},j^{(2)}$ with $\bar{p}_{j^{(1)}}=\bar{p}_{j^{(2)}}=\bar{p}_{j_{k'}}$
it holds that $p_{j^{(1)}}+p_{j^{(2)}}\ge\frac{2}{1+\epsilon}p_{j_{k'}}$.
In fact, $\frac{2}{1+\epsilon}p_{j_{k'}}$ is by a constant factor
larger than $p_{j_{k'}}$ (so we cover substantially more) which will
be crucial later in order to argue that we can \lr{omit} some of
the large rectangles when we define the sets $\R'_{Q}$, i.e., an \lr{omitted}
rectangle $R\in\R(Q)$ is included in $\R'$ but not in $\R'_{Q}$.

Unfortunately, it might be that we do not find two such jobs $j^{(1)},j^{(2)}$
for a job $j_{k'}\in\{j_{1},\dotsc,j_{k}\}$. For example, this happens
if $j_{k'}$ is the job with largest processing time for which the
first $s'(j_{k'})$ rectangles cost at most $B_{\mathrm{large},g,k'}^{\mathrm{round}}(C)$
and for any other job $j$ it holds that $\bar{p}_{j}<\bar{p}_{j_{k'}}$.
Our strategy is to guess such a job $j_{k'}$ directly. To this end,
for each $k'\in\{1,\dotsc,k\}$ we consider the $k$ jobs $j$ with maximum
processing time in group $g$ with $\R(j,C)\ne\emptyset$ such that
buying the first $s'(j_{k'})$ rectangles of $j$ costs at most $B_{\mathrm{large},g,k'}^{\mathrm{round}}(C)$.
Let $J_{\mathrm{high},g,k'}(C)$ denote the corresponding jobs for
$k'$, and let $J_{\mathrm{high},g}(C):=\bigcup_{k'\in\{1,\dotsc,k\}}J_{\mathrm{high},g,k'}(C)$.
Note that $\left|J_{\mathrm{high},g}(C)\right|\le k^{2}\le1/\delta^{2}$.
If a job $j_{k'}$ is contained in $J_{\mathrm{high},g}(C)$ then
we select the first $s'(j_{k'})$ of its rectangles, i.e., add them
to $\R'_{\l}$. We call $j$ \emph{easy}. Later, we can guess the
easy jobs in time $2^{1/\delta^{2}}$ since they are all contained
in $J_{\mathrm{high},g}(C)$. If a job $j\in\{j_{1},\dotsc,j_{k}\}$
is not easy then we call $j$ \emph{hard. }Denote by $\Je(C)$ and
$\Jh(C)$ the easy and hard jobs in $\{j_{1},\dotsc,j_{k}\}$, respectively.

\paragraph{Hard jobs.}

We describe now which rectangles we select in order to cover as much
as the rectangles of the hard jobs in $\R^{*}\cap\Rl\cap\R(C)$ .
To this end, we partition the jobs in $\{j_{1},\dotsc,j_{k}\}\cap\Jh(C)$
according to the values $\bar{p}_{j_{1}},\dotsc,\bar{p}_{j_{k}}$, i.e.,
in each set of the partition each job $j$ has the same value $\bar{p}_{j}$.
Let $J'=\{j'_{1},\dotsc,j'_{|J'|}\}\subseteq\{j_{1},\dotsc,j_{k}\}$ be
a set of this partition. Assume that these jobs are ordered non-increasingly
according to $\prec$, i.e., $j'_{k'+1}\prec j'_{k'}$ for each $k'$.
In particular, note that the rectangles of $j'_{k'}$ are further
down than the rectangles of $j'_{k'+1}$ for each $k'$. We consider
the jobs in $J'$ in this order. Consider a job $j_{k'}\in J'$. We
define two jobs $j^{(1)},j^{(2)}$ such that $j^{(1)}$ and $j^{(2)}$
are the two maximal jobs $j$ according to $\prec$ (i.e., $j\prec j^{(1)}$
and $j\prec j^{(2)}$ for any candidate job $j$ with $j\neq j^{(1)}$ and $j\ne j^{(2)}$)
with the properties that 
\begin{itemize}
\item $\bar{p}_{j}=\bar{p}_{j_{k'}}$, 
\item $j$ is in group $g$ and $\R(j,C)\neq\emptyset$ , 
\item buying the first $s'(j_{k'})$ rectangles of $j$ costs at most $B_{\mathrm{large},g,k'}^{\mathrm{round}}(C)$, 
\item $j$ is smaller according to $\prec$ than the last job $\hat{j}\notin J_{\mathrm{high},g}(C)$
with $\bar{p}_{j}=\bar{p}_{\hat{j}}$ from which we have selected
rectangles before (in the first iteration this condition does not
apply; in particular $\hat{j}$ is not defined yet).
\end{itemize}
If we find two such jobs $j^{(1)},j^{(2)}$ then one can show that
together their respective first $s'(j_{k'})$ rectangles cover as
much as the first $s'(j_{k'})$ rectangles of $j_{k'}$. More formally,
for job $j$ and each $\ell$ let $\R_{\ell}(j,C)\subseteq\R(j,C)$
denote the first $\ell$ rectangles in $\R(j,C)$; then for each interval
$I$ one can show that 
\[
d\left(\left(\R_{s'(j_{k'})}(j^{(1)},C)\cup\R_{s'(j_{k'})}(j^{(2)},C)\right)\cap\R(I)\right)\ge d\left(\R_{s'(j_{k'})}(j_{k'},C)\cap\R(I)\right).
\]
Intuitively, we would like to select the respective first $s'(j_{k'})$
rectangles of $j^{(1)}$ and $j^{(2)}$ (i.e., add them to $\R'_{\l}$)
and continue with the next job in $J'$. However, it might be that
$j=j^{(1)}$ and $j^{(2)}\in J'$. Then we cannot use the rectangles
from $j^{(2)}$ to argue later that we cover strictly more than $\R^{*}\cap\Rl\cap\R(C)$
(which will be crucial later), because some rectangles from $j^{(2)}$
are already included in $\R^{*}$.  Instead, if $j_{k'}=j^{(1)}$
then we do not select rectangles from $j^{(2)}$ (to avoid the case
that $j^{(2)}\in J'$) but instead select rectangles from some job
$\tilde{j}\in J_{\mathrm{high},g}(C)\setminus\Je(C)$ from which we
have not yet selected any rectangle. In particular, then $\tilde{j}\notin J'$
and the mentioned problem does not occur.

Formally, we distinguish the two cases 
\begin{itemize}
\item $j_{k'}\ne j^{(1)}$: we select the first $s'(j_{k'})$ rectangles
from $j^{(1)}$ and $j^{(2)}$. Then potentially $j_{k'}=j^{(2)}$
but we will ensure that $j^{(1)}\ne j\ne j^{(2)}$ for each job $j'\in\{j_{k'+1},\dotsc,j_{k}\}$, 
\item $j_{k'}=j^{(1)}$: we select the first $s'(j_{k'})$ rectangles from
$j^{(1)}$ and additionally the first $s'(j_{k'})$ rectangles from
the job in $\tilde{j}\in J_{\mathrm{high},g}\setminus\Je(C)$ with
largest processing time among all jobs in $J_{\mathrm{high},g}\setminus\Je(C)$
for which buying the first $s'(j_{k'})$ rectangles costs at most
$B_{\mathrm{large},g,k'}^{\mathrm{round}}(C)$ and from which we have
not selected any rectangle so far. Using that $|J_{\mathrm{high},g,k'}|=k$
(if $|J_{\mathrm{high},g,k'}|<k$ then $j_{k'}$ would not be hard),
one can show that we always find such a job $\tilde{j}$. 
\end{itemize}
We will show that $j^{(1)}$ is always defined since $j_{k'}$ itself
will always be a candidate. The job $j^{(2)}$ might not be defined
though; however, then the second case applies and thus our procedure
is well-defined.

We repeat the procedure above for each set of the partition according
to the values $\bar{p}_{j_{1}},\dotsc,\bar{p}_{j_{k}}$ which completes
the our treatment of group $g$. Let $\mathcal{R}'_{\l}(C)\subseteq\Rl(C)$
denote the set of all rectangles selected by this procedure for the
cell $C$. We do this procedure for each cell $C\in\C$ and finally
define $\mathcal{R}'_{\l}:=\bigcup_{C\in\C}\mathcal{R}'_{\l}(C)$.
\begin{lem}
It holds that $c(\mathcal{R}'_{\l})\le2c(\R^{*}\cap\Rl)+3\epsilon\cdot c(\R^{*})$. 
\end{lem}

\begin{proof}
Consider a cell $C$ and a group $g$. Let $J(g)$ denote the set
of all jobs $j$ such that $\R(C,j)$ is in group $g$. Let $J^{*}(g):=\{j_{1},\dotsc,j_{k}\}\subseteq J(g)$
be the jobs from $J(g)$ for which $\R(j,C)\cap\R^{*}\ne\emptyset$.
For each job $j_{k'}\in\{j_{1},\dotsc,j_{k}\}$ we bought the first $s'(j_{k'})$
rectangles of two jobs for which buying the first $s'(j_{k'})$ rectangles
costs at most $B_{\mathrm{large},g,k'}^{\mathrm{round}}(C)$. This
yields a total cost of at most $2\sum_{k'=1}^{k}B_{\mathrm{large},g,k'}^{\mathrm{round}}(C)$.
For each $k'\in\{1,\dotsc,k\}$ we have that $B_{\mathrm{large},g,k'}^{\mathrm{round}}(C)\le c\left(\R^{*}\cap\R(C,j_{k'})\right)+\epsilon\delta\cdot B_{\mathrm{large}}^{\mathrm{round}}(C)$.
Across all groups, there can be at most $1/\delta$ jobs $j$ that
are large for $C$ and which satisfy $\R(j,C)\cap\R^{*}\ne\emptyset$.
Thus, if we define $B_{\mathrm{large},g,j_{k'}}^{\mathrm{round}}(C)$
to be the value $B_{\mathrm{large},g,k'}^{\mathrm{round}}(C)$ that
corresponds to the job $j_{k'}$, we obtain that

\begin{eqnarray*}
c(\mathcal{R}'_{\l}\cap\R(C)) & \le & \sum_{g}\sum_{j_{k'}\in J^{*}(g)}2\cdot B_{\mathrm{large},g,j_{k'}}^{\mathrm{round}}(C)\\
 & \le & \sum_{g}\sum_{j_{k'}\in J^{*}(g)}2\left(c\left(\R^{*}\cap\R(C,j_{k'})\right)+\epsilon\delta\cdot B_{\mathrm{large}}^{\mathrm{round}}(C)\right)\\
 & \le & 2c\left(\R^{*}\cap\Rl\cap\R(C)\right)+2\epsilon\cdot B_{\mathrm{large}}^{\mathrm{round}}(C).
\end{eqnarray*}
Therefore, 
\begin{eqnarray*}
c(\mathcal{R}'_{\l}) & = & \sum_{C}c(\mathcal{R}'_{\l}\cap\R(C))\\
 & \le & \sum_{C}\left(2c\left(\R^{*}\cap\Rl\cap\R(C)\right)+2\epsilon\cdot B_{\mathrm{large}}^{\mathrm{round}}(C)\right)\\
 & \le & 2c(\R^{*}\cap\Rl)+2\epsilon\sum_{C}B_{\mathrm{large}}^{\mathrm{round}}(C)\\
 & \le & 2c(\R^{*}\cap\Rl)+2(1+\epsilon)\epsilon\sum_{C}B^{\mathrm{round}}(C)\\
 & \le & 2c(\R^{*}\cap\Rl)+2(1+\epsilon)^{2}\epsilon c(\R^{*}).\\
 & \le & 2c(\R^{*}\cap\Rl)+3\epsilon c(\R^{*}).
\end{eqnarray*}
\end{proof}

\subsubsection{Definition of the sets $\protect\R'_{\protect\l,Q}$ }

We define now the sets $\left\{ \R'_{\l,Q}\right\} _{Q\in\Q}$ with
$\R'_{\l,Q}\subseteq\R'_{\l}$ for each $Q\in\Q$. Let $Q\in\Q$.
Let $v_{\tilde{C}}$ be the bottom-most vertex of $Q$ and let $v_{\tilde{C}'}$
be its parent vertex; hence, $\tilde{C}$ and $\tilde{C}'$ are their
associated cells. The set $\R'_{\l,Q}$ will contain all rectangles
in $\R(Q)\cap\Rl$ that we will use in order to satisfy the demand
of intervals $I=[s,t)$ such that $t\in\tilde{C}$. One might think
that those are the rectangles in $\R'_{\l}\cap\R(Q)$. However, $\R'_{\l,Q}$
will not necessarily contain \emph{all} rectangles in $\R'_{\l}\cap\R(Q)$
but potentially only a \emph{subset}. We will ensure that for this
subset there will be only few options, which will ensure that $|\chi_{Q}|\le(nP)^{O(1)}$.

Consider a cell $C\in Q$ and a group $g$. First, for each job $j\in\Je(C)$
we add to $\R'_{\l,Q}$ all rectangles in $\R'_{\l}\cap\R(j,C)$.
Intuitively, regarding the hard jobs, for each hard jobs $j$ with
$\R(j,C)\cap\R^{*}\ne\emptyset$ we selected rectangles from \emph{two
}jobs $j^{(1)},j^{(2)}$. Therefore, for each interval $I$ we selected
rectangles in $\R(I)$ with larger total capacity than $\R^{*}$.
Therefore, when we define $\R'_{\l,Q}$ we can omit some of the rectangles
in $\R_{\l}'\cap\R(Q)$. We will ensure that for the remaining rectangles
there are only few options. 

For each cell $C\in Q$ and each job $j$ with $\R(j,C)$ we add the
rectangles in $\R(j,C)\cap\R'_{\l}$ to $\R'_{\l,Q}$ only if the
pair $(C,\bar{p}_{j})$ is \emph{relevant}. Formally, for each pair
$(C,\bar{p})$ where $C\in Q$ and $\bar{p}$ is a power of $1+\epsilon$,
we say that the pair $(C,\bar{p})$ is \emph{irrelevant }if 
\begin{enumerate}
\item there is no hard job $j$ with $\bar{p}_{j}=\bar{p}$ for which $\R(j,C)\cap\R^{*}$
  contains a rectangle $R(j, S')$ with $\tilde C\subseteq S'$ or 
\item if there is a hard job $j'\in\Jh(C')$ for some cell $C'\in Q\setminus\{\tilde{C},\tilde{C}'\}$
such that 
\begin{itemize}
\item $v_{C'}$ is a descendant of $v_{C}$, 
\item $\bar{p}\le\bar{p}_{j'}\cdot\delta\cdot\epsilon{}^{\dist(v_{C},v_{C'})}$,
and 
\item $\R^{*}$ contains a rectangle $R(j',S')$ such that $\tilde{C}\subseteq S'$. 
\end{itemize}
\end{enumerate}
Otherwise, we say that $(C,\bar{p})$ is \emph{relevant}.

Note that if $(C,\bar{p})$ is irrelevant because of condition 2.,
then for the mentioned (hard) job $j'$ we considered\emph{ two }jobs
($j^{(1)}$ and additionally $j^{(2)}$ or $\tilde{j}$) whose total
processing time is at least $2p_{j}/(1+\epsilon)$, and we added to
$\R'_{\l}$ the rectangle $R(j^{(1)},S')$ and additionally $R(j^{(2)},S')$
or $R(\tilde{j},S')$. Hence, these rectangles together cover more
than $R(j',S')$ (essentially at least twice as much) and this additional
coverage compensates for all rectangles in $\R^{*}$ corresponding
to the pair $(C,\bar{p})$ (which is irrelevant due to $j'$). Note
that for each irrelevant pair $(C,\bar{p})$ there can be at most
$1/\delta$ corresponding sets $\R(j,C)$ with $\R(j,C)\cap\R'_{\l}\cap\R^{*}\ne\emptyset$.
For each relevant pair $(C,\bar{p})$ we add to $\R'_{\l,Q}$ all
rectangles in $\R'_{\l}$ that belong to a set $\R(j,C)$ with $\R(j,C)\cap\R'_{\l}\ne\emptyset$
and $\bar{p}=\bar{p}_{j}$.

By construction, it follows that $\R'_{\l,Q}\subseteq\R'_{\l}$. Also,
with the above intuition, we can prove that $\R'_{\l,Q}$ covers as
much from each interval $I\in\I$ with $\R(I)\subseteq\R(Q)$ as the
large rectangles in $\R^{*}$. 
\begin{lem}
For each $I\in\I$ with $\R(I)\subseteq\R(Q)$ we have $p(\R'_{\l,Q}\cap\R(I))\ge p(\R^{*}\cap\R_{\l}\cap\R(I))$. 
\end{lem}

\begin{proof}
First, we observe that for each cell $C\in Q$ and each group $g$,
the set $\R'_{\l,Q}$ contains all rectangles in $\R^{*}\cap\R(Q)$
that correspond to (easy) jobs in $\Je(C)$.

Consider a cell $C\in Q$, a group $g$ and a (hard) job $j\in\Jh(C)$
with $\R^{*}\cap\R(j,C)\ne\emptyset$. Recall that when we defined
$\R'_{\l}$, at some point we considered this cell $C$ and the group
$g$. We considered the corresponding hard jobs $\{j_{1},\dotsc,j_{k}\}$.
We had one iteration for each job $j_{k'}\in\{j_{1},\dotsc,j_{k}\}$
and in one of these iterations $j_{k'}=j$. Then we defined the jobs
$j^{(1)},j^{(2)},\tilde{j}$ and added the first $s'(j_{k'})$ rectangles
of either both $j^{(1)}$ and $j^{(2)}$, or of both $j^{(1)}$ and
$\tilde{j}$. In the former case we define $J(j):=\{j^{(1)},j^{(2)}\}$,
in the latter case we define $J(j):=\{j^{(1)},\tilde{j}\}$. Also,
in case that $C'\notin\{\tilde{C},\tilde{C}'\}$, there might be some
rectangles in $\R^{*}\cap\R_{\l}$ that are irrelevant due to $j$,
in which case we define $\Rir(j)$ to be all of these rectangles,
i.e., all rectangles $R\in\R(j',C')$ for some job $j'$ and a cell
$C'$ such that $v_{C}$ is a descendant of $v_{C'}$, $\bar{p}_{j'}\le\bar{p}_{j}\cdot\epsilon\delta\cdot\epsilon^{\dist(v_{C},v_{C'})}$,
and $\R^{*}$ contains a rectangle $R(j,S)$ such that $\tilde{C}\subseteq S$.

Let $I\in\I$ with $\R(I)\subseteq\R(Q)$. We want to show that 
\begin{equation}
p(\R'_{\l,Q}\cap\bigcup_{j'\in J(j)}\R(j',C)\cap\R(I))\ge p((\R(j,C)\cup\Rir(j))\cap\R_{\l}\cap\R(I)).\label{eq:compensation}
\end{equation}
 We distinguish two cases. First consider the case that $J(j):=\{j^{(1)},j^{(2)}\}$.
Recall that $\bar{p}_{j^{(1)}}=\bar{p}_{j^{(2)}}=\bar{p}_{j}$. Also,
$j^{(1)}$ and $j^{(2)}$ are chosen maximally according to $\prec$
(i.e., with largest release dates). Also, in this case $j^{(1)}\ne j$.
The following part of the definition of $\R'_{\l}$ is important now:
whenever $j_{k'}=j_{k'}^{(1)}$ (where we define $j_{k'}^{(1)}$ and
$j_{k'}^{(2)}$ to be the respective jobs $j^{(1)}$ and $j^{(2)}$in
the iteration $k'$) for some $k'$ then we selected rectangles from
$j_{k'}^{(1)}$ and rectangles from $\tilde{j}$. Also, $j_{k'}^{(1)}$
and $j_{k'}^{(2)}$ are smaller according to $\prec$ than the last
job $\hat{j}\notin J_{\mathrm{high},g}(C)$ with $\bar{p}_{j}=\bar{p}_{\hat{j}}$
from which we had selected rectangles before. Therefore, for the jobs
$j,j^{(1)}$, and $j^{(2)}$ we know that $j\prec j^{(1)}$ and $j\preceq j^{(2)}$
and hence the rectangles of $j^{(1)}$ are further down in our visualization
than the rectangles of $j$. The same is true for $j^{(2)}$, unless
$j=j^{(2)}$. Also, the rectangles of $j^{(1)}$ and $j^{(2)}$ are
further down in our visualization than any rectangle in $\Rir(j)$.
We have that

\begin{eqnarray*}
p_{j^{(1)}}+p_{j^{(2)}} & \ge & \frac{2}{1+\epsilon}p_{j}\\
 & \ge & p_{j}+\sum_{C\in Q:v_{C'}\,\mathrm{is\,descendent\,of}\,v_{C}}\bar{p}_{j}\cdot2\epsilon^{\dist(v_{C},v_{C'})}\\
 & \ge & p_{j}+\sum_{C\in Q:v_{C'}\,\mathrm{is\,descendent\,of}\,v_{C}}\sum_{j:\R(j,C)\cap\R_{\l}\cap\R^{*}\ne\emptyset}\bar{p}_{j}\cdot2\delta\cdot\epsilon^{\dist(v_{C},v_{C'})}\\
 & \ge & p_{j}+p\left(\Rir(j)\cap\R(I)\right)
\end{eqnarray*}
which implies that our selected rectangles from $j^{(1)}$ and $j^{(2)}$
satisfy as much demand from $I$ as the rectangles from $j'$ in $\R^{*}$
and the rectangles in $\Rir(j)\cap\R^{*}\cap\R(I)$. Thus, inequality~\ref{eq:compensation}
holds in this case.

Now assume that $J(j):=\{j^{(1)},\tilde{j}\}$. In this case $j^{(1)}=j$.
Also, all rectangles from $\tilde{j}$ are further down in our visualization
than any rectangle in $\Rir(j)$. Similarly as above, we calculate
that

\begin{eqnarray*}
p_{\tilde{j}} & \ge & \frac{1}{1+\epsilon}p_{j}\\
 & \ge & \sum_{C\in Q:v_{C'}\,\mathrm{is\,descendent\,of}\,v_{C}}\bar{p}_{j}\cdot2\epsilon^{\dist(v_{C},v_{C'})}\\
 & \ge & \sum_{C\in Q:v_{C'}\,\mathrm{is\,descendent\,of}\,v_{C}}\sum_{j:\R(j,C)\cap\R_{\l}\cap\R^{*}\ne\emptyset}\bar{p}_{j}\cdot\delta\cdot2\epsilon^{\dist(v_{C},v_{C'})}\\
 & \ge & p\left(\Rir(j)\cap\R(I)\right)
\end{eqnarray*}
and thus inequality~\ref{eq:compensation} holds also in this case.
We complete the proof by calculating

\begin{eqnarray*}
p(\R'_{\l,Q}\cap\R(I)) & \ge & \sum_{j,C:j\,\mathrm{is\,easy}}p(\R'_{\l,Q}\cap\R(j,C)\cap\R(I))\\
 &  & +\sum_{j,C:j\,\mathrm{is\,hard\,or\,}\R(j,C)\cap\R^{*}=\emptyset}p(\R'_{\l,Q}\cap\R(j,C)\cap\R(I))\\
 & \ge & \sum_{j,C:j\,\mathrm{is\,easy}}p(\R^{*}\cap\R(j,C)\cap\R(I))\\
 &  & +\sum_{j,C:j\,\mathrm{is\,hard}}p(\R'_{\l,Q}\cap\bigcup_{j'\in J(j)}\R(j',C)\cap\R(I))\\
 & \ge & \sum_{j,C:j\,\mathrm{is\,easy}}p(\R^{*}\cap\R(j,C)\cap\R(I))\\
 &  & +\sum_{j,C:j\,\mathrm{is\,hard\,and}\,(C,\bar{p}_{j})\,\mathrm{is\,relevant}}p(\R(j,C)\cup\Rir(j)\cap\R_{\l}\cap\R(I))\\
 & \ge & p(\R^{*}\cap\R_{\l}\cap\R(I)).
\end{eqnarray*}
\end{proof}
In order to satisfy the properties of Lemma~\ref{lem:large-jobs},
we need to prove the third property of Definition~\ref{def:framework-split}. 
\begin{lem}
For any two paths $Q,Q'\in\Q$ with $Q\supseteq Q'$ we have that
$\R'_{\l,Q}\cap\R(Q')\subseteq\R'_{\l,Q'}$. 
\end{lem}

\begin{proof}
This follows from the definition of irrelevant pairs. If a pair $(C,\bar{p})$
with $C\in Q'$ is \emph{irrelevant }when we defined \emph{$\R'_{\l,Q'}$,
}then it is also irrelevant when we defined $\R'_{\l,Q}$. Therefore,
when a pair $(C,\bar{p})$ is relevant when we defined $\R'_{\l,Q}$,
it is also relevant when we defined $\R'_{\l,Q'}$, and hence $\R'_{\l,Q}\cap\R(Q')\subseteq\R'_{\l,Q'}$. 
\end{proof}

\subsubsection{Definition of the families $\left\{ \chi_{\protect\l,Q}\right\} _{Q\in\protect\Q}$ }

Let $Q\in\Q$; we want to define the set $\chi_{\l,Q}$ such that
$\R'_{\l,Q}\in\chi_{\l,Q}$. We will argue that for $\R'_{\l,Q}$
there are at most $(nP)^{O_{\epsilon}(1)}$ options and we define
$\chi_{\l,Q}$ to be the family of sets that contains each of these
possible options. From the definition of $\R'_{\l,Q}$ it follows
that $\R'_{\l,Q}$ is completely defined once we know for each cell
$C\in Q$ and for each group $g$ 
\begin{itemize}
\item the number $k(C,g)=k\le1/\delta$ of jobs of group $g$ for which
at least one rectangle is contained in $\R^{*}\cap\Rl\cap\R(C)$;
denote by $j_{1},\dotsc,j_{k}$ these jobs, 
\item the budget $B_{\mathrm{large},g,k'}^{\mathrm{round}}(C)$ for each
$k'\in\{1,\dotsc,k\}$, 
\item the value $s'(j)\le K=O_{\epsilon}(1)$ for each job $j\in\{j_{1},\dotsc,j_{k}\}$,
\item the set $\Je(C)$,
\item for each job $j\in\Je(C)$ the number of rectangles from $\R(j,C)$
that are contained in $\R'_{\l,Q}$, 
\item the value $\bar{p}_{j_{k'}}$ for each $k'\in\{1,\dotsc,k\}$ such that
$(C,\bar{p}_{j_{k'}})$ is relevant, 
\item the order of the jobs $j_{1},\dotsc,j_{k}$ according to $\prec$ ,
\item for each job $j_{k'}$ whether $j^{(1)}=j_{k'}$ for the job $j^{(1)}$
that is defined in the iteration of $j_{k'}$ in the construction
of $\R'_{\l}$.
\end{itemize}

Consider a cell $C\in Q$ and a group $g$. \aw{We first guess $B^{\mathrm{round}}(C)$ using Lemma~\ref{lem:guessA}.}
Then we can guess $k(C,g)$
in time $1/\delta$. Then, we can guess the budgets $B_{\mathrm{large},g,k'}^{\mathrm{round}}(C)$
in time $O_{\epsilon,\delta}(1)$ since there are only $1/\delta$
values to guess, we know $B^{\mathrm{round}}(C)$, and each of the
budgets $B_{\mathrm{large},g,k'}^{\mathrm{round}}(C)$ is an integral
multiple of $\epsilon\delta\cdot B_{\mathrm{large}}^{\mathrm{round}}(C)$.
Also, we can guess $\left\{ s'(j_{k'})\right\} _{k'}$ in time $K/\delta=O_{\delta,\epsilon}(1)$.
This yields the set $J_{\mathrm{high},g}(C)$. Since $\Je(C)\subseteq J_{\mathrm{high},g}(C)$
and $\left|J_{\mathrm{high},g}(C)\right|\le2k(C,g)\le2/\delta$, we
can guess $\Je(C)$ in time $2^{2/\delta}$. Also, we can guess in
time $O_{\epsilon,\delta}(1)$ the rectangles of each job $j\in\Je(C)$
that are contained in $\R'_{\l,Q}$. Also, in time $(1/\delta)!=O_{\delta}(1)$
we can guess the ordering of the jobs $j_{1},\dotsc,j_{k}$ according
to~$\prec$. For each job $j_{k'}\in\left\{ j_{1},\dotsc,j_{k}\right\} $
there are only two options for whether $j^{(1)}=j_{k'}$ in the iteration
corresponding to $j_{k'}$, and thus we can guess this for all these
$k$ jobs in time $2^{k}\le2^{1/\delta}=O_{\delta}(1)$. Since there
are $O_{\epsilon}(\log nP)$ cells $C\in Q$ and $O_{\epsilon}(1)$
groups $g$, this yields $O_{\epsilon,\delta}(1)^{O_{\epsilon,\delta}(\log(nP))}=(nP)^{O_{\epsilon,\delta}(1)}$
possible guesses overall.

It remains to argue that we can guess also the values $\bar{p}_{j_{k'}}$
in time $(nP)^{O_{\epsilon,\delta}(1)}$. The intuition is that there
are only $O_{\epsilon,\delta}(\log nP)$ relevant pairs $(C,\bar{p})$
and they admit a certain structure that allows us to guess them in
time $(nP)^{O_{\epsilon,\delta}(1)}$. 
\begin{lem}
In time $(nP)^{O_{\epsilon,\delta}(1)}$ we can guess all relevant
pairs $(C,\bar{p})$ with $C\in Q$. 
\end{lem}

\begin{proof}
As defined previously, let $v_{\tilde{C}}$ be the bottom-most vertex
of $Q$ and let $v_{\tilde{C}'}$ be its parent vertex. There are
$O_{\epsilon}(\log nP)$ pairs of the form $(\tilde{C},\bar{p})$
or $(\tilde{C}',\bar{p})$ and we can guess in time $2^{O_{\epsilon}(\log nP)}=(nP)^{O_{\epsilon}(1)}$
which of them are relevant.

Let us consider the cells $C\in Q$ with $\tilde{C}\ne C\ne\tilde{C}'$,
let $Q'\subseteq Q$ denote the set of all these cells. We group the
relevant pairs $(C,\bar{p})$ with $C\in Q'$ into groups $\G_{\ell}$
where for each $\ell\in\N$ we define 
\[
\G_{\ell}:=\{(C,\bar{p})|C\in Q'\wedge\bar{p}=(1+\epsilon)^{\ell+\dist(v_{\tilde{C}},v_{C})\cdot\left\lfloor \log_{1+\epsilon}\left(\epsilon\delta\right)\right\rfloor }\}.
\]
Now each set $\G_{\ell}$ can contain at most one relevant pair $(C,\bar{p})$:
assume by contradiction that $\G_{\ell}$ contains two relevant pairs
$(C,\bar{p}),(C',\bar{p}')$. Asssume w.l.o.g.~that $v_{C}$ is closer
to $v_{\tilde{C}}$ than $v_{C'}$. Then there is a hard job $j$
with $\bar{p}_{j}=\bar{p}$ such that $\R(j,C)\cap\R^{*}\ne\emptyset$.
In particular,

\[
\bar{p}=(1+\epsilon)^{\ell+\dist(v_{\tilde{C}},v_{C})\cdot\left\lfloor \log_{1+\epsilon}\left(\epsilon\delta\right)\right\rfloor }
\]
and

\[
\bar{p}'=(1+\epsilon)^{\ell+\dist(v_{\tilde{C}},v_{C'})\cdot\left\lfloor \log_{1+\epsilon}\left(\epsilon\delta\right)\right\rfloor }
\]
which implies that 
\[
\bar{p}'=\bar{p}_{j}\cdot(1+\epsilon)^{\dist(v_{C},v_{C'})\cdot\left\lfloor \log_{1+\epsilon}\left(\epsilon\delta\right)\right\rfloor }\le\bar{p}_{j}\cdot\left(\epsilon\delta\right)^{\dist(v_{C},v_{C'})}\le\bar{p}_{j}\cdot\delta\cdot\epsilon^{\dist(v_{C},v_{C'})}
\]
and hence $(C',\bar{p}')$ is irrelevant.

Also, observe that there are only $O_{\epsilon}(\log nP)$ values
$\ell$ such that $\G_{\ell}$ contains a relevant pair. It remains
to show that we can guess these relevant pairs efficiently. First,
we guess in time $2^{O_{\epsilon}(\log nP)}=(nP)^{O_{\epsilon}(1)}$
for which cells $C$ there exists a relevant pair $(C,\bar{p})$ for
some value $\bar{p}$.

Let $C^{*}$ denote the topmost cell in $Q$. We order the relevant
pairs $(C,\bar{p})$ non-increasingly according to $\dist(v_{C},v_{C^{*}})$,
breaking ties by ordering them increasingly by their values $\bar{p}$.
For each pair $(C,(1+\epsilon)^{\ell})$ we introduce the value $\alpha_{(C,(1+\epsilon)^{\ell})}:=\ell+\dist(v_{C^{*}},v_{C})\cdot\left\lfloor \log_{1+\epsilon}\left(\epsilon\delta\right)\right\rfloor $.
We claim that in our ordering of the relevant pairs $(C,(1+\epsilon)^{\ell})$
the values $\alpha_{(C,(1+\epsilon)^{\ell})}$ are strictly increasing.
Indeed, consider two pairs $(C,\bar{p}),(C',\bar{p}')$ that are adjacent
in this ordering such that $(C,\bar{p})$ appears directly before
$(C',\bar{p}')$. If $C=C'$ then $\bar{p}<\bar{p}'$ and hence $\alpha_{(C,\bar{p})}<\alpha_{(C',\bar{p}')}$.
Suppose now that $C\ne C'$. Assume by contradiction that $\alpha_{(C,\bar{p})}\ge\alpha_{(C',\bar{p}')}.$
Assume that $\bar{p}=(1+\epsilon)^{\ell}$ and $\bar{p}'=(1+\epsilon)^{\ell'}$
and let $j$ and $j'$ be the hard jobs corresponding to the pairs
$(C,\bar{p})$ and $(C',\bar{p}')$, respectively. Then 
\[
\alpha_{(C,\bar{p})}=\ell+\dist(v_{C^{*}},v_{C})\cdot\left\lfloor \log_{1+\epsilon}\left(\epsilon\delta\right)\right\rfloor \ge\ell'+\dist(v_{C^{*}},v_{C'})\cdot\left\lfloor \log_{1+\epsilon}\left(\epsilon\delta\right)\right\rfloor =\alpha_{(C',\bar{p}')}
\]
which implies that 
\[
\ell+\dist(v_{C},v_{C'})\cdot\left\lfloor \log_{1+\epsilon}\left(\epsilon\delta\right)\right\rfloor \ge\ell'
\]
and therefore 
\[
\bar{p}_{j}\cdot\delta\cdot\epsilon{}^{\dist(v_{C},v_{C'})}\ge\bar{p}_{j}\cdot\left(\epsilon\delta\right)^{\dist(v_{C},v_{C'})}\ge(1+\epsilon)^{\ell+\dist(v_{C},v_{C'})\cdot\left\lfloor \log_{1+\epsilon}\left(\epsilon\delta\right)\right\rfloor }\ge(1+\epsilon)^{\ell'}=\bar{p}_{j'}.
\]
This implies that the job $j$ makes the pair $(C',\bar{p}')$ irrelevant
which is a contradiction.

Since the $\alpha_{(C,\bar{p})}$ are increasing and can attain only
$O_{\epsilon,\delta}(\log(nP))$ different values, we can guess in
time $2^{O_{\epsilon,\delta}(\log nP)}=(nP)^{O_{\epsilon,\delta}(1)}$
which of these possible values are attained by some $\alpha_{(C,\bar{p})}$
(however, this does not tell us the corresponding pairs $(C,\bar{p})$
since a value $\alpha_{(C,\bar{p})}$ might belong to more than one
pair $(C,\bar{p})$). Then, since $|Q'|=O_{\epsilon,\delta}(\log(nP))$
and the values $\alpha_{(C,\bar{p})}$ are ordered according to their
cells, we guess in time $2^{O_{\epsilon,\delta}(\log nP)}=(nP)^{O_{\epsilon,\delta}(1)}$
which of these attained values corresponds to which cell $C$. One
way to do this is to guess a bit-string with $O_{\epsilon,\delta}(\log(nP))$
bits, which describes in unary the number of relevant pairs for each
cell $C\in Q'$, with the 0-bits being the separators between these
values for the different cells $C\in Q'$. Once we know each value
$\alpha_{(C,\bar{p})}$ and its corresponding cell $C$, we can deduce
the corresponding pair $(C,\bar{p})$ and hence we know all relevant
pairs. 

Overall, there are $(nP)^{O_{\epsilon,\delta}(1)}$ possible guesses
in total.
\end{proof}
Once we know all relevant pairs, we can guess in time $O_{\epsilon,\delta}(1)$
per relevant pair $(C,\bar{p})$ for which group $g$ there is a job
$j_{k'}$ with $\bar{p}_{j_{k'}}=\bar{p}$ and the corresponding value
$k'$, which yields $O_{\epsilon,\delta}(1)^{O_{\epsilon,\delta}(\log(nP))}=(nP)^{O_{\epsilon,\delta}(1)}$
possible guesses overall. 

We define that $\chi_{\l,Q}$ contains the resulting set of large
tasks for each of the possible guesses for the above values. This
completes the proof of Lemma~\ref{lem:large-jobs}.

\bibliographystyle{alpha}
\bibliography{bibliography}

\appendix

\section{Omitted proofs}

\subsection{Simplification of input instance\label{subsec:Bounded-weights}}
First we establish that $P=\max_j p_j/\min_j p_j=\max_j p_j$:
We scale all values $p_j$ and $r_j$ by the same factor so that $\min_j p_j = n^2 \cdot 2/\epsilon$.
This preserves the approximation rate of a solution (if it is scaled accordingly), but 
the values $p_j$ and $r_j$ are no longer integers. Hence, we round all these values
to the next integer. A solution for the non-rounded values can be transformed to
a solution for the rounded values by delaying each jobs completion time by at most
$n^2 + 1$ ($n^2$ for rounding $p_j$ and $1$ for rounding $r_j$).
Since the optimum is at least $\sum_j w_j p_j \ge \sum_j w_j n^2 \cdot 2/\epsilon$,
this increases the optimum by at most a factor of $(1 + \epsilon)$.
Then we create a dummy job with processing time $1$ and negligible weight to ensure that $\min_j p_j=1$
and hence $P=\max_j p_j$. This transformation increases $P$ only by a polynomial factor.

To obtain bounded weights we scale each
$w_{j}$ by the same factor such that $\max_{j}w_{j}=4/\epsilon^{2}\cdot n^{2}P$.
This transformation preserves the approximation rate of a solution.
Now remove all jobs $j$ with $w_{j}<1/\epsilon$, round each remaining
$w_{j}$ to the next integer (increasing the optimum by a factor at
most $(1+\epsilon)$), and solve the remaining instance. Note that
the optimum is at least $\max_{j}w_{j}$ and all jobs are finished
before $T$. We now schedule all jobs that were previously arbitrarily
in the interval $[T,2T]$. The cost of these jobs is at most $n\cdot1/\epsilon\cdot2T\le1/\epsilon\cdot4n^{2}P\le\epsilon\max_{j}w_{j}$
and thus negligible.

\end{document}